\documentclass{amsart}
\usepackage[utf8]{inputenc}
\usepackage{amsmath,amsfonts,amssymb,amsthm,amscd,amsbsy,epsfig,array,mathrsfs,tikz}
\usetikzlibrary{positioning}
\usetikzlibrary{decorations.pathmorphing}
\usetikzlibrary{decorations.markings}

\theoremstyle{plain}
\newtheorem{thm}{Theorem}[section]
\newtheorem*{thm*}{\bf Theorem }
\newtheorem{lem}[thm]{Lemma}
\newtheorem{prop}[thm]{Proposition}
\newtheorem*{prop*}{\bf Proposition}
\newtheorem{cor}[thm]{Corollary}

\theoremstyle{definition}
\newtheorem{defn}[thm]{Definition}

\newtheorem*{example}{Example}

\theoremstyle{remark}
\newtheorem*{rem}{Remark}

\newcommand{\cupdot}{\mathbin{\dot{\cup}}}
\newcommand{\set}[1]{ \{1,\ldots,   {#1} \} }
\newcommand{\m}[1]{\mathcal{#1}}
\newcommand{\mb}[1]{\mathbb {#1}}

\newcommand{\ti}[1]{\tilde {#1}}
\newcommand{\ds}{/\!/}
\newcommand{\un}[1]{\underline{#1}}
 
\author{Marko Berghoff}
\title{Wonderful Compacifications in Quantum Field Theory}

\begin{document}
\begin{abstract}
In \cite{bbk} it was shown how so-called wonderful compactifications can be used for renormalization in the position space formulation of quantum field theory. This article aims to continue this idea, using a slightly different approach; instead of the subspaces in the arrangement of divergent loci, we use the poset of divergent subgraphs as the main tool to describe the whole renormalization process. This is based on \cite{fe} where wonderful models were studied from a purely combinatorial viewpoint. The main motivation for this approach is the fact that both, renormalization and the model construction, are governed by the combinatorics of this poset. Not only simplifies this the exposition considerably, but also allows to study the renormalization operators in more detail. Moreover, we explore the renormalization group in this setting by studying how the renormalized distributions behave under a change of renormalization points. 
\end{abstract}
\maketitle

\section{Introduction}

Quantum Field Theory (QFT), the unification of quantum mechanics and special relativity, is the last century's most successful physical theory. Although plagued with infinities and ill-defined quantities all over the place, perturbative calculations are in astonishing agreement with data obtained from particle physics experiments.
The art of taming those mathematical monsters, i.e.\ extracting physical sensible data from a priori divergent expressions, is called renormalization. Over the years it has turned from a ``black magic cooking recipe'' into a well-established and rigorous formulated theory, at the latest since the 90's when Kreimer discovered a Hopf algebra structure underlying renormalization. The main implication is that (perturbative) QFT is governed by the combinatorics of Feynman diagrams. This has proven to be a very powerful tool, both in computational problems as 
well as in improving our understanding of QFT in general. In addition, it has revealed surprising connections to deep questions in pure mathematics, for example in number theory and algebraic geometry \cite{dk}.

The mathematical reason for divergences arising in perturbative calculations is that quantum fields are modeled by operator-valued distributions for which products are in general not well-defined. In the position space formulation of QFT renormalization translates directly into the problem of extending distributions as shown by Epstein and Glaser in \cite{eg}. Although they formulated and solved the renormalization problem already in the early 70's, since then no real progress has been made in this direction. This is mainly due to two reasons: Firstly, their approach was mathematical precise but conceptually difficult. It involved a lot of functional analysis, in some sense disguising the beauty and simplicity of the idea. Secondly, it is not applicable to calculations at all. Only recently, in a first approximation to quantum gravity, physicists have started to study quantum fields on general spacetimes and in this setting one is naturally forced to work in position space \cite{bf}. 

In \cite{bbk} another, more geometric approach to the renormalization problem was presented. In position space, Feynman rules associate to a graph $G$ a pair $(X^G,v_G)$ where $X^G$ is a product of the underlying spacetime and $v_G:X^G \to \mb R$ a rational function. One would like to evaluate this to
\begin{equation*}
 \big(X^G,v_G\big)=\int_{X^G} v_G,
\end{equation*}
but this fails in general as the integrand need not be an element of $L^1(X^G)$. If $v_G$ does not vanish fast enough at infinity this is called an infrared divergence. The problem is circumvented by viewing $v_G$ as a distribution on the space of compactly supported test functions. On the other hand, ultraviolet divergences arise from $v_G$ having poles along certain subspaces of $X^G$. These subspaces are determined by $\m D$, the set of (ultraviolet-)divergent subgraphs of $G$, and form the \textit{divergent arrangement} $X_{\m D}^G$. In this setting renormalization translates into the problem of finding an extension of $v_G$ onto $X^G_{\m D}$. In \cite{bbk} this is solved with a geometric ansatz: The idea is to resolve the divergent arrangement into a normal crossing divisor and then define canonical renormalization operators that extend $v_G$ to a distribution defined on the whole space $X^G$. Such a model, also called a compactification of the complement of $X_{\m D}^G$, is provided by the \textit{
wonderful model} construction by DeConcini and Procesi \cite{dp}, based on techniques from Fulton and MacPherson's seminal paper \cite{fm}. What makes it so well-suited for renormalization is the fact that the whole construction is governed by the combinatorics of the arrangement which translates directly into the subgraph structure of $G$. 

The idea of employing a resolution of singularities to extend distributions is not new. It is based on a paper by Atiyah \cite{at} that highlighted the usefulness of Hironaka's famous theorem in seemingly unrelated areas of mathematics. In addition, the same technique was applied in Chern-Simons perturbation theory independently by Kontsevich \cite{ko} as well as Axelrod and Singer \cite{as}. For an application of this idea to renormalization in parametric space see~\cite{bek}. 

The present article aims to continue the work of \cite{bbk} emphasizing a slightly different point of view. We use another language to formulate the wonderful construction and the renormalization process; instead of the subspaces in the divergent arrangement, we express the central notions in terms of the poset $\m D$, formed by all divergent subgraphs of $G$, partially ordered by inclusion. This is inspired by \cite{fe} where the wonderful model construction is studied from a combinatorial point of view. Not only does this simplify the definitions and proofs immensely, it also highlights the combinatorial flavour in the construction of both, the wonderful models and the renormalization operators. In addition, instead of the vertex set of $G$ we use \textit{adapted} spanning trees $t$ to define coordinates on $X^G$, naturally suited to the problem. This is also mentioned in \cite{bbk}, but not used to its full extent. The main point is that such spanning trees are stable under certain graph theoretic 
operations like contraction of divergent subgraphs and therefore provide a convenient tool to formulate the wonderful construction. It allows to treat the definition of the renormalization operators in more detail (\cite{bbk} focuses mainly on the model construction for arrangements coming from graphs) and to study the renormalization group, a powerful tool (not only) in QFT, that allows even for statements beyond perturbation theory. The main result is a formula for the change of renormalization points, the parameters involved in defining the renormalization operators. It relates a so-called finite renormalization of the renormalized distribution $\mathscr R[v_G]$ to a sum of distributions determined by the divergent subgraphs of $G$.

The presentation is organized as follows. The first section covers some topics of distribution theory that will be needed later; it finishes with a definition of Feynman rules, i.e.\ how QFT associates distributions to Feynman diagrams, and an analysis of the divergent loci of these distributions.
The next section is devoted to the two other central objects in this article, smooth models and posets. It starts with an exposition of the wonderful model construction as in \cite{dp}, then, following \cite{fe}, we introduce the necessary combinatorial language to review this construction from a purely combinatorial viewpoint; special emphasis is given to the case of arrangements coming from graphs via Feynman rules. 
After these mostly preliminary steps we come to the main part, the \textit{wonderful renormalization} process. We first study the pole structure of the pullback of a Feynman distribution onto an associated wonderful model and then define two renormalization operators. This definition requires some choices to be made and a natural question, considered in Section \ref{renormalization group}, is to ask what happens if one varies these parameters. We derive and proof a formula for these so-called finite renormalizations. 
The last section finishes the wonderful renormalization process by showing that it is physical reasonable, i.e.\ it satisfies the Epstein-Glaser recursion principle, in other contexts known as locality of counterterms. After that we discuss the connection between the renormalization operation for single graphs presented here and the Epstein-Glaser method. We finish with an outlook to further studies: The treatment of amplitudes and the role of Fulton-MacPherson compactifications in this setting, and the Hopf algebraic formulation of wonderful renormalization. 

\section{Distributions} \label{distributions}

In this section we collect some preliminary material about distributions. Although crucial in the definition of QFT (see for example \cite{ws}), in most textbooks the distributional character of the theory is largely neglected. This is fine as long as one works in momentum space, but in position space they play a central role in every aspect. We start by defining distributions on manifolds. Then we state the extension problem and study its solution in a toy model case. The section finishes with a definition of Feynman distributions, i.e.\ distributions associated to graphs via Feynman rules, and an analysis of the corresponding divergent loci. 
\newline

Let $X \subseteq \mb R^d$ be open and denote by $\m D(X):= C^{\infty}_0(X)$ the space of compactly supported smooth functions on $X$. We write $\m D'(X)$ for the space of continuous linear functionals on $\m D(X)$ and $\langle u | \varphi \rangle$ for the value of $u \in \m D'(X)$ at $\varphi \in \m D(X)$. By (the usual) abuse of notation we use the same symbol $f$ for a function and the functional $u_f$ it represents. In the latter case we refer to $f$ as the kernel of $u_f$. The locus where $u$ cannot be given by a function is called the \textit{singular support} of $u$.

For $X \subseteq \mathbb{R}^d$ open it seems natural to define the space of distributions as above, $\mathcal D'(X)\!:=\!(\mathcal D(X))^*$. To generalize this to the manifold case there are two possibilities, depending on whether distributions should generalize functions or measures (cf.\ \cite{ho}). In the following let $\{ \psi_i : U_i \to \ti U_i \subseteq \mathbb{R}^d\}_{i\in I}$ be an atlas for a smooth manifold $X$. 
\begin{defn}
A \textit{distribution} $u$ on $X$ is given by a collection of distributions $\{ u_i \in \mathcal D'(\ti U_i) \}_{i \in I}$ such that for all $i,j \in I$
\begin{equation*}
 u_j=(\psi_i\circ\psi_j^{-1})^*u_i  \quad \text{in}\quad \psi_j(U_i\cap U_j).
\end{equation*}
The space of distributions on $X$ is denoted by $\mathcal D' (X)$.
\end{defn}

This is the way to define distributions as generalized functions on $X$ (every $u \in  C^0 (X)$ defines a distribution by setting $u_i:=u\circ \psi_i^{-1}$). If we start from the point of view that they are continuous linear forms on $\m D(X)$, we arrive at generalized measures on $X$, also called \textit{distribution densities}:

\begin{defn}
 A \textit{distribution density} $\tilde u$ on $X$ is a collection of distributions $\{ \tilde u_i \in \mathcal D'(\ti U_i) \}_{i\in I}$ such that for all $i,j \in I$
\begin{equation*}
 \tilde u_j=\big|\text{det} \ D(\psi_i\circ\psi_j^{-1})\big|(\psi_i\circ\psi_j^{-1})^*\tilde u_i \quad \text{in} \quad \psi_j(U_i\cap U_j).
\end{equation*}
The space of distribution densities on $X$ is denoted by $\mathcal{\tilde D}' (X)$.
\end{defn}

Because of their transformation properties, distribution densities are also called \textit{pseudoforms}. They generalize differential forms in the sense that they can be integrated even on non-orientable manifolds. For more on pseudoforms and integration on non-orientable manifolds we refer to \cite{ni}. Note that if $X$ is orientable, there is an isomorphism $\mathcal{D}'  (X) \cong \mathcal{\tilde D}' (X)$ via $u\mapsto u \nu $ for $\nu$ a strictly positive density (i.e.\ a volume form) on $X$. In particular, on $\mathbb{R}^d$ such a density is given by the Lebesgue measure $\nu=|dx|$ and we write $\tilde u$ for $u|dx|$ with $u \in \m D'(\mb R^d)$.

For later purposes we introduce two operations on distributions and densities, the pullback and pushforward along a smooth map $f:X\to X'$. 

\begin{defn}[Pushforward]
Let $X\subseteq \mathbb{R}^m$ and $X'\subseteq \mathbb{R}^n$ be open and $f: X\! \to\! X'$ be surjective and proper (if $u$ is compactly supported this requirement can be dropped). For a distribution $u$ on $X$ the pushforward $f_*u\in \mathcal D'(X')$ is defined by 
\begin{equation*}
\langle f_*u | \varphi \rangle = \langle u| f^*\varphi \rangle \quad \text{for all} \quad \varphi \in \mathcal D(X').  
\end{equation*}
For $X$ and $X'$ manifolds with atlantes $(\psi_i,U_i)_{i\in I}$ and $(\psi'_j,U'_j)_{j\in J}$ we define the pushforward $f_*\ti u \in \mathcal{\tilde D}' (X')$ of $\ti u\in \mathcal{\tilde D}' (X)$ by 
\begin{equation*}
 (f_*u)_j:= (\psi'_j \circ f \circ \psi_i^{-1})_*u_i \quad \text{in} \quad U'_j \cap (\psi'_j \circ f \circ \psi_i^{-1})(U_i).
\end{equation*}
\end{defn}
The question under what conditions the pullback of distributions is defined is more delicate, see \cite{ho} for a detailed exposition. We state only one special case where it is possible to define a pullback: Let $X, X'$ be open subsets of $\mathbb{R}^n$ and $f:X\to X'$ a smooth submersion. Then there exists a unique linear operator $f^*: \mathcal D'(X') \to \mathcal D'(X)$ such that $f^*u=u\circ f$ if $u \in C^0(X')$. If $X, X'$ are manifolds and $\ti u$ is a density on $X'$, then $f^*\ti u\in \mathcal{ \ti  D}'(X)$ is defined by
\begin{equation*}
 (f^*u)_i:= (\psi'_j \circ f \circ \psi_i^{-1})^*u_j \quad \text{in} \quad U'_j \cap (\psi'_j f \psi_i^{-1})(U_i).
\end{equation*}

\subsection{Extension of distributions}\label{ext}

We present the theory of extending distributions by studying a toy model, distributions on $\mathbb{R}\setminus \{0\}$ given by kernels having an algebraic singularity at 0, following the exposition in \cite{gs}. Applying this toy model to the extension problem for Feynman distributions is precisely the idea behind wonderful renormalization.  

\begin{defn}[Extension problem]\label{extension prob}
Let $X$ be a smooth manifold and $Y\subseteq X$ an immersed submanifold. Given a density $\ti u \in \m{\ti D}'(X \setminus Y)$ find an extension of $\ti u$ onto $X$, i.e.\ find a density $\ti u_{\text{ext} } \in \m{ \ti D}'(X)$ with 
\begin{equation*}
 \langle \ti u_{\text{ext} } | \varphi \rangle = \langle \ti u | \varphi \rangle \quad\text{for all}\quad \varphi \in \mathcal D(X \setminus Y).
\end{equation*}
\end{defn}
In this very general formulation the problem is not always solvable. Moreover, if there is a solution, it need not be unique since by definition two extension may differ by a distribution supported on $Y$. Therefore, additional conditions are sometimes formulated to confine the space of solutions. Usually one demands that the extension should have the same properties as $u$, for example scaling behaviour, Poincare covariance or solving certain differential equations.

The toy model: Let $u \in \m D'(\mathbb{R}\setminus \{0\})$ be defined by the kernel $x \mapsto |x|^{-1}$. A priori $u$ is only defined as a distribution on the space of test functions vanishing at $0$. 
The first step in the process of extending $u$ is to regularize it by introducing a complex parameter $s\in \mathbb{C}$. Raising $u$ to a complex power $u^s$ is known as \textit{analytic regularization}. It justifies the following calculations:
\begin{align*}
\langle u^s | \varphi \rangle & = \int_{\mathbb{R}} dx \frac{1}{|x|^s}\varphi(x) \\
& = \int_{[-1,1]} dx \frac{\varphi(x)-\varphi(0)}{|x|^s}  + \varphi(0) \int_{[-1,1]} dx \frac{1}{|x|^s}   + \int_{\mathbb{R}\setminus [-1,1]} dx \frac{\varphi(x)}{|x|^s}\\
& = \int_{[-1,1]} dx \frac{\varphi(x)-\varphi(0)}{|x|^s}  + \frac{2\varphi(0)}{1-s}  + \int_{\mathbb{R}\setminus [-1,1]} dx \frac{\varphi(x)}{|x|^s}
\end{align*}
where the last term is defined for all $s\in \mathbb{C}$, the second term for $s\neq1$ and the first one for $\text{Re}(s)<3$. 
We have thus found a way to split the regularized distribution $u^s=u_{\infty}(s,\cdot) + u_{\heartsuit}(s,\cdot) $ into a divergent and a convergent part. The divergent part is the principal part of the Laurent expansion of the meromorphic distribution-valued function $s \mapsto u^s$ in a punctured disc around 1 in $\mathbb{C}$: 
\begin{align*}
 \langle u_{\infty}(s,\cdot) | \varphi \rangle & = \frac{2 \varphi (0)}{1-s}, \\
 \langle u_{\heartsuit}(s,\cdot) | \varphi \rangle & = \int_{\mathbb{R}}dx \frac{\varphi(x)- \theta(1-|x|)\varphi(0)}{|x|^s}.
\end{align*}
To continue the process of extending $u$ we have to get rid of the divergent part in some sensible way (in physics this is the choice of a \textit{renormalization scheme}) and take the limit $s\to 1$. The most straightforward way to do so is by subtracting the pole (\textit{minimal subtraction}) and set 
\begin{equation*}
 \hat u = r_1[u^s]|_{s=1} := \left(u^s - \frac{2\delta}{1-s}\right)\bigg|_{s=1} = u_{\heartsuit}(1,\cdot).
\end{equation*}  
Obviously this technique can be generalized to extend distributions $u$ with higher negative powers of $|x|$ - one simply subtracts a higher order Taylor polynomial from $\varphi$.

Another renormalization scheme, \textit{subtraction at fixed conditions}, is given by 
\begin{equation*}
\langle r_{\nu} [u^s] | \varphi \rangle := \langle u^s | \varphi \rangle - \langle u^s | \varphi(0) \nu \rangle 
\end{equation*}
where $\nu \in \mathcal{D}(\mathbb{R})$ is a smooth cutoff function with $\nu(0)=1$. Another way to formulate the subtracted distribution is 
\begin{equation*}
 \langle u^s | \varphi(0) \nu \rangle = \langle (p_0)_* (\nu u^s) | \delta_0[ \varphi] \rangle.
\end{equation*}
Here $p_0: \mathbb{R}\to \{0\}$ is the projection onto the divergent locus and $\delta_0$ is interpreted as an operator $\mathcal{D}(\mathbb{R}) \to \mathcal{D}(0)$ mapping test functions on $\mathbb{R}$ onto test functions supported on the divergent locus. From this it is also clear that the difference between two such renormalization operators $r_{\nu}$ and $r_{\nu'}$ is given by a distribution supported on $\{0\}$, i.e.\ a linear combination of $\delta$ and its derivatives. 
This formulation will turn out to be very useful.

A nice feature of these renormalization operators $r$ is that they commute with multiplication by smooth functions, $r[fu]=fr[u]$ for $f\in \mathcal{C}^{\infty}(\mathbb{R})$. In addition, $r$ belong to the class of \textit{Rota-Baxter operators}, a fact extensively used in the Hopf algebraic formulation of renormalization (see for example~\cite{efg}).

Later we will work with distributions given by kernels 
\begin{equation*}
 u^s(x)=\frac{1}{|x|^{1+d(s-1)}}.
\end{equation*}
In this case $u^s$ splits into  
\begin{equation}\label{distro split}
 u^s = - \frac{2}{d} \frac{\delta_0}{s-1} + u_{\heartsuit}(s)
\end{equation}
with $u_{\heartsuit}(s)$ holomorphic for Re$(s)<\frac{2+d}{d}$.
\newline

\subsection{Feynman distributions}\label{graphs to arr}

Feynman diagrams are convenient book-keeping devices for the terms in the perturbative expansion of physical quantities. The map that assigns to every Feynman diagram its corresponding analytical expression is called \textit{Feynman rules} and denoted by $\Phi$. In position space the map $\Phi$ assigns to every diagram $G$ a pair $(X^G,\ti v_G)$ where $\ti v_G$ is a differential form on the space $X^G$, a cartesian product of the underlying spacetime $M$. We would like to evaluate 
\begin{equation*}
 \big(X^G,\ti v_G\big) \longmapsto \int_{X^G} \ti v_G,
\end{equation*}
but this is in general not possible due to the problem of ultraviolet and infrared divergences. While we avoid the infrared problem by viewing $\ti v_G$ as a distribution density, the ultraviolet problem translates into an extension problem for $\ti v_G$. The ultraviolet divergences of $\ti v_G$ are assembled in a certain subspace arrangement that we will describe at the end of this section, after the definition of $\Phi$. 

We consider a massless scalar quantum field in $d$-dimensional Euclidean spacetime $M:=\mathbb R^d$. The case of fields with higher spin differs only by notational complexity. On the other hand, the massive case is much harder because already the simplest examples have special functions arising as propagators of the free theory. Working in the Euclidean metric is justified by the technique of Wick rotation (see \cite{weinberg}) that allows one to do calculations in $\mb R^d$ and transform the results back to Minkoswki spacetime. The position space propagator of a massless scalar field is given by the Fourier transform of the momentum space propagator,
\begin{equation*}
\triangle(x)= \m F \left(k\mapsto \frac{1}{k^2}\right)(x)  =\frac{1}{x^{d-2}}\, , \quad x\in M.
\end{equation*}
Let $G$ be a connected graph. By a graph we mean the following combinatorial object:
\begin{defn}
 A graph $G$ is an ordered pair $G=(V,E)$ of a finite set $V$ of vertices and a finite multiset $E$ of unordered distinct (we do not allow loops, i.e.\ edges connecting a vertex with itself) pairs of elements of $V$.
\end{defn}
\begin{example}
 The \textit{dunce's cap} graph (Figure \ref{dunce}) will serve as main example later throughout the text. Here $V=\{v_1,v_2,v_3\}$ and $E=\{ e_1=(v_1,v_2),e_2=(v_1,v_3),e_3=(v_2,v_3), e_4=(v_2,v_3) \}$.
  \begin{figure}[h]
  \centering
  \begin{tikzpicture}[node distance=1cm and 1cm]
   \coordinate[label=left:$v_1$] (v0);
   \fill[black] (v0) circle (.05cm);
   \coordinate[below right=of v0,yshift=0.3cm,label=below:$v_2$] (v1);
   \fill[black] (v1) circle (.05cm);
   \coordinate[above right=of v0,yshift=-0.3cm,label=above:$v_3$] (v2);
   \fill[black] (v2) circle (.05cm);
   \draw (v0) to (v1) node [below,pos=0.5,xshift=0.4cm,yshift=-0.3cm] {$e_1$};
   \draw (v0) to (v2) node [above,pos=0.5,xshift=0.4cm,yshift=0.3cm] {$e_2$};
   \draw (v2) to[out=-135,in=135] (v1) node [left,pos=0.5,xshift=1.2cm,yshift=0.01cm] {$e_3$};
   \draw (v2) to[out=-45,in=45] (v1) node [left,pos=0.5,xshift=1.8cm,yshift=0.01cm] {$e_4$};  
  \end{tikzpicture}
\caption{Dunce's cap}\label{dunce}
 \end{figure}
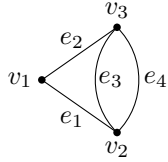
\end{example}

\begin{defn}
 A \textit{subgraph} $g$ of $G$, denoted by $g\subseteq G$, is determined by a subset $E(g)\subseteq E(G)$. 
 Usually one defines the vertex set of $g$ to be the set of vertices of $V(G)$ that are connected to edges of $g$, so that $g$ is a graph itself, $g=(V(g), E(g))$. For our purposes it is more convenient to allow also for isolated vertices. Therefore we define a subgraph $g\subseteq G$ to be an equivalence class under the relation
 \begin{equation*}
  g \sim g' \Longleftrightarrow E(g')=E(g).
 \end{equation*}
 For subgraphs $g,h\subseteq G$ we introduce the following operations:
 \begin{enumerate}
  \item Union and intersection: $g\cup h$ and $g\cap h$ are the subgraphs of $G$ defined by the corresponding operations on the edge sets of $g$ and $h$.
  \item Deletion: For $g\subseteq h$ the deletion $h\setminus g$ is the graph $h$ with all edges of $g$ removed.
  \item Contraction: For $g\subseteq h$ the contraction $h/g$ is the graph $h$ with all edges $e$ in $E(g)$ removed and for every $e\in E(g)$ the two vertices connected to $e$ identified.
 \end{enumerate}
\end{defn}

As shown in \cite{bbk}, Feynman rules are determined by the topology of $G$:
Pick a labelling $V=\{v_0,\ldots, v_n\}$ of the vertices of $G$ and an orientation on the edges of $G$. For a finite set $F$ let $\mathbb{R}^F$ denote the vector space with (fixed) basis the elements of $F$. The cohomology of the simplicial complex $G=(V,E)$ gives rise to an exact sequence,
\begin{equation*}
 0 \longrightarrow \mathbb{R} \overset{\sigma}{\longrightarrow} \mathbb{R}^V \overset{\delta}{\longrightarrow} \mathbb{R}^{E} \longrightarrow H^1(G,\mathbb{R}) \longrightarrow 0.
\end{equation*}
Here the map $\sigma$ sends $1$ to $v_0 + \cdots + v_n$ and $\delta$ is given by $\delta(v)= \sum_{e\in E}(v:e)e$ with $(v:e)=\pm 1$ if $e$ starts/ends at $v$ and $0$ otherwise. Fix a basis of coker($\sigma)$ by an isomorphism $\varphi: V':=V\setminus \{v_0\} \rightarrow \text{coker}(\sigma)$. This defines an inclusion $ \iota=\delta \circ \varphi: \mb R^{V'} \cong \text{coker}(\sigma) \hookrightarrow \mb R^E$. Doing this component-wise on the space $X^G:=M^{V'}=(\mathbb{R}^d)^{V'}$ we obtain an inclusion $I:=\iota^{\oplus d}: X^G \hookrightarrow M^{E}$ and define a rational function $v_G := I^*( \triangle^{\otimes E(G)} ): X^G \to \mb R$ by
\begin{equation*}
v_G: \sum_{v \in V'} x_v v \longmapsto \prod_{e \in E} \triangle\left( \sum_{v \in V'} (v:e)x_v \right).
\end{equation*}

Every edge $e\in E$ defines a linear form $\omega_e:=e^*\circ \iota$ on $\mb R^{V'}$ and a linear subspace of $(X^G)^*$ by
\begin{equation*}
 A_e:=\langle\omega_e\rangle^{\oplus d} =  \left\{ (x_1, \ldots,x_n) \mapsto \sum_{i=1}^d\alpha_i \omega_e ( x^i_1, \ldots , x_n^i) , \ \alpha_i \in \mathbb{R} \right\}.
\end{equation*}
For a subgraph $g\subseteq G$ we define $A_g:=\sum_{e \in E(g)} A_e $. Every family $\mathcal{P}$ of subgraphs of $G$ gives then rise to a \textit{subspace arrangement} in $(X^G)^*$, 
\begin{equation*}
 \m A_{\mathcal{P}}:= \{ A_g \mid g \in \mathcal P \}.
\end{equation*}
Note that two subgraphs $g,h \subseteq G$ may define the same subspace, $A_g=A_h$. Therefore, we consider only subfamilies of $\m G(G)$, the set of \textit{saturated} subgraphs of $G$. Saturated subgraphs are maximal with respect to the property of defining their corresponding subspaces. A precise definition is given in Section \ref{graph poset}. 
Two arrangements are especially important for our purposes, the \textit{singular arrangement} 
\begin{equation*}
 \mathcal{A}_{\m G(G)} := \{ A_g \mid g\subseteq G \text{ is saturated} \},
\end{equation*}
and the arrangement coming from the family $\m D(G)$ of \textit{divergent} subgraphs of~$G$,
\begin{equation*}
 \mathcal A_{\mathcal{D}(G)} = \{ A_g \mid g \subseteq G \text{ is divergent}\}.
\end{equation*}

\begin{defn}\label{sdd}
 Let $h_1(\cdot)$ denote the first Betti number. Define the \textit{superficial degree of divergence} $\omega$ of $G$ by 
 \begin{equation*}
  \omega(G):=dh_1(G)-2 |E|.
 \end{equation*}
Then $G$ is called \textit{divergent} if $ \omega(G)\geq 0$. $G$ is \textit{at most logarithmic} if $ \omega(g)\leq 0$ holds for all $g\subseteq G$. If $\mathcal D(G)=\{\emptyset,G\}$, then $G$ is called \textit{primitive}.
\end{defn}

\begin{lem}
 Let $G$ be at most logarithmic and $d>2$. Then 
 \begin{equation*}
  g\in \m D(G) \Longrightarrow g \text{ is saturated}.
 \end{equation*}
\end{lem}

\begin{proof}
 Suppose $A_g=A_{g\cup e} \subseteq (M^{V'})^*$ for some $e\in E(G\setminus g)$. Thus, $g\cup e$ connects the same set of vertices as $g$ and adding $e$ to $g$ must produce a new, independent cycle, $h_1(g\cup e)=h_1(g)+1$. Therefore, $\omega(g\cup e)= \omega(g) + d-2 >0$, a contradiction to $G$ at most logarithmic.
\end{proof}
As shown in the next proposition, the divergent arrangement $\m A_{\m D(G)} $ describes exactly the locus where extension is necessary. 
\begin{prop} \label{sing supp}
 Let $G=(V,E)$ be connected and at most logarithmic. Set 
 \begin{equation*}
 X^G_{\text{s}}:=\bigcup_{e\in E} A^{\bot}_e  \quad\text{and}\quad  X^G_{\mathcal D}:=\bigcup_{g\in \mathcal D(G)} A^{\bot}_g .
\end{equation*}
Then $v_G$ is a well-defined distribution on $X^G \setminus X^G_{\mathcal D}$ and the singular support of $v_G$ is $X^G_{\text{s}} \setminus X^G_{\mathcal D}$.
 \end{prop}
 
\begin{proof}
Let $V=\{v_0, \ldots, v_n \}$, $V'=V\setminus \{v_0\}$. Wherever defined, we have
 \begin{equation*}
  v_G(x_1, \ldots ,x_n)= I^*\left(\triangle^{\otimes |E|}\right)(x_1,\ldots,x_n) =\prod_{e\in E} \triangle\left(\sum_{i=1}^n(v_i:e)x_i\right).
 \end{equation*} 
 Since $\text{sing supp}(\triangle)=\{0\}$, the singular support of $\triangle^{\otimes E}$ is the set where at least one $x_e \in M$ vanishes. But this is precisely the image of $A_e^{\bot}$ under $I$. Thus, sing supp$(v_G)\subseteq X^G_s$.
 
 For $K\subseteq X^G$ compact and $\chi_K$ the (smooth approximation of the) characteristic function of $K$ we need to show that $\langle v_G | \chi_K \rangle=\int_K dx \ v_G(x) < \infty $ as long as $K$ is disjoint from $X^G_{\mathcal D}$. Assume the contrary, $K \cap X^G_{\mathcal D} \neq \emptyset$; more precisely, $K$ intersects $A^{\bot}_g$ for some $g\in \mathcal D(G)$, but no other divergent loci. Moreover, assume that $g$ is connected --- otherwise $A^{\bot}_g=A^{\bot}_{g_1}\cup A^{\bot}_{g_2}$ and a smaller $K$ will intersect only one of these subspaces. Then $v_G$ splits into two factors 
 \begin{equation*}
 v_G(x)=\prod_{e\in E(g)}\triangle\left(\sum_{i=1}^n(v_i:e)x_i\right)  \prod_{e\in E\setminus E(g)}\triangle\left( \sum_{i=1}^n(v_i:e)x_i \right),
 \end{equation*}
 with the second factor being smooth on $A_g^{\bot} \setminus \bigcup_{g\subsetneq g'}A_{g'}^{\bot}$. 
Now we need some power counting: The integral $\int_K dx \, v_G(x)$ is over a $dn$-dimensional space. Since $A_g$ is the sum over all $A_e$ with $e\in E(g)$, it is already spanned by the edges in a \textit{spanning tree} $t$ of $g$ (a spanning tree is a subgraph without cycles meeting every vertex exactly once --- see Definition \ref{spanning tree}). A spanning tree of a connected graph with $n$ vertices has necessarily $n-1$ edges, therefore $\dim A_g=d(|V(g)|-1)$. Adding an edge to $t$ produces an independent cycle, so that $h_1(g)=|E(g)|-|V(g)|+1$. We conclude that $\dim A_g=d(|E(g)|-h_1(g))$. Each $\triangle(x)$ is of order $\scriptstyle{\mathcal{O}}(x^{2-d})$ as $x \to 0$ and there are $|E(g)|$ products in the first factor expressing $v_G$. Thus, the whole product scales with $(2-d)|E(g)|$ as $x$ approaches $A_g^{\bot}$ in $X^G$,
\begin{equation*}
\int dx \, v_G(x) \propto \int dr \, r^{\dim A_g-1+(2-d)|E(g)|}
\end{equation*}
and the integral converges if and only if 
\begin{align*}
  \dim A_g+(2-d)|E(g)| > 0 & \Longleftrightarrow  d(|E(g)|-h_1(g)) + (2-d)|E(g)|  > 0 \\
  & \Longleftrightarrow  \omega(g)  < 0  \\
  & \Longleftrightarrow g \notin \mathcal D(G).\\[-3.5em]
  \end{align*}
\end{proof}

In Section \ref{comps2} we will employ a more practical point of view. We use coordinates on $X^G$ not given by the vertex set $V'$, but by the edges of an \textit{adapted} spanning spanning tree $t$. Since every spanning tree of $G$ must have $|V|-1$ vertices, reformulating everything in coordinates given by edges of $t$ is just a change of basis for $M^{V'}$. The point here is that although it might seem to be more intuitive and ``positional'' to work with the vertex set of $G$, the formulation with $t$ is more convenient because the combinatorics of renormalization show up in the subgraph structure of $G$ and subgraphs are determined by subsets of $E$, not of $V$.

\section{Compactifications}

To systematically renormalize distributions $v_G$ coming from Feynman diagrams we want to arrange the loci of divergences in a ``nice'' way by resolving their singularities. This means, we are looking for a \textit{compactification} of $X\setminus X_{\m D}$, or, in other words, a \textit{smooth model} for the divergent arrangement in $X$. This section consists of two parts: First we study compactifications from a geometric point of view, then we focus on the underlying combinatorics.

\subsection{Geometry} \label{comps}
The problem of resolving singularities has been a major topic in algebraic geometry since the time of Newton who solved the problem of resolving curves in the complex plane. In its most basic form the problem can be formulated as follows.
\begin{defn}
 Let $X$ be an algebraic variety over a field $k$. Then a non-singular variety $Y$ is a \textit{resolution} for $X$ if there exists a proper and surjective rational map $\beta: Y \longrightarrow X$.
\end{defn}
There are various types of resolutions, depending on additional conditions on $Y$ and $\beta$. Here we demand that $\beta$ is the composition of blow-ups along smooth subvarieties of $X$. This allows for an explicit description of the manifold $Y$. 
Hironaka showed in his celebrated work \cite{hi} that for fields of characteristic zero a resolution always exists; for fields of non-trivial characteristic this is still an open problem. He gave a constructive proof using a sequence of blow-ups. The difficulty lies in the fact that one cannot proceed by just blowing up all singularities in $X$, but must choose a specific order in doing so. For an extensive treatment of this topic, including a comparison of different resolutions, we refer to \cite{jk}. 

\subsubsection{Blow-ups}\label{blow-ups}

What is meant by blowing up a subvariety of a variety $X$? 
First, we define the blow-up of the origin in $X=\mathbb{R}^n$, following \cite{gh}. The idea is to replace the origin by the space of all possible directions entering it, in such a way that all directions are disjoint. To do so set $\mathcal{E}:=\mathbb{P}(X)$ with homogeneous coordinates $[y_1: \cdots :y_n]$ and define $Y \subseteq X \times \mathcal{E}$ by
\begin{equation*}
 Y:=\big\{ \big(x_1,\ldots, x_n, [y_1:\cdots :y_n]\big)  \mid x_iy_j=x_jy_i\ \text{ for all }\ i\neq j  \big\}.
\end{equation*}
The map $\beta: Y\to X$ is then simply the projection onto the first factor.
Since the defining equations are smooth, $Y$ is a smooth submanifold of $X \times \mathcal{E}$. To define an atlas for $Y$ let for $i=1,\ldots,n$ the maps $\rho_i:\mathbb{R}^n \to X \times \mathcal{E}$ be given by
\begin{equation*}
 (x_1,\ldots, x_n) \mapsto \big(y_1, \ldots,  y_n, [y_1 : \cdots :y_n] \big) 
\end{equation*} 
where \begin{equation*} 
y_k = 
\begin{cases}  x_i & \text{ if } k=i, \\
            x_ix_k & \text{ if } k\neq i. 
          \end{cases}
\end{equation*}
Set $U_i=\rho_i(\mathbb{R}^n)$ and $\kappa_i:=\rho_i^{-1}$. Then the collection of charts $(U_i,\kappa_i)_{i\in \{1,\ldots,n\}}$ forms an atlas for $Y$. 
The submanifold $\mathcal{E}$, called the \textit{exceptional divisor}, is locally given by $\{x_i=0\}$ and covered by induced charts $(V_i,\phi_i)_{i\in \{1,\ldots,n\}}$ where $V_i:=\hat \rho_i(\mathbb{R}^{n-1})$ and $\phi_i:=\hat \rho_i^{-1}$ with 
\begin{equation*}
 \hat \rho_i := \rho_i |_{x_i=0}: \mathbb{R}^{n-1} \longrightarrow \{0\}\times \mathcal{E} \subseteq Y.
\end{equation*}
Blowing up a submanifold $S$ of $\mathbb{R}^n$ is done similarly by replacing $S$ by the projectivization of its normal bundle. More precisely, if $S$ is locally given by $\{x_1=\cdots=x_k=0\}$, then one proceeds as above but restricts the defining equation to these coordinates,
\begin{equation*}
 Y:=\big\{ \big(x_1,\ldots, x_n, [y_1:\cdots :y_k]\big)  \mid x_iy_j=x_jy_i\ \text{ for all }\ i\neq j \in \{1,\ldots,k \} \big\}.
\end{equation*}
Note that this construction is independent of the chosen coordinates. It can be generalized to the case where $S$ is a subvariety of a smooth variety $X$: Blow up locally, then globalize by patching together the local blow-ups. 

If $S'\subseteq X$ is another submanifold that is distinct from $S$, then $S'$ is essentially unaffected by the blow-up process. However, if it has nonempty intersection with $S$, then $S'$ has two ``preimages'' in $Y$: The \textit{strict transform} of $S'$ is defined as the closure of $\beta^{-1}(S' \setminus S)$ in $Y$, while the preimage $\beta^{-1}(S')$ is called the \textit{total transform} of $S'$. Loosely speaking, the blow-up makes degenerate intersections transversal and transversal ones disjoint. Therefore, if building a resolution consists of multiple blow-ups, the order of blowing up is important! 

We introduced here the algebro-geometric version of blowing up. There is also a differential-geometric equivalent, where one replaces the locus to be blown up by its normal spherebundle (as used in \cite{as}). Both cases have drawbacks: Using the projective normal bundles leads to $Y$ being non-orientable in general, while the differential-geometric blow-up produces a manifold with boundary.

\subsubsection{Wonderful models} \label{wonderful models}
The general setup is the following: Let $X$ be a finite dimensional smooth variety over a field $k$ of characteristic zero. An \textit{arrangement} $\m A$ in $X$ is a finite family of smooth subvarieties of $X$. Let $M(\m A)$ denote the complement of the arrangement, $M(\m A):=X\setminus \cup_{A\in \m A} A$.

\begin{defn}
A \textit{smooth model} for the arrangement $\mathcal A$ is a pair $(Y_{\mathcal A}, \beta)$ where $Y_{\m A}$ is a smooth variety and $\beta : Y_{\mathcal A} \longrightarrow X$ is a proper surjective map with the following properties:
\begin{enumerate}
 \item  $\beta$ is an isomorphism outside of $\m E:= \beta ^{-1}( X\setminus M(\m A))$.
 \item  $\m E$ is a normal crossing divisor, i.e.\ there exist local coordinates such that it is given by $\m E=\{(x_1,\ldots,x_n) \ | \ x_1 \cdots  x_k = 0\} $.
 \item $\beta$ is a composition of blow-ups along smooth centers.
\end{enumerate}
\end{defn}
Recall that $\beta$ is proper if and only if $\beta^{-1}(K)$ is compact for all compact sets $K\subseteq X$; this is why smooth models are sometimes also called compactifications.
From \cite{hi} we know that such a model always exists, even in much more general situations. In their seminal paper \cite{fm} Fulton and MacPherson constructed a compactification of the configuration space 
\begin{equation*}
F_n(X):= \{ (x_1, \ldots, x_n) \in X^n \ | \ x_i\neq x_j\  \text{ for all }\ i\neq j\}
\end{equation*}
for non-singular varieties $X$. This is just an example of a smooth model for the arrangement given by all diagonals $D_I$ in $X^n$,
\begin{equation*}
 \m A=\big\{ D_I \mid I\subseteq \set{n}\big\}\quad \text{where}\quad D_I=\{x_i=x_j \mid \forall i,j \in I \}.
\end{equation*}
Inspired by the techniques used in \cite{fm}, DeConcini and Procesi developed a systematic way to construct smooth models for general linear arrangements. Since their technique is local, it can be generalized to arrangements in smooth varieties (see \cite{ll}), but we do not need this here and stick to the notation of \cite{dp}.

Let $V$ be a finite dimensional $k$-vector space (here $k=\mb R$) and $\m A$ be a linear arrangement in the dual $V^*$, i.e.\ a finite family $\{A_1, \ldots , A_k \}$ of linear subspaces of $V^*$ (for the DeConcini-Procesi construction it is more convenient to work in the dual). We first give an abstract definition of a smooth model $Y_{\mathcal A}$ for $\m A$, then we construct it explicitly. 

\begin{defn}[Wonderful definition I]\label{wonderful definition 1}
Let $\m A$ be a linear arrangement in $V^*$. For every $A \in \mathcal A $ the projection $ \pi_A : V \longrightarrow V/A^{\bot} \longrightarrow \mathbb{P}(V/A^{\bot})$ is a well-defined map outside of $A^{\bot}$. Doing this for every element in the arrangement we obtain a rational map 
\begin{equation*}
 \pi_{\mathcal{A}} : M(\m A) \longrightarrow \prod_{A \in \mathcal A} \mathbb{P}(V/A^{\bot}).
\end{equation*}
The graph $\Gamma ( \pi_{\mathcal{A}} )$ of this map, a closed subset of $M(\m A) \times \prod_{A \in \mathcal A} \mathbb{P}(V/A^{\bot})$, embeds as locally closed subset into $ V \times \prod_{A \in \mathcal A} \mathbb{P}(V/A^{\bot})$. The \textit{wonderful model} $Y_{\mathcal A}$ is defined as the closure of the image of this embedding.
\end{defn}

The second way of defining $Y_{\m A}$ is to explicitly construct it by a sequence of blow-ups (this sequence is actually completely determined by the combinatorics of the \textit{intersection poset} $P(\mathcal A)$, a point we will use extensively in the following sections). For the wonderful construction we need to introduce some terminology. The first notion is based on the fact that $Y_{\mathcal A}$ is also a wonderful model for arrangements $ \m A' $, as long as $\mathcal A \subseteq  \mathcal A'  $ is a \textit{building set} for $\m A'$. The idea is that an arrangement may carry too much information and in this case one needs only a subfamily $\mathcal B \subseteq \mathcal A$ to encode this information. While the choice of a building set controls the geometry of the wonderful model, more precisely of the exceptional divisor $\m E$, certain subsets of $\m B$, the $\mathcal B$-\textit{nested sets}, and the choice of a \textit{$\m B$-adapted} \textit{basis} of $V$ are the crucial elements in the explicit 
construction of an atlas for $Y_{\mathcal A}$.
We cite the main definitions and results from DeConcini and Procesi; for the proofs we refer the reader to \cite{dp}.

\begin{defn}[Building sets]
Let $\mathcal A$ be an arrangement in $V^*$. A subfamily $\mathcal B \subseteq \mathcal A$ is a \textit{building set} for $\m A$ if 
\begin{enumerate}
 \item Every $A \in \mathcal A$ is the direct sum $A = B_1 \oplus \cdots \oplus B_k $ of the maximal elements of $\mathcal B$ contained in $A$. 
 \item This decomposition property also holds for all $A' \in \m A$ with $A'\subseteq A$, i.e.\ $A'=(B_1\cap A') \oplus \cdots \oplus (B_k\cap A')$.
\end{enumerate}
\end{defn}
There are two important examples, the \textit{maximal building set}, given by all elements of $\m A$, and the \textit{minimal building set} $I(\m A)$. The latter consists of all $A\in \m A$ that do not allow for a non-trivial decomposition. Note that every other building set $\m B$ satisfies $I(\m A)\subseteq \m B \subseteq \m A$. 
In \cite{dp} it is shown that for every building set $\m B\subseteq \m A$ the variety $Y_{\m B}$ as defined above is a smooth model for $\m A$. Moreover, the exceptional divisor $\m E$ is the union of smooth irreducible components $\m E_B$, one for each $B \in \m B$. 

\begin{defn} [Nested sets]
Let $\mathcal B$ be a building set. $\mathcal N \subseteq \mathcal B$ is $\mathcal B$-\textit{nested} if the following holds: For all subsets $\{ A_1, \ldots , A_k \} \subseteq \mathcal N$ of pairwise incomparable elements their direct sum does not belong to $\mathcal B$. 
\end{defn}
Nested sets are one main ingredient in the description of $Y_\m B$, the second one being markings of an adapted basis of $V^*$. While nested sets reflect the combinatorics of the stratification of $\m E$, the markings are related to the dimension of each submanifold in this stratification. Together they describe all components of the exceptional divisor.
\begin{defn}[Adapted bases]\label{adapted basis 1}
A basis $B$ of $V^*$ is $\mathcal N$-\textit{adapted} if for all $A \in \mathcal N$ the set $B\cap A$ generates $A$. A \textit{marking} of an $\mathcal N$-adapted basis is for every $A \in \mathcal N$ the choice of an element $b_A \in B$ with $p(b_A)=A$. Here $p=p_{\mathcal N}$ is the map assigning to $x\in V^* \setminus \{0\}$ the minimal element of $\mathcal N \cup \{V^*\}$ containing $x$ (it exists because $\m N$ is nested).
\end{defn}
The map $p$ and a marking define a partial order on $B$, 
\begin{equation*}
 b \preceq b' \Longleftrightarrow p(b) \subseteq p(b') \quad\text{and}\quad b' \text{ is marked. }
\end{equation*}
This partial order defines a map $\rho=\rho_{\m N,B}: \mb R^B \to V$ as follows: For every $x=\sum_{b\in B}x_b b \in \mb R^B$ the image $\rho(x)$ is an element of $V=\hom(V^*,\mb R)$ given by 
\begin{equation*}
 B \ni b \mapsto \begin{cases} \prod_{p(b)\subseteq A} x_{b_A}  &\text{ if $b$ is marked,} \\
                    x_b \prod_{p(b)\subseteq A} x_{b_A}  &\text{ else}.
      \end{cases}
\end{equation*}
Viewing the elements of $B$ as nonlinear coordinates on $V$ and setting $x_A:=x_{b_A}$ we can write $\rho$ as
\begin{equation*}
\rho(x)_b=\rho(x)(b) = \begin{cases}
                                     \prod_{p(b)\subseteq A} x_{A} \ &\text{ if $b$ is marked,} \\
                                     x_b \prod_{p(b)\subseteq A} x_{A} \ &\text{ else}.
                                    \end{cases}
\end{equation*}
The next proposition shows that $\rho$ has all the properties of a local description of a composition of blow-ups.

\begin{prop}
For every nested set $\m N$ and an adapted and marked basis $B$ the map $\rho=\rho_{\m N,B}$ is a birational morphism with the following properties: 

It maps the subspace defined by $x_A = 0$ onto $A^{\bot}$ and it restricts to an isomorphism 
\begin{align*}
V \setminus \bigcup_{A\in \mathcal N} \{ x_A = 0 \} & \overset{\cong}{\longrightarrow} V \setminus \bigcup_{A\in \mathcal N} A^{\bot}.
\end{align*} 
Furthermore, every $v$ in $V^* \setminus \{0\}$ with $p(v)=A\in \mathcal N$ is mapped by $\rho(x)$ to 
\begin{equation*}                                            
  \rho(x)(v)=P_v(x) \prod_{b_A \preceq b} v_{b} x_{b}  
\end{equation*}
 where $P_v$ is a polynomial, depending only on $\{x_b\}_{b \prec b_A}$ and linear in each variable.
\end{prop}

\begin{defn}[Wonderful definition II] \label{wonderful definition 2}
Let $\mathcal N$ be a $\mathcal B$-nested set for a building set $\mathcal B \subseteq \mathcal A$ and $B$ an adapted, marked basis. Define $Z_A \subseteq \mb R^B$ by $Z_A=\{ P_v=0, v\in A \}$, the vanishing locus of all $P_v$ for $v \in A$. Then for every $A \in \m B$ the composition of $\rho$ with the projection $\pi_A: V \setminus A^{\bot} \to \mb P(V / A^{\bot})$ is well-defined outside of $Z_A$. 
Composing the map $\rho$ with $\Gamma(\pi_\m B): M(\m B) \to V \times \prod_{A \in \mathcal{B} } \mathbb{P}(V / A^{\bot})$ defines an open embedding 
\begin{equation*}
 (\Gamma( \pi_\m B) \circ \rho)_{\mathcal N,B}: \mb R^B \setminus \bigcup_{A \in \mathcal{B} } Z_A \longrightarrow Y_{\mathcal B}. 
\end{equation*}
Set $U_{\mathcal{N},B}:= \text{im}\big((\Gamma( \pi_\m B) \circ \rho)_{\mathcal N,B}\big) $ and $\kappa_{\mathcal N,B}:=(\Gamma( \pi_\m B) \circ \rho)^{-1}_{\mathcal N,B}$. Varying over all $\mathcal B$-nested sets $\m N$ and adapted, marked bases $B$, we obtain an atlas $(U_{\mathcal{N},B},\kappa_{\mathcal{N},B} )$ for the wonderful model $Y_{\mathcal B}$. The map $\beta$ is just the projection onto the first factor, in local coordinates given by $\rho$.
\end{defn}

That this really defines a smooth model for the arrangement $\m A$ follows from
\begin{thm}[Geometry of the wonderful model]\label{geom of wm}
Let $\m B$ be a building set for $\m A$. The wonderful model $Y_{\m B}$ has the following properties:
\begin{enumerate}
 \item The exceptional divisor $\m E$ is normal crossing, i.e. 
\begin{equation*}
 \m E:=\beta^{-1}\left(\bigcup_{A\in \mathcal B} A^{\bot}\right)\overset{loc.}{=}\left\{\prod_{A \in \mathcal N}x_A=0 \right\}.
\end{equation*}
 \item $\m E$ is the union of smooth irreducible components $\m E_A$ where $A \in \mathcal B$ and $\beta(\m E_A)=A^{\bot}$. A family of these components $\m E_{A_1}, \ldots, \m E_{A_k}$ has non-empty intersection if and only if $\{ A_1, \ldots, A_k\}$ is a $\mathcal B$-nested set. In this case the intersection is transversal and irreducible. 
 \item For $A$ minimal in $\m B \setminus I(\m A)$ let $A=A_1\oplus \cdots \oplus A_k$ be its irreducible decomposition. Set $\m B'=\m B \setminus \{A\}$. Then $Y_{\m B}$ is obtained from $Y_{\m B'}$ by blowing up $\m E_A=\m E_{A_1}\cap \cdots \cap \m E_{A_k}$.
 \item For $A$ minimal in $\m B= I(\m A)$ set $\m B'=\m B \setminus \{A\}$. Then $Y_{\m B}$ is obtained from $Y_{\m B'}$ by blowing up the proper transform of $A^{\bot}$. 
 \end{enumerate}
\end{thm}

As stated before, the most famous example of a wonderful model is the Fulton-MacPherson compactification of the configuration space $F_n(X)$ in the case where $X$ is a linear space. It is the minimal wonderful model for the arrangement of all (poly-)diagonals in $X^n$,
\begin{align*}
\m A&=\{D_{\pi} \mid \pi \text{ is a partition of } \set{n} \}, \\
D_{\pi}&=\{x_i=x_j \mid i,j \text{ lie in the same partition block of $\pi$ }\}. 
\end{align*}
Here the minimal building set consists of all simple diagonals in the $n$-fold product of $X$. The wonderful model for the maximal building set was studied by Ulyanov in \cite{ul} and called a \textit{polydiagonal compactification} of configuration space. The main difference, apart from the geometry of the exceptional divisor, is the blowup sequence in the construction. In \cite{ul} the model is obtained by successively blowing up (the strict transforms of) all elements of the building set by increasing dimension, but in the minimal case one has to proceed with care; some strict transforms of diagonals to be blown up in the next step might still have nonempty intersection and in this case the result depends on the order of blowups. To separate them before proceeding requires additional blow-ups, exactly those given by the additional elements in the maximal building set. These are the polydiagonals, obtained by intersecting simple diagonals.
The interested reader is encouraged to study the example $F_n(X)$ for $X=\mb R$ and $n>3$ (for smaller $n$ minimal and maximal models coincide). It is a well studied object, the \textit{real rank $n-1$ braid arrangement}, see for example \cite{fe}.

The next step is to adapt this construction to the case of the divergent arrangement associated to a Feynman graph $G$. In \cite{bbk} this is done by examining the special structure of the elements of the arrangement $\m A_{\m D(G)}=\{ A_g \mid g \subseteq G \text{ divergent} \}$. These properties can be directly formulated in graph theoretical terms. Here we will focus even more on this combinatorial flavour and express everything with the help of the poset of divergent subgraphs of $G$.
\subsection{Combinatorics}\label{comps2}

We reformulate the central objects of the last section in terms of a poset associated to $\m A$. We focus on arrangements coming from graphs via Feynman rules, but note that from every given arrangement we can form the \textit{intersection poset} to study its combinatorics. 

\begin{defn}
 A \textit{poset} $(\m P,\leq )$ is a finite set $\m P$ (we consider here only finite graphs and posets) endowed with a partial order $\leq$. 
 
 We say that $p$ \textit{covers} $q$ if $p>q$ and there is no $r\in \m P$ with $p>r>q$. The \textit{closed interval} $[p,q]=\m P_{[p,q]}$ is defined as the set of elements $r \in \m P$ satisfying $p\leq r\leq q$ . The \textit{open interval} $(p,q)=\m P_{(p,q)}$ and the subsets $\m P_{< p}$, $\m P_{\leq p}$, $\m P_{> p}$, $\m P_{\geq p}$ are defined similarly. We denote by $\hat 0$ and $\hat 1$ the unique minimal and maximal elements of $\m P$ if they exist. 
\end{defn}
A poset is best visualized by drawing its \textit{Hasse diagram}, a directed graph with its vertices given by the elements of $\m P$ and edges between every pair of elements $p,q \in \m P$ such that $p$ covers $q$. Another way to encode the data of $\m P$ is the \textit{order complex} $\Delta (\m P)$. It is the abstract simplicial complex defined by its $k$-faces being the linearly ordered $k+1$-element subsets of $\m P$. The order complex stores all the combinatorial information of $\m P$ as is demonstrated by Theorem \ref{gm} taken from \cite{gm}.

\begin{defn}[Intersection lattice]
  Let $V$ be an $n$-dimensional real vector space and let $\mathcal A:=\{ A_1, \ldots, A_m \}$ be an arrangement in $V$. Every arrangement gives rise to a poset (actually a \textit{lattice}, defined below) with its underlying set consisting of all possible intersections of elements in $\mathcal A$,
\begin{equation*}
\m P=\m P(\mathcal A)=\left\{\bigcap_{i\in I} A_i \,\bigg|\,  I \subseteq \{ 1, \ldots, m \} \right\},
\end{equation*}
partially ordered by reverse inclusion. It is called the \textit{intersection poset/lattice} of $\m A$.
In addition, $\m P(\m A)$ is equipped with a ranking, i.e.\ a map $r:\m P(\m A) \to \mb N$ mapping each element of $\m P(\m A)$ to the codimension of the corresponding intersection in $V$.
\end{defn}

\begin{thm}(Goresky, MacPherson) \label{gm}
Let $H$ denote the (singular) hom\-ology functor. Let $\mathcal A$ be an arrangement in $V$ and let $M(\m A)$ denote the complement. Then
 \begin{equation*}
  H_k(M(\m A),\mathbb{Z}) \cong \bigoplus_{A \in \m P(\m A)} G_k(A)
 \end{equation*}
 where
 \begin{equation}\label{cohom}
  G_k(A):= \begin{cases}
           H^{-k}(\text{point},\mb Z) & \text{if} \ A=\hat 0, \\
           H^{r(A)-k-1}(\text{point},\mb Z) & \text{if $A$ covers $\hat 0$,} \\
           \tilde H^{r(A)-k-2} \big(\Delta(\m P_{(\hat 0,A)}),\mb Z\big) & \text{otherwise}. 
          \end{cases}         
 \end{equation}
 \end{thm}

Recall from Section \ref{distributions} the definition of the singular and divergent arrangements of a graph $G$. They give rise to corresponding intersection posets, but we can also define them directly in terms of $G$. 

\begin{defn}\label{graph poset}
 To a graph $G$ we associate the \textit{(saturated) graph poset} $(\m G(G),\subseteq)$ consisting of the set of all saturated subgraphs of $G$, partially ordered by inclusion. A connected subgraph $g\subseteq G$ is \textit{saturated} if the following holds:
\begin{equation*}
 \forall t \text{ span. tree of }g : \forall e\in E(G\setminus g) : t\text{ is not a spanning for }g\cup e.
\end{equation*}
If $g$ has more than one connected components, it is saturated if every component is.
\end{defn}
In terms of the singular arrangement a saturated subgraph $g$ is the maximal subgraph of $G$ defining $A_g \in \m A_{\m G(G)}$. This means, that adding an edge to a saturated graph necessarily enlarges the space $A_g$, while removing an edge might still define the same subspace of $(X^G)^*$.

\begin{example}
Let $K_3$ be the complete graph on 3 vertices. The saturated subgraphs are the three single-edged subgraphs and $K_3$ itself. 
\end{example}

\begin{defn}
The \textit{divergent graph poset} $\m D(G)$ is given by the subset of $\m G(G)$ formed by all divergent subgraphs, partially ordered by inclusion.
\end{defn}

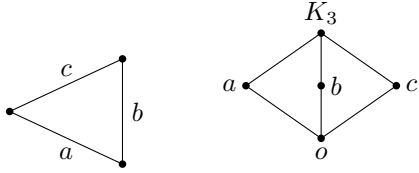
\begin{figure}[ht]
\centering 
\begin{tikzpicture}[node distance=0.7cm and 1.5cm]
\coordinate (v1);
\coordinate[above right=of v1] (v3);
\coordinate[below right=of v1] (v2);
\draw (v1) -- node[label=below:$a$] {}  (v2);
\draw (v1) -- node[label=above:$c$] {} (v3);
\draw (v2) -- node[label=right:$b$] {} (v3);
\fill[black] (v1) circle (.05cm);
\fill[black] (v2) circle (.05cm);
\fill[black] (v3) circle (.05cm);
\end{tikzpicture}
\quad \quad
 \begin{tikzpicture}[node distance=0.7cm and 1cm]
\coordinate[label=left:$a$] (v1);
\coordinate[below right=of v1,label=below:$o$] (v3);
\coordinate[right=of v1,label=right:$b$] (v2);
\coordinate[right=of v2,label=right:$c$] (v4);
\coordinate[above=of v2,label=above:$K_3$] (v5);
\draw (v1) -- (v5);
\draw (v3) -- (v1);
\draw (v3) -- (v2);
\draw (v3) -- (v4);
\draw (v2) -- (v5);
\draw (v4) -- (v5);
\fill[black] (v1) circle (.05cm);
\fill[black] (v2) circle (.05cm);
\fill[black] (v3) circle (.05cm);
\fill[black] (v4) circle (.05cm);
\fill[black] (v5) circle (.05cm);
\end{tikzpicture}
\caption{$K_3$ and the Hasse diagram of $\m G(K_3)$}\label{kthree}
\end{figure}

As already seen in Proposition \ref{sing supp}, $\m G$ and $\m D$ (from now on we drop the index $G$) carry all the information necessary for renormalization. Note that both posets have an unique minimal element, the empty graph, which we denote by $o$. In our convention $o$ is defined by $E(o)=\emptyset$. 

For the divergent arrangement of a connected and at most logarithmic graph $G$, Theorem \ref{gm} allows us to compute the homology of $M(X^G_{\m D})$, the complement of the divergent loci in $X^G$. It is determined by the set of \textit{atoms} of $\m D$, the minimal elements in $\m D_{>o}$. These elements are precisely the primitive subgraphs of $G$.
\begin{prop}
 Let $G$ be connected and at most logarithmic. Define $n_i$ to be the number of atoms $g\in \m D$ with $r(g)=\dim A_g=di$ (i.e.\ the primitive subgraphs on $i+1$ vertices). Let $\alpha \in \mb N^l$ be a multi-index with $\alpha_i \geq \alpha_j$ for $1\leq i < j \leq l$ and $|\!|\alpha |\!|_1=l$. The homology of $M( X^G_{\m D} )$ is then given by
 \begin{equation} \label{hom}
  H_k(M(X^G_{\m D}),\mathbb{Z}) \cong 
  \begin{cases}
           \mathbb{Z} & \text{ if }  k=0 , \\
           \mathbb{Z}^{n_i} & \text{ if }  k=d i -1, \\
           \mathbb{Z}^{\binom{n_i}{2}} & \text{ if }  k=2 d i -2, \\
           \mathbb{Z}^{n_{i_1} n_{i_2} } & \text{ if }  k= d (i_1 + i_2 ) - 2, \\
           \cdots &   \\
           \mb Z^{ \prod_{ j = 1}^l \binom{ n_{i_j} }{\alpha_j}} & \text{ if } k=d \sum_{j=1}^l \alpha_j i_j    - l, \\
           \cdots & 
          \end{cases}         
 \end{equation}

\end{prop}
\begin{proof}
The atoms of $\m D$ determine the topology of the complement because the corresponding subspaces $A_g^{\bot}$ contain all other divergent subspaces. 
Ab\-breviate $H_k(M(X^G_{\m D}),\mathbb{Z})$ by $H_k$. Using Theorem \ref{gm} we have $H_0=\mb Z$. Moreover, there is a generator in $H_k$ with $k=r(g)-1=di-1$ for every atom $g$ such that $r(g)=di$.  
 For an element $\gamma$ that is given by the union of atoms we have to use the third row in Equation \eqref{cohom}: If $\gamma$ is the union of two atoms $g$ and $h$, the subcomplex $\Delta(\m D_{(o,\gamma)})$ consists of $2$ disconnected points (representing the two atoms). Therefore we have $\binom{n_i}{2}$ generators in dimension $k=2di-2$ if $r(g)=r(h)=di$ and $n_{i_1} n_{i_2}$ generators in dimension $k=d(i_1+i_2)-2$ if $r(g)=di_1$ and $r(h)=di_2$. 
 If $\gamma$ is the union of $l > 2$ atoms, the interval $ (o,\gamma)$ consists of these atoms and all unions thereof. It is the face poset of the standard $(l-1)$-simplex $\bigtriangleup^{l-1}$ with interior removed. Thus, $\Delta(\m D_{(o,\gamma)})=\Delta( \m F (\partial  \bigtriangleup^{l-1} ) ) \cong \partial \bigtriangleup^{l-1}  \cong S^{l-2}$. Since $\ti H^k(S^{l-2})$ equals $\mb Z $ if $k=l-2$ and is trivial else, we conclude that there are generators in $H_k$ coming from such elements $\gamma$ if $k=r(\gamma)-l$. Let $\alpha \in \mb N^l$ with $\alpha_i \geq \alpha_j$ for $1\leq i < j \leq l$ and $|\!|\alpha |\!|_1=l$. If $r(\gamma)=d \sum_{j=1}^l \alpha_j i_j$ then such $\gamma$ can be formed out of $l$ atoms in $\prod_{ j = 1}^l \binom{ n_{i_j} }{\alpha_j}$ possible ways and \eqref{hom} follows.
\end{proof}
 
\subsubsection{The divergent graph lattice}

We continue studying the graph poset $\m G$ in more detail. As it turns out it has extra structure, it is a \textit{lattice}. 

\begin{defn}[Lattices]
 Let $(\m P,\leq)$ be a poset and $p,q \in \m P$. A least upper bound or \textit{join} of $p$ and $q$ is an upper bound $r$ for both elements such that every other upper bound $s$ satisfies $r\leq s$. If the join of $p$ and $q$ exists, it is unique and denoted by $p\vee q$.
 Dually one defines a greatest lower bound or \textit{meet} of two elements $p$ and $q$ in $\m P$, denoted by $p\wedge q$.
 
 $\m P$ is called a \textit{join-semilattice} (\textit{meet-semilattice}) if for all $p,q \in \m P$ the join $p\vee q$ (the meet $p\wedge q$) exists. $\m P$ is called a \textit{lattice} if it is both a join- and a meet-semilattice.
\end{defn} 
For any arrangement $\m A$ the intersection poset $\m P(\m A)$ is a lattice: If one orders the elements of $\m P(\m A)$ by reverse inclusion, the join operation is just given by set theoretic intersection. The statement then follows from the fact that every finite join-semilattice with $\hat 0$ (represented by the empty intersection, the ambient space $V$) is a lattice (Proposition 3.3.1 in \cite{st}). 
Regarding the definition of the partial order by inclusion or reverse inclusion there are different conventions used in the literature. Both have their advantages and can be converted into the other since the dual of any lattice, i.e.\ the lattice with reversed order, is a lattice as well. We use reverse inclusion because it matches the convention in \cite{dp} using arrangements in the dual and it fits with the natural partial order on subgraphs. 

Since $\m G$ is the intersection poset of the (dual) singular arrangement $\m A_{\m G}$ in $X^G$, it is a lattice. Clearly, if $\m P \subseteq \m G$ is closed under union and intersection, it is the intersection lattice of some corresponding arrangement. This is the case for the set of divergent subgraphs:

\begin{prop}\label{latt}
 Let $G$ be at most logarithmic. Then $(\m D,\subseteq)$ is a lattice.
\end{prop}

\begin{proof}
 For $g,h\subseteq G$ divergent subgraphs we define the join and meet operations in $\m D$ by
 \begin{align*}
  g\vee h := g\cup h, \\
  g\wedge h := g \cap h.
 \end{align*}
 Suppose $g$ and $h$ have $k$ shared edges and abbreviate $h_1(g \cap h)$ by $l$. Let $m$ be the number of ``new cycles'' created by uniting $g$ and $h$. In formulae
 \begin{align*}
  E(g \cup h)&= E(g) + E(h) - k \\
  h_1(g \cup h)&= h_1(g) + h_1(h) + m - l. 
  \end{align*}
  From this we conclude that the superficial degree of divergence of $g\cup h$ is given by
  \begin{align*}
  \omega(g \cup h)&= d(h_1(g) + h_1(h) + m - l)-2(E(g) + E(h) - k) \\
  &=d(m - l) + 2k \overset{!}{\leq} 0.
 \end{align*}
 Split $k=k_l + k_0$ into edges in the generators of $H_1(g\cap h)$ and those that are not. Then $dl\leq 2k_l$ and
  \begin{equation*}
   0 \geq \omega(g\cup h)\geq dm + 2k_0 \geq 0,
  \end{equation*}
  thus $m= k_0 = 0$ and 
  \begin{equation*} 
  0 = \omega(g \cup h)=dl - 2k_l = \omega(g \cap h).
  \end{equation*}
  Therefore $g \cup h $ and $g \cap h $ are both divergent subgraphs of $G$. Clearly, they are the minimal (maximal) elements of $\m D$ bounding $g$ and $h$ from above (below).    
\end{proof}
With the methods used in the above proof we are able to show another property of $\m G$ and $\m D$. They are \textit{graded} lattices.

\begin{defn}
 A poset $(\m P, \leq )$ is \textit{graded} if it is equipped with a map $\tau: \m P \to \mathbb{N}$ that has the following two properties: $\tau$ is order preserving with respect to the natural order on $\mathbb{N}$ and if there are $p,q \in \m P$ with $p$ covering $q$, then $\tau(p)=\tau(q)+1$.
\end{defn}

\begin{prop}
 For any connected graph $G$ the graph lattice $\m G$ is graded.
\end{prop}
\begin{proof}
 The map $\tau$ sends every saturated subgraph $g \subseteq G$ to $d^{-1}  \dim A_{g}=V(g)-c_g$ ($c_g$ denoting the number of connected components of $g$, cf.\ the proof of Proposition \ref{sing supp}). Clearly, $\tau$ is order preserving and $\tau(p)=\tau(q)+1$ holds for $p$ covering $q$ because of the saturated condition.
\end{proof}

\begin{prop}\label{gradlatt}
 Let $G$ be at most logarithmic. Then $\m D$ is a graded lattice.
\end{prop}
To prove this we use Proposition 3.3.2 from \cite{st}.

\begin{prop}
 Let $\m L$ be a finite lattice. The following two conditions are equivalent:
 \begin{itemize}
 \item[1.] $\m L$ is graded and the map $\tau$ satisfies $\tau(x) + \tau(y) \geq \tau(x \wedge y) + \tau (x \vee y)$ for all $x,y \in \m L$.
 \item[2.] If $x$ and $y$ both cover $x \wedge y$, then $x \vee y$ covers both $x$ and $y$.
 \end{itemize}
\end{prop}

\begin{proof}[Proof of Proposition \ref{gradlatt}]
  We argue by contradiction: Let $g,h \subseteq G$ be divergent and suppose there is a $\gamma \in \m D$ with $g < \gamma < g \vee h$, i.e.\ $g \vee h$ does not cover both $g$ and $h$. First, note that $\gamma \cap h \neq \emptyset $ because otherwise $\gamma$ would not be a subgraph of $g\vee h$. From Proposition \ref{latt} we know that $\gamma \cap h$ is divergent. But then $g \wedge h \ < \  \gamma \wedge h \ < \ h$, which means $h$ is not covering $g\wedge h $.
\end{proof}
We will not use this here but for the sake of completeness we mention one additional property of $\m D$ which actually implies gradedness and modularity ($\m L$ is \textit{modular} if one of the properties in the previous proposition holds for both $\m L$ and its dual). From a combinatorial viewpoint \textit{distributive} lattices are important because this extra structure allows one to prove many powerful theorems, for example Birkhoff's famous Representation Theorem~\cite{bh}.
\begin{prop}
 Let $G$ be at most logarithmic. Then $\m D$ is a \textit{distributive} lattice:
 \begin{align*}
  f \vee (g \wedge h) = & (f \vee g) \wedge ( f \vee h) , \\
  f \wedge (g \vee h) = & (f \wedge g) \vee (f \wedge h),
 \end{align*}
 for all $f,g,h$ in $\m D$.
\end{prop}

\begin{proof}
 Since one of the properties implies the other, we will only proof the first one. Moreover, the proof works exactly the same in the second case. Let $f,g,h \subseteq G$ be divergent. Compare the edge set of the graphs on the left and the right:
 \begin{align*}
  E\big( f \vee (g \wedge h) \big) = & E\big( f \cup (g \cap h) \big) = E ( f ) \cup E (g \cap h ) \\
  = & \big( E( f )\cup E(g) \big) \cap \big( E (f) \cup E( h ) \big) \\
  = & E (f \cup g) \cap E ( f \cup h) = E \big( (f \cup g)  \cap ( f \cup h) \big)  \\
  = & E \big( (f \vee g) \wedge ( f \vee h) \big).
 \end{align*}
\end{proof}

\subsubsection{Wonderful models revisited}\label{wm revis}

We reformulate wonderful models in terms of the graph lattice $\m G$. This is based on \cite{fe} where a combinatorial version of the wonderful model construction is developed for any (finite) lattice $\m L$.
In general we can associate to every arrangement the corresponding intersection lattice defined in the previous section. It is the combinatorics of this lattice that reflect the topological properties of the wonderful models as seen for example in Theorem \ref{gm}. Another example is the following theorem by Feichtner that relates combinatorial and geometric wonderful models via a \textit{combinatorial blow-up} (Definition 3.5 and Theorem 3.6 in \cite{fe}).

\begin{thm}
 Let $\m L$ be an intersection lattice, $\m B$ a combinatorial building set in $\m L$, and $B_1,\ldots,B_t$ a linear order on $\m B$ that is non-increasing with respect to the partial order on $\m L$. Then consecutive combinatorial blowups in $B_1,\ldots,B_t$ result in the face poset of the nested set complex $\Delta_{\m N}(\m L,\m B)$,
 \begin{equation*}
 \text{Bl}_{G_t}(\cdots(\text{Bl}_{G_2}(\text{Bl}_{G_1})))=\m F\left( \Delta_{\m N}(\m L,\m B) \right).
 \end{equation*}
\end{thm}

Although the following definitions apply to any lattice $\m L$, to connect with Section \ref{graphs to arr} think of $\m L$ as being given by the singular or divergent arrangement of a connected and at most logarithmic graph $G$. In the (equivalent) combinatorial formulation below, building sets and nested sets are certain subposets of $\m L$ (a subposet of a poset $(\m P,\leq)$ is a subset of $\m P$ with the induced partial order). In some cases these subsets are even lattices, although not necessarily sublattices since the meet and join operations need not be induced by the corresponding operations on~$\m L$.

\begin{defn}[Combinatorial building sets]
 Let $\m L$ be a lattice. A non-empty subset $\m B $ of $\m L$ is a \textit{combinatorial building set} for $\m L$ if the following holds: For all $p \in \m L_{>\hat 0}$ and $\{q_1, \ldots , q_k \}= \text{max}\ \m B_{\leq p}$ there is an isomorphism of posets
 \begin{equation}\label{comb build set}
  \varphi_p: \prod_{i=1}^{k} [\hat 0,q_i] \longrightarrow [\hat 0,p]
 \end{equation}
 with $\varphi_p(\hat 0,\ldots , q_j , \ldots, \hat 0)=q_j$ for $j=1, \ldots , k$. 
\end{defn}
This defines combinatorial building sets which are more general than the building sets introduced in Section \ref{comps}. To get the notion according to DeConcini and Procesi we have to demand an additional geometric compatibility condition.
\begin{defn}[Geometric building sets]
 We call $\m B$ a \textit{geometric building set} for $\m L$ if it is a combinatorial building set and  
 \begin{equation*}
 \dim A_p = \sum_{i=1}^{k} \dim A_{q_i}.
 \end{equation*}
\end{defn}
Note that if $\m L\subseteq \m G$, since $\dim A_g=d(|V(g)|-1)$ (or $d(|E(g)|-h_1(g))$ if $g$ is divergent), we can express this geometric condition also purely in graph theoretic terms.

\begin{example}
 For every lattice $\m L$ itself is a building set, the \textit{maximal building set}. The \textit{minimal building set} is given by the irreducible elements of $\m L$. It is formed by all $p \in \m L$ for which there is no product decomposition of the interval $[\hat 0,p]$ as in \eqref{comb build set}. We denote this building set by $I(\m L)$.
\end{example}

The geometric condition gives a handy criterion to check whether a given element is irreducible or not.

\begin{lem} \label{irr}
 Let $\m L \subseteq \m G$ be a lattice. Let $g\in \m L$ be the union of irreducible subgraphs $g= g_1 \cup \cdots \cup g_k$ with non-empty overlap $h=g_1 \cap \cdots \cap g_k$. W.l.o.g.\ assume that the $g_i$ are maximal with this property. Then $g$ is irreducible. 

 Vice versa, for every reducible element $g\in \m L \setminus I(\m L)$ we have that $g$ is the union of some $g_1, \ldots, g_k \in I(\m L)$ with
 \begin{equation*}
  \bigcap_{i =1 }^k g_i= \bigcup_{v \in V'}v \sim o 
 \end{equation*}
 for some vertex set $V' \subseteq V(G)$.   
\end{lem}

\begin{proof}
Write $d(g)$ for $\dim A_g$. If $g$ would be reducible, then $d(g) = \sum_{i=1}^{k} d(g_i)$ because the $g_i$ form the set $\text{max}\ \m I(\m L)_{\leq g}$. On the other hand, $d(g)=\linebreak \sum_{i=1}^{k} d(g_i) - d(h)$ --- the sum can not be direct because of the overlap $h$. Thus, $d(h) =0$, i.e.\ $A_h=\{0\}$ which means $h=o$. 

The second statement follows from the same argument. The geometric condition for reducibility $d(g) = \sum_{i=1}^k d(g_i)$ cannot hold if the $g_i$ have common edges.
\end{proof}

Recall that the choice of a building set $\m B$ determines the structure of the exceptional divisor $\m E$ in the wonderful model; the elements of $\m B$ control the number of components of $\m E$ and how they intersect. To construct $Y_{\m B}$ explicitly we needed another family of sub(po)sets of $\m B$, the $\m B$-\textit{nested sets}.

\begin{defn}[Nested sets]
 Let $\m B$ be a building set in a lattice $\m L$. A subset $\m N \subseteq \m B$ is $\m B$-\textit{nested} if for all subsets $\{p_1, \ldots , p_k\} \subseteq \m N$ of pairwise incomparable elements the join (in $\m L$!) $p_1 \vee \cdots \vee p_k$ exists and does not belong to $\m B$.
\end{defn}

 With nested sets we can build another abstract simplicial complex, the \textit{nested set complex} $\Delta_{\m N}(\m L,\m B)$. Its $k$-faces consist of the $\m B$-nested sets with $k+1$ elements. It is the generalization of the order complex for non-maximal building sets. For the maximal building set $\m B=\m L$ a subset is nested if and only if it is linearly ordered in $\m B$, so that in this case we have $\Delta(\m L)=\Delta_{\m N}(\m L,\m B)$. By Theorem \ref{geom of wm} it contains all the information about the stratification of the exceptional divisor $\m E$ in $Y_{\m B}$.

Since $\m D$ is a graded lattice, we have proven here a little conjecture (in the case $G$ at most logarithmic) that appears in many texts on Hopf algebraic renormalization (for example \cite{bk}): 
\begin{cor}
 Every maximal forest of a graph $G$ has the same cardinality. 
\end{cor}
\begin{proof}
 In the language of posets this translates into the fact that every maximal nested set has equal cardinality. But this is equivalent to $\m D$ being graded because the grading map $\tau$ forbids maximal linearly ordered subsets of different length.
\end{proof}

\begin{rem}
This property also seems to hold for the minimal building set $I(\m D)$, but not for intermediate building sets as these are built from $I(\m D)$ by successively adding maximal elements of $\m D \setminus I(\m D)$.  
\end{rem}

 \begin{example} Here are some examples, all in $d=4$:
 
 \begin{enumerate} 
   \item  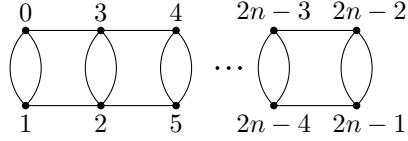
\begin{figure}[h]
 \centering
\begin{tikzpicture}[node distance=1cm and 1cm]
\coordinate[label=above:0] (v1);
\coordinate[below=of v1,label=below:$1$] (v2);
\coordinate[right=of v2,label=below:$2$] (v3);
\coordinate[above=of v3,label=above:$3$] (v4);
\coordinate[right=of v3,label=below:$5$] (v6);
\coordinate[right=of v4,label=above:$4$] (v5);

\coordinate[below right=of v5,xshift=-0.15cm,yshift=0.5cm] (d1);
\coordinate[right=of d1,xshift=-1.15cm] (d2);
\coordinate[right=of d2,xshift=-1.15cm] (d3);

\coordinate[right=of v5,xshift=0.3cm,label=above:$2n-3$] (v7);
\coordinate[below=of v7,label=below:$2n-4$] (v8);
\coordinate[right=of v8,xshift=0.1cm,label=below:$\quad 2n-1$] (v10);
\coordinate[right=of v7,xshift=0.1cm,label=above:$\quad 2n-2$] (v9);
\draw (v2) -- (v3);
\draw (v1) -- (v4);
\draw (v3) to (v6);
\draw (v4) to (v5);
\draw (v7) to (v9);
\draw (v8) to (v10);
\draw (v7) to[out=-135,in=135] (v8);
\draw (v7) to[out=-45,in=45] (v8);
\draw (v9) to[out=-135,in=135] (v10);
\draw (v9) to[out=-45,in=45] (v10);
\draw (v1) to[out=-135,in=135] (v2);
\draw (v1) to[out=-45,in=45] (v2);
\draw (v4) to[out=-135,in=135] (v3);
\draw (v4) to[out=-45,in=45] (v3);
\draw (v5) to[out=-135,in=135] (v6);
\draw (v5) to[out=-45,in=45] (v6);
\fill[black] (v1) circle (.05cm);
\fill[black] (v2) circle (.05cm);
\fill[black] (v3) circle (.05cm);
\fill[black] (v4) circle (.05cm);
\fill[black] (v5) circle (.05cm);
\fill[black] (v6) circle (.05cm);
\fill[black] (d1) circle (.025cm);
\fill[black] (d2) circle (.025cm);
\fill[black] (d3) circle (.025cm);
\fill[black] (v7) circle (.05cm);
\fill[black] (v8) circle (.05cm);
\fill[black] (v9) circle (.05cm);
\fill[black] (v10) circle (.05cm);
\end{tikzpicture}
\caption{The $n$-bubble graph}\label{nbubble}
\end{figure}
Let $G^n$ be the graph in Figure \ref{nbubble}. Here the index $n$ stands for the number of atoms, the \textit{fish} subgraphs on two edges with one cycle, and the numbering of vertices is chosen to match the most ``natural'' choice of an adapted spanning tree $t$ (see Definition \ref{adsptr}). 
Let $g^k_l$ denote the full subgraph of $G^n$ given by the vertex set $V(g^k_l)=\{ 2l \!-\! 2, \ldots , 2l \!-\! 2\! +\! 2k \!-\! 1\}$. From the fact that $\m D(G^{n+1})$ contains two copies of $\m D(G^n)$, given by the intervals $[o , g^n_1]$ and $[o, g^n_2]$, and Lemma \ref{irr} it follows by induction that
 \begin{equation*}
    I(\m D(G^n)) = \{ g^k_l \subseteq G^n \mid k = 1, \ldots, n \ \text{ and } \ l = 1, \ldots , n-k+1 \}.
 \end{equation*}

  \item \begin{figure}[h]
 \centering
\begin{tikzpicture}[node distance=1cm and 1cm]
\coordinate[label=above:$0$] (v1);
\coordinate[below=of v1,label=below:$1$] (v2);
\coordinate[right=of v2,label=below:$3$] (v3);
\coordinate[above=of v3,label=above:$2$] (v4);
\coordinate[right=of v3,label=below:$5$] (v6);
\coordinate[right=of v4,label=above:$4$] (v5);
\coordinate[below right=of v5,xshift=-.05cm,yshift=0.5cm] (d1);
\coordinate[right=of d1,xshift=-1.2cm] (d2);
\coordinate[right=of d2,xshift=-1.2cm] (d3);
\coordinate[right=of v5,xshift=0.3cm,label=above:$n-2$] (v7);
\coordinate[below=of v7,label=below:$n-1$] (v8);
\coordinate[right=of v7,label=above:$n$] (v9);
\draw (v2) -- (v3);
\draw (v1) -- (v4);
\draw (v3) to (v6);
\draw (v4) to (v5);
\draw (v7) to (v9);
\draw (v7) to (v8);
\draw (v1) to[out=-135,in=135] (v2);
\draw (v1) to[out=-45,in=45] (v2);
\draw (v4) to (v3);
\draw (v5) to (v6);
\draw (v2) to (v4);
\draw (v3) to (v5);
\draw (v8) to (v9);
\fill[black] (v1) circle (.05cm);
\fill[black] (v2) circle (.05cm);
\fill[black] (v3) circle (.05cm);
\fill[black] (v4) circle (.05cm);
\fill[black] (v5) circle (.05cm);
\fill[black] (v6) circle (.05cm);
\fill[black] (d1) circle (.025cm);
\fill[black] (d2) circle (.025cm);
\fill[black] (d3) circle (.025cm);
\fill[black] (v7) circle (.05cm);
\fill[black] (v8) circle (.05cm);
\fill[black] (v9) circle (.05cm);
\end{tikzpicture}
\caption{The $n$-insertions graph}\label{nins}
\end{figure}
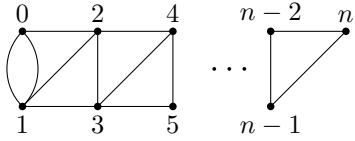
Next we look at the graph $G_n$, depicted in Figure \ref{nins}, constructed by a sequence of $n$ insertions of the fish into itself. Here minimal and maximal building set coincide because all divergent subgraphs are nested into each other:
  \begin{equation*}
  \m D(G_n)=I(\m D(G_n))=\{ g_1, g_2, \ldots , g_n=G_n \}
  \end{equation*}
  where $g_i$ is the full subgraph of $G^n$ corresponding to the vertex set $\{0,\ldots, i \}$. The partial order is a total order. Thus, the $\m D(G_n)$-nested sets are all non-empty subsets of the power set $\mathcal P\left( \m D(G^n) \right) $.
    
  \item  
 Let $G^{n,m}$ be the graph obtained by inserting $n$ bubbles on the left and $m$ bubbles on the right into the fish graph (Figure \ref{nmbubble}). Here
  \begin{equation*}
   I(\m D(G^{n,m}))= \{ g_1, \ldots, g_n , h_1, \ldots, h_m, G^{n,m} \}
  \end{equation*}
  where $g_i$ is the fish subgraph on the vertex set $\{i-1,i\}$ for $i \in \{1, \ldots ,\linebreak n\}$ and $h_j$ is the fish subgraph on the vertex set $\{n+j,n+j+1\}$ for $j\in \{1, \ldots, m\}$.
  All subgraphs in $I(\m D(G^{n,m})) \setminus \{G^{n,m}\}$ have disjoint edge sets. Therefore, as in the previous example, the $I(\m D(G^{n,m}))$-nested sets are all non-empty subsets of $\mathcal{P} \left( I(\m D(G^{n,m}) )\right)$.
  \begin{figure}[h]
  \centering
  \begin{tikzpicture}[node distance=1cm and 1cm]
   \coordinate[label=left:$0$] (v0);
   \fill[black] (v0) circle (.05cm);
   \coordinate[below=of v0,label=left:$1$] (v1);
   \fill[black] (v1) circle (.05cm);
   \coordinate[below=of v1,label=left:$2$] (v2);
   \fill[black] (v2) circle (.05cm);
   
   \coordinate[below=of v2,yshift=0.55cm] (d1);   \fill[black] (d1) circle (.025cm); 
   \coordinate[below=of d1,yshift=.9cm] (d2);   \fill[black] (d2) circle (.025cm);
   \coordinate[below=of d2,yshift=.9cm] (d3); \fill[black] (d3) circle (.025cm);
   
   \coordinate[below=of v2,label=left:$n-1$] (v3);
   \fill[black] (v3) circle (.05cm);
   \coordinate[below=of v3,label=left:$n$] (v4);
   \fill[black] (v4) circle (.05cm);
   
   \draw (v0) to[out=-135,in=135] (v1);
   \draw (v0) to[out=-45,in=45] (v1);
   \draw (v1) to[out=-135,in=135] (v2);
   \draw (v1) to[out=-45,in=45] (v2);   
   \draw (v3) to[out=-135,in=135] (v4);
   \draw (v3) to[out=-45,in=45] (v4);
   
   \coordinate[below right=of v0,xshift=0.5cm,yshift=0.5cm,label=right:$n+m+1$] (v5);
   \fill[black] (v5) circle (.05cm);
   \coordinate[below=of v5,label=right:$n+m$] (v6);
   \fill[black] (v6) circle (.05cm);
   
   \coordinate[below=of v6,yshift=.66cm] (d4);   \fill[black] (d4) circle (.025cm);   
    \coordinate[below=of d4,yshift=.9cm] (d5);   \fill[black] (d5) circle (.025cm);   
   \coordinate[below=of d5,yshift=.9cm] (d6);   \fill[black] (d6) circle (.025cm);    
   
    \coordinate[below=of v6,yshift=.2cm,label=right:$n+2$] (v7);
   \fill[black] (v7) circle (.05cm);
    \coordinate[below=of v7,label=right:$n+1$] (v8);
   \fill[black] (v8) circle (.05cm);
   
   \draw (v5) to[out=-135,in=135] (v6);
   \draw (v5) to[out=-45,in=45] (v6);
     \draw (v7) to[out=-135,in=135] (v8);
   \draw (v7) to[out=-45,in=45] (v8); 
   
   \draw (v0) to (v5);
   \draw (v4) to (v8);
  \end{tikzpicture}
\caption{The $n,m$-bubble graph}\label{nmbubble}
 \end{figure}
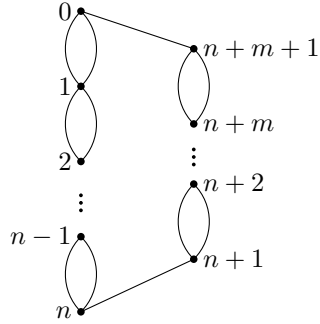
  \item For $n>0$ let $G=K_{n+1}$ be the complete graph on $n+1$ vertices. By induction it follows that saturated subgraphs are either disjoint unions or complete subgraphs on their respective vertex set. Thus, if $n=2$ then $\m G=I(\m G)$. For $n= 3$ (Figure \ref{k4}) the three subgraphs given by the disjoint union of edges $a\cupdot c, b\cupdot d$ and $e\cupdot f$ are reducible while the four embeddings of $K_3$ given by $a\cup b \cup f$ etc.\ are irreducible ($|[o, a\cup b \cup f]|=5$ is not divisible by two).  
  In general, $I(\m G(K_{n+1}))$ consists of all subgraphs that are embeddings of $K_i$ into $K_{n+1}$ for $i=1,\ldots,n$ while the reducible subgraphs are the disjoint unions of embeddings of $K_i$ and $K_j$ for $i+j\leq n+1$. These disjoint unions represent the polydiagonals that make the difference in the blow-up sequence of the Fulton-MacPherson compactification $M[n]$ and Ulyanov's polydiagonal compactification $M\langle n \rangle$. 
  \end{enumerate}
   \begin{figure}[h]
  \centering
  \begin{tikzpicture}[node distance=1.5cm and 1.5cm]
   \coordinate (v1);
\coordinate[below =of v1] (v2);
\coordinate[right =of v2] (v3);
\coordinate[above =of v3] (v4);
\draw (v1) -- node[label=left:$a$] {}  (v2);
\draw (v1) --  (v3);
\draw (v1) -- node[label=above:$d$] {} (v4);
\draw (v2) -- node[label=below:$b$] {}  (v3);
\draw (v2) --  (v4);
\draw (v3) -- node[label=right:$c$] {} (v4);
\coordinate[right =of v2,xshift=-1cm,yshift=1cm, label=$e$];
\coordinate[right =of v2,xshift=-1cm,yshift=-0.05cm, label=$f$];
\fill[black] (v1) circle (.05cm);
\fill[black] (v2) circle (.05cm);
\fill[black] (v3) circle (.05cm);
\fill[black] (v4) circle (.05cm);
  \end{tikzpicture}
\caption{$K_4$}\label{k4}
 \end{figure}
  \end{example}
  
It remains to define the combinatorial version of adapted bases. For this we need \textit{adapted spanning trees}.

\begin{defn}\label{spanning tree}
Let $G$ be a connected graph. A \textit{spanning tree} for $G$ is a simply-connected subgraph $t\subseteq G$ with $V(t)=V(G)$. If $G$ is not connected, $G=G_1\cupdot \cdots \cupdot G_n$, a \textit{spanning n-forest} for $G$ is the disjoint union $t=t_1 \cupdot \cdots \cupdot t_n$ of $n$ spanning trees $t_i$ for $G_i$.
\end{defn}

\begin{defn}[Adapted spanning trees]\label{adsptr}
 Let $G$ be a graph and $\m P \subseteq \m G$ a family of subgraphs of $G$. A spanning tree $t$ of $G$ is \textit{$\m P$-adapted} if for each $g\in \m P$ the graph $t_g$, defined by $E(t_g):=E(t)\cap E(g)$ is a spanning tree for $g$. More precisely, if $g$ is not connected, then we demand $t_g$ to be a spanning forest for $g$.
\end{defn}

\begin{example}
 For dunce's cap an $\m D$-adapted spanning tree ($d=4$) is given by $E(t)=\{e_1,e_3\}$ or $E(t)=\{e_2,e_4\}$, while $t$ with $E(t)=\{e_1,e_2\}$ is spanning but not adapted.
\end{example}

\begin{prop} \label{adsptr exists}
A $\m D$-adapted spanning tree always exists for $G$ at most logarithmic.
\end{prop}
\begin{proof}
 We construct $t$ using the fact that divergent graphs can be built from primitive ones using the insertion operation. Moreover, this process is reversible, i.e.\ in the dual process of contracting subgraphs no information is lost.  
 Start with the primitive subgraphs of $G$ and let $G_1$ be the graph obtained from $G$ by contracting all primitive subgraphs. $G_1$ might have primitive subgraphs itself (the $g\in \m D$ with \textit{coradical degree} equal to two, cf.\ \cite{dk}). Repeat the process. After a finite number of steps $G_k$ will be free of subdivergences. Now choose a spanning tree $t_1$ for $G_k$ and spanning trees for all subgraphs contracted in the step from $G_{k-1}$ to $G_k$. Then $t_2$, the union of all these spanning trees, is a tree in $G_{k-1}$ visiting every vertex exactly once. Thus, it is a spanning tree for $G_{k-1}$. Repeat this process until after $k$ steps we have an $\m D$-adapted spanning tree $t=t_k$ of $G$.
\end{proof}

\begin{rem}
 An interesting question arising here is for which families $\mathcal P \subseteq \mathcal G$ does such a $\mathcal P$-adapted spanning tree exist? For a counterexample just take $\mathcal P = \mathcal G$ or $I(\m G)$: In the first case every edge of $G$ lies in $\m G$, so there cannot exist a $\m G$-adapted spanning tree. For the second case consider the example $K_4$; there is no spanning tree that generates all four irreducible ``triangle'' subgraphs. 
 
Another question is for which class of graphs this holds, i.e.\ if the assumption of $G$ being at most logarithmic can be dropped?
\end{rem}

With adapted spanning trees we are able to define $\m N$-adapted bases of $(X^G)^*$ in combinatorial terms. If the divergent lattice is considered, a $\m D$-adapted spanning tree will automatically be $\m N$-adapted for any nested set of any building set in $\m D$. This allows us to fix a convenient basis from the beginning on. Every spanning tree $t$ of $G$ has $|V|-1$ edges (otherwise it would contain a cycle, contradicting simply-connectedness). Therefore, for every spanning tree $t$ of $G$ (with the same orientation) we have a linear map $\psi_t: M^{E(t)}\rightarrow M^{V'}$ defined by 
\begin{equation} \label{adbas}
 e \mapsto \begin{cases}
            v_j-v_i & \text{ if $e$ starts at $v_i$ and ends $v_j$, } \\
            \pm v_k & \text{ if $e$ connects $v_0$ to $v_k$.}
           \end{cases}
\end{equation}
Pulling back $v_G$ along $\psi_t$ amounts to a linear change of coordinates on $X^G$ (as well as altering the numbering of the vertices of $G$, its orientation or the choice of a different (adapted) spanning tree). Any automorphism of $X^G$ will not change the topology of the arrangement and, as is shown in \cite{dp}, induces an isomorphism on the corresponding wonderful models. Therefore the wonderful construction and renormalization do not depend on these choices and we can work in a convenient basis given by an adapted spanning tree. 
In this basis $v_G$ is given by
\begin{equation*}
 v_G(\{x_e\}_{e\in E(t)})=\prod_{e\in E(t)}\triangle(x_e)\prod_{e\in E(G\setminus t)}\triangle\left(\sum_{e'\in E(t_e)}\sigma_t(e')x_{e'}\right)
\end{equation*}
where $t_e$ is the unique path in $t$ connecting the source and target vertices of $e$ and $\sigma_t:E(t)\to \{-1,+1\}$ is determined by the chosen orientation of $G$. The point is that for the divergent poset $\m D$ in these coordinates $x=\sum_{e\in E(t)}x_e e$ we have $A^{\bot}_g=\{ x_e=0 \mid e\in E(t_g)\}$ for all $g\in \m D$. Dually this means that the elements in $B |_{e \in E(t_g)}$, defined below, form a basis of $A_g$. In other words, we have an adapted basis in the sense of DeConcini-Procesi! 
By duality $B$ also defines a basis of $X^G$ and by abuse of notation we will denote both bases by $B$ - the meaning should always be clear from the context. This choice of basis will be important when we study the pullback of $v_G$ onto the wonderful model in the next section. 

For $\m G$ and other lattices there need not be an adapted spanning tree, but we can always find $\m N$-adapted spanning trees and bases for any nested set $\m N$.
\begin{prop}
 Let $\m N$ be nested for some building set $\m B$ in some lattice $\m L\subseteq \m G$. Then there exists an $\m N$-adapted spanning tree.
\end{prop}

\begin{proof}
 The idea is the same as in Proposition \ref{adsptr exists}. Start with the set $\m M$ of maximal elements in $\m N$ and contract all other elements. Pick a spanning tree for the resulting graphs. Proceed in the same manner with $\m N \setminus \m M$ and repeat the process until all of $\m N$ has been exhausted. This produces a spanning forest $t$ for $\cup_{\gamma \in \m N} \gamma $, except if there are $g,h$ in $\m N$ that are non-comparable and have non-empty intersection. In this case we argue like in the proof of Lemma \ref{irr} to see that the join $g \vee h$ must also be in $\m B$. But this is impossible since $\m N$ is $\m B$-nested.
 In a last step contract all elements of $\m N$ in $G$ and pick a spanning tree $t'$ for the resulting graph. The union $t\cup t'$ is then an $\m N$-adapted spanning tree for $G$.   
\end{proof}

\begin{defn}
 Let $G$ be at most logarithmic and $\m N$ a $\m B$-nested set for some building set $\m B$ in a lattice $\m L \subseteq \mathcal G$. Given an $\m N$-adapted spanning tree $t$ define the map $\psi_t$ as in \eqref{adbas}. With the linear forms $\omega_e$ introduced in Section \ref{graphs to arr} we define an $\mathcal N$-\textit{adapted} basis of $(X^G)^*$ by
 \begin{equation*}
  B:=\{ b^i_e:=(\omega_e \circ \psi_t )^i \mid e\in E(t), \ i=1, \ldots , d \}.
 \end{equation*}
 In such a basis the map $p:(X^G)^*  \rightarrow \m N \cup \{G\}$ from Definition \ref{adapted basis 1} is then given by 
 \begin{align*}
  p:  x= \sum_{ \substack{ e\in E(t) \\ i=1,\ldots, d} } x_e^i b_e^i  \longmapsto \min\{g \in \m N \cup \{G\} \mid x_e^i=0 \ \text{ for all }\  e \in E(t \setminus t_g) \}.
 \end{align*} 
 A \textit{marking} of an adapted basis is for every $g\in \m N$ the choice of a $b_g^{i_g} \in B$ with $p(b_g^{i_g})=g$. Equivalently, we can view it as a labelling on the elements of $\m N$: 
 \begin{equation*}
   g \longmapsto b_e^j \quad\text{for $e\in E(t\cap (g\setminus \m N_{<g}))$ and some $j \in \{1,\ldots,d\}$}.
 \end{equation*}
 Here $g\setminus \m N_{<g}:= g \setminus (h_1 \cup \cdots \cup h_k)$ for $\{h_1,\ldots, h_k\}=\{ h \in \m N \mid h < g \}$ de\-notes the graph $g$ with all its lower bounds in $\m N$ removed.
The partial order $\preceq$ on $B$ that determines the local blow-up $\rho_{\m N,B}$ is given by 
 \begin{equation*}
  b_e^i \preceq b_{e'}^j \Longleftrightarrow  e \in E(t_g), e' \in E(t_{g'}) \ \text{ with } g\subseteq g'\text{ and }  b_{e'}^j \text{ is marked.}
 \end{equation*}
 \end{defn}
 This finishes all necessary definitions and from here on we could repeat the construction of a wonderful model in purely combinatorial terms. In the divergent case we thus conclude that all ingredients are already determined by the topology and subgraph structure of $G$. Therefore, there is really no need for purely geometric data to build an atlas for $Y_{\m B}$. However obtained, now after the planting has been done, it is time to reap the fruits and see what a wonderful model can do for us.

\section{Wonderful renormalization}

Having constructed the wonderful models $(Y_{\m A},\beta)$ for general arrangements $\m A$, we now focus on the divergent and singular arrangements $\m A= \m A_{\m D}, \m A_{\m G}$ of a connected and at most logarithmic graph $G$. We study the pullback of $v_G$ onto the model and the pole structure of its Laurent expansion, then define (local) renormalization operators. The first two sections follow the exposition in \cite{bbk}, especially the proofs of Proposition \ref{pullback} and Theorem \ref{laurent}. The difference lies in the emphasis on the combinatorics of $\m D$ and the role of adapted spanning trees in our formulation. We correct some minor flaws and fill out missing details in the proofs.

Since $Y_{\m A}$ is non-orientable, we use from now on distribution densities (cf.\ Section \ref{distributions}).
As charts for $Y_{\m A}$ are indexed by nested sets $\m N$ and markings of adapted bases $B$ we will abbreviate this data by an underlined letter, $\underline{i}=(\m N,B)$. Adapted bases are here always given by the choice of an adapted spanning tree $t$.
We write $x=\{ x_e \}_{e \in E(t)}$ for a point in $X=M^{E(t)}$ with $x_e=(x_e^1,\ldots,x_e^d)$. A marking of $B$ assigns individual coordinates to the elements of a nested set of graphs. We denote the marked elements by $\mathcal N \ni g \mapsto x_g^{i_g}$. If a vector $x_g$ is marked in this way, let $\hat x_g$ denote $x_g$ with $i_g$-th coordinate equal to 1. 

\subsection{The pullback of $v_G$ onto the wonderful model}\label{pullback}

Let $(Y,\beta)$ be a wonderful model for $\m G$ or $\m D$ and $v=v_G$ the Feynman distribution associated to a graph $G$. We start the renormalization program by disassembling the pullback of $v$ onto $Y$ into a regular and a singular part.
\begin{prop}\label{pullback}
 Let $\m N$ be $\m B$-nested for a building set $\m B$ of $\m D$ (or $\m G$) and $B$ an adapted, marked basis. In local coordinates on $U_{\un i}$, $\un i=(\m N,B)$, the pullback of $\tilde v^s:= v^s|dx|$ onto the wonderful model is given by
 \begin{equation}
  \tilde w^s_{\un i}:=(\beta^*\tilde v^s)_{\un i}=f^s_{\un i}\prod_{g\in \m N} u_g^{-1 + d_g + s(2-d)|E(g)|} |dx|
 \end{equation}
 where $u_g(x^{i_g}_g)=|x^{i_g}_g|^{-1}$ and $d_g:=\dim A_g = d(|E(g)|-h_1(g))$. 
 
 The map $f_{\un i}:\kappa_{\un i}( U_{\un i} ) \longrightarrow \mathbb{R}$ is in $L^1_{\text{loc}}(\kappa_{\un i} ( U_{\un i}) )$ (or in $\mathcal{C}^{\infty}(\kappa_{\un i}( U_{\un i})) $ if the singular arrangement is considered) but smooth in the variables $x^{i_g}_g$, $g\in \m N$. 
\end{prop}

\begin{proof}
The crucial point here is that locally $\beta$ is given by the map 
 \begin{align*}
  \rho_{\un i}  :X & \longrightarrow X, \\
  \sum_{i=1}^d \sum_{e\in E(t)} x^i_e b^i_e & \mapsto \sum_{i=1}^d\sum_{e\in E(t)}\prod_{x^i_e \preceq x^k_{e'}}x_{e'}^kb_e^i.
 \end{align*}
 Recall the choice of coordinates given by $t$ in \eqref{adbas}. In these coordinates 
\begin{align*}
 \beta^*v^s(x)&=v^s(\rho_{\un i}(x))=\left( \rho_{\un i}^* \left( I^* \triangle^{s \otimes |E(G)|} \right) \right)(x) \\
  & = \rho_{\un i}^*\left( \prod_{e\in E(t)}\triangle^s(x_e)\prod_{e\in E(G\setminus t)} \triangle^s\left(\sum_{e' \in E(t_e)} \sigma_t(e')x_{e'}\right) \right) \\
  & = \!\prod_{e\in E(t)}\!\triangle^s\left( \prod_{p(x_e)\subseteq g}x^{i_g}_g \hat x_e\!\right)\prod_{e\in E(G\setminus t)}\! \triangle^s\left(\sum_{e' \in E(t_e)} \prod_{p(x_{e'})\subseteq g}x^{i_g}_g \sigma_t(e')\hat x_{e'}\!\right) 
\end{align*}
where $\hat x_e:=(x_e^1, \ldots , \overbrace{1}^{i_g} , \ldots , x^d_e)$ if $x_e$ has a marked component. Since $\triangle$ is homogeneous of degree $(2-d)$, we can pull out all the factors $x^{i_g}_g$ in the first product of $\triangle$'s, so that the kernel of $\ti w_{\un i}^s(x)$ is given by
\begin{align*}
\prod_{e\in E(t)} \left(\prod_{p(x_e)\subseteq g}x^{i_g}_g\right)^{s(2-d)} \triangle^s(\hat x_e)\prod_{e\in E(G\setminus t)} \triangle^s\left(\sum_{e' \in E(t_e)} \prod_{p(x_{e'})\subseteq g}x^{i_g}_g  \sigma_t(e') \hat x_{e'}\right).
\end{align*}
In the second factor we can pull out $x_g^{i_g}$ if it appears in every term in the sum, i.e.\ if $t_e$ is a subgraph of some $g\in \m N$. But this is equivalent to $e \in E(g)$ because $t$ is an adapted spanning tree. Thus, $x_g^{i_g}$ appears exactly $|E(g)|$-times and we conclude  
\begin{equation*}
 w_{\un i}^s(x)=f_{\un i}^s(x)\prod_{g\in \m N}(x_g^{i_g})^{s(2-d)|E(g)|}.
\end{equation*}
Under the coordinate transformation $\rho_{\un i}$ every $dx_e=\bigwedge_{i=1,\ldots, d} dx_e^{i}$ transforms into 
\begin{equation*}
  \rho_{\un i}^*dx_e= \begin{cases} (x_g^{i_g})^{d-1} dx_e &\text{ if $x_e$ contains a marked component} , \\
  (x_g^{i_g})^d dx_e &\text{ if $x_e$ has no marked component}.
 \end{cases}
\end{equation*}
How many $x_e$ are scaled by the same $x_g^{i_g}$? As many as there are edges in $E(g)$. Therefore, there are in total $(\dim A_g - 1)$ factors and the measure 
\begin{equation*}
 |dx|=\left|\bigwedge_{e \in E(t), i=1,\ldots,d } dx^{i}_e\right| 
\end{equation*}
transforms into  
\begin{equation*}
 \rho_{\un i}^*|dx|=\prod_{g\in \m N}|x_g^{i_g}|^{\dim A_g - 1}|dx|.
\end{equation*}
Putting everything together we conclude
\begin{equation*}
 \tilde w_{\un i}^s(x)= f_{\un i}^s(x)\prod_{g\in N}|x_g^{i_g}|^{-1 + d_g + s(2-d)|E(g)|}|dx|.
\end{equation*}
For the divergent lattice the exponents of $|x_g^{i_g}|$ are given by $- 1 +d_g (s-1)$ because $d|h_1(g)|=2|E(g)|$ and $d_g=\dim A_g=d(|E(g)|-h_1(g))$ (cf.\ the proof of Proposition \ref{sing supp}). 

It remains to show that $f_{\un i} \in L^1_{\text{loc}}(\kappa_{\un i} ( U_{\un i} ) )$ or $\mathcal{C}^{\infty}( \kappa_{\un i} ( U_{\un i}) )$, respectively. Recall the definition of $U_{\un i}= X \setminus \cup_{\gamma \in \m B} Z_{\gamma}$ where $Z_{\gamma}$ is the vanishing locus of the polynomials $P_v$ for $v\in X^*$ such that $p(v) = \gamma$ (note that $p$, $P_v$ and therefore also $Z_\gamma$ depend on $\un i$!). For the singular arrangement every building set $\m B$ must contain all subgraphs consisting of a single edge. But for these elements of $\m B$ the $Z_{\gamma}=Z_e$ are precisely the sets where an entire sum $\sum_{e' \in E(t_e)} \sigma_t(e')x_{e'}$ expressing an edge $e$ of $G$ vanishes. Since all functions $\triangle$ are smooth off the origin it follows that $f_{\un i}$ is a smooth function. The same reasoning works for the divergent arrangement: Every building set of $\m D$ must contain all irreducible subgraphs. In addition, every element of $\m D$ is built out of elements of $I(\m D)$ 
by the join operation (i.e.\ using $\cup$). Therefore, as in the singular case, it follows that linear 
combinations expressing edges in any divergent subgraph can not vanish on $U_{\un i}$. The map $f_{\un i}$ fails to be smooth only at propagators of edges that do not lie in some element of $\m B$. But by the proof of Proposition \ref{sing supp} we know that there $f_{\un i}$ is still locally integrable, hence $f_{\un i} \in L^1_{\text{loc}}( \kappa_{\un i}(U_{\un i}) )$. 
Smoothness in the marked elements $x_g^{i_g}$, $g\in \m N$, follows from the simple fact that in the definition of $f_{\un i}$ already all marked elements have been pulled out of the linear combinations expressing edges in $G$. If one such expression would vanish at $x_g^{i_g}=0$, then all $x_{e'}$ were scaled by $x_g^{i_g}$ and this factor would have been absorbed into the exponent of $u_g$. Therefore no argument in the product of $\triangle$'s can vanish at $x_g^{i_g}=0$.
\end{proof}

\subsection{Laurent expansion}\label{laurent expansion}

From now on we consider the divergent lattice $\m D$ only. In this case we define $u^s_g(x^{i_g}_g):=|x^{i_g}_g|^{- 1 + d_g(1-s)}$ and for a finite product of maps $F_i$, $i\in I$, we write $F_I:=\prod_{i\in I}F_i$. Then, under the assumptions of Proposition \ref{pullback},
\begin{equation*}
 \tilde w_{\un i}^s = f_{\un i}^s \prod_{g\in \m N} u^s_g |dx|= f_{\un i}^s u_{\m N}^s |dx|.
\end{equation*}
To define local renormalization operators we need a better understanding of the pole structure of $\tilde w^s$. As it turns out, this structure is already encoded in the geometry of the exceptional divisor $\m E$ and reflects the structure of the divergent lattice $\m D$. 

Consider first the case of primitive graphs. Then $Y$ is the blow-up of the origin in $X$, covered by charts $U_i$ where $i$ runs from 1 to $dn$ (corresponding to all possible markings of an adapted basis). We already know from the extension theory for distributions that the Laurent expansion around $s=1$ of $\ti w^s$ has a simple pole with its residue given by
\begin{equation*}
 \ti a_{-1}\overset{loc.}{=}-\frac{2}{d_G} f_i \delta_{\m E i} .
\end{equation*} 
Here $\delta_{\m E}$ is a density on $Y$, the delta distribution \textit{centered on} $\m E$ (cf. \cite{gs}), locally in $U_i$ given by the delta distribution in the marked coordinate $x^i$, i.e.
\begin{equation*}
 \langle \delta_{\m E} | \varphi \rangle \overset{loc.}{=} \int dx \, \delta(x_i) \varphi(x).
\end{equation*}
 Pairing $\ti a_{-1}$ with the characteristic function $\chi$ of $Y$ produces a projective integral
\begin{equation*}
 \left\langle \ti a_{-1} | \chi \right\rangle = -\frac{2}{d_G}\int_{\m E}  f.
\end{equation*}
Recall from Section \ref{blow-ups} the definition of induced charts $(V_i,\phi_i)$ for $\m E$. Since any such chart covers $\m E$ up to a set of measure zero, it suffices to do this integral in one of them. Thus, 
\begin{equation*}
 \int_{\m E}  f = \int d\hat x f_i(\hat x),
\end{equation*}
where $d\hat x=dx^1\wedge \cdots \wedge \widehat{dx^i} \wedge \cdots \wedge dx^{nd}$ for some $i \in \set{dn}$.

\begin{defn}[Period of a primitive graph]\label{period}
Let $G$ be primitive. The \textit{period} $\mathscr P(G)$ of $G$ is defined as the integral
\begin{equation*}
\mathscr P(G) := \left\langle \ti a_{-1} | \chi \right\rangle = -\frac{2}{d_G}\int_{\m E} f.
\end{equation*} 
\end{defn}
For more on periods see the overview in \cite{os}. Until recently it was believed that all periods in massless $\phi^4$-theory (i.e. $d=4$ and all vertices of the Feynman diagram corresponding to $G$ are 4-valent) are rational combinations of multiple zeta values. But counterexamples \cite{bd} have proven this false, relating a better understanding of these periods to deep questions in algebraic geometry \cite{fb}.  

For general $G$ the Laurent expansion of $\ti w^s$ will contain terms corresponding to contracted graphs. Since we work in local coordinates indexed by $\m B$-nested sets, we need a more sophisticated (local) contraction operation on graphs: 
\begin{defn}
Let $g\subseteq G$ and $\m P \subseteq \m D$. The \textit{contraction relative to $\m P$} is defined as 
\begin{equation*}
 g \ds {\m P}:= \begin{cases}
                 g / (\bigcup_{\gamma \in \m P_{<g}} \gamma ) & \text{ if } g \in \m P, \\
                 g / (g \cap \bigcup_{\gamma \in \m P: \gamma \cap g < g} \gamma ) & \text{ else}.
               \end{cases}
\end{equation*}
\end{defn}
Especially important will be the contraction relative to nested sets. The reader should think of it as a local version of the contraction in the definition of the coproduct in the Hopf algebra of Feynman graphs. It will show up in all formulae that include the coproduct in their usual formulation, say in momentum space. Note that for $g \subseteq G$ the ``normal'' contraction $G/g$ is included in this definition as contraction of $G$ with respect to the nested set $\m N=\{g, G\}$. Moreover, if $\m N$ is nested and $g\in \m N$ or all elements of $\m N$ are contained in $g$, then $g\ds \m N$ is at most logarithmic as well. For a general discussion of which classes of graphs are closed under the contraction operation we refer the reader to \cite{bk}. 

\begin{thm}\label{laurent}
 Let $Y$ be a wonderful model for some building set $\m B$ of $\m D$. Let $\tilde w^s = \beta^* \ti v^s$ be the pullback of the density $\tilde v^s \in  \mathcal{\tilde D'}(X)$ onto $Y$. Then:
 \begin{enumerate}
  \item The Laurent expansion of $\tilde w^s$ at $s=1$ has a pole of order $N$ where $N$ is the cardinality of the largest $\m B$-nested set.
  \item The coefficients $\tilde a_k$ in the principal part of the Laurent expansion, 
  \begin{equation*}
   \tilde w^s= \sum_{-N \leq k \leq -1} \tilde a_k (s-1)^k,
  \end{equation*}
 are densities with $\text{supp } \tilde a_k = \bigcup_{|\m N|=-k} \mathcal{\m E}_{\m N}$, where $\m E_{\m N}:= \bigcap_{\gamma \in \m N}\m E_{\gamma}.$
\item Consider the irreducible elements $I(\mathcal{D})$ as building set. Assume $G\in I(\m D)$. Let $\m N$ be a maximal nested set and denote by $\chi$ the constant function on the wonderful model $Y_{I(\mathcal{D})}$. Then for $N=|\m N|$

\begin{equation*}
 \langle \tilde a_{-N} | \chi \rangle = \sum_{|\m M|=N} \prod_{\gamma \in \m M} \m P(\gamma \ds \m M).
\end{equation*}
 \end{enumerate}
\end{thm}

\begin{proof}
1. This follows from the local expression for $\ti w^s$. Using Formula \eqref{distro split} we have
   \begin{equation*}
   \tilde w^s \overset{loc.}{=} f_{\un i}^s u^s_{\m N} |dx| = f_{\un i}^s \prod_{g\in \m N} \left( -\frac{2}{d_g} \delta_g (s-1)^{-1} + u_{g\heartsuit}^s \right) |dx| .
  \end{equation*}
 Since $u_{g\heartsuit}^s$ is regular in $s$, the highest pole order is given by $|\m N|$.

\medskip
2. Expand $ u_{\heartsuit}^s \in  \m {D}'(\mb R)$ into a Taylor series at $s=1$,
\begin{equation*}
 u_{\heartsuit}^s = \sum_{k=0}^{\infty}  u_k (s-1)^k
\end{equation*}
where the distributions $u_k$ are given by
\begin{equation}\label{expansion}
u_k: \varphi \mapsto \int dx \, |x|^{-1}\log^k(|x|) \big(  \varphi(x) - \theta(1-|x|) \varphi(0) \big).
\end{equation}
In the following we write $\theta_g$ for the map $x_g^{i_g} \mapsto \theta \big( 1- |x_g^{i_g}| \big)$ in all coefficients of the expansion of the regular part of $u_g^s$. Expanding $f^s$ gives 
\begin{equation*}
f^s= \exp( \log(f^s))=f \exp \big( (s-1)\log(f) \big)=f\sum_{k=0}^{\infty} \frac{ \log^k(f) }{ k! } (s-1)^k.
\end{equation*}
Fix a $\m B$-nested set $\m N$ with a corresponding marking $B$ and set $\un i=(\m N,B)$. To determine the lower pole parts in the local expression for $\ti w^s$ we multiply all series $u_{g\heartsuit}^s$ for $g \in \m N$ and reorder the sum. Denote by $(u_g)_l$ the $l$-th order coefficient of the expansion of $u_g^s$. Then for $n \in \{ 1, \ldots,  N \}$ the kernel of $\ti a_{-n}$ is given by
\begin{align}\label{ugly formula}
& \sum_{k=0}^{N-n-1} \frac{f_{\un i}\log^k(f_{\un i})}{k!} \left(\! \sum_{j=k}^{N-n-1}\!\! \sum_{ \substack{ \m L \subseteq \m N \\  |\m L|=n+j}} \prod_{\gamma \in \m L} \left(-\frac{2}{d_{\gamma}}\right)\delta_{\gamma \un i} \!\!\!\! \prod_{\substack{ \eta \in \m N \setminus \m L \\ \{ l_{\eta} \in \mb N \mid \sum_{\eta\in \m N\setminus \m L} l_{\eta}=j-k \} }}\!\!\!\! (u_{\eta})_{l_{\eta}}\! \right)\hspace{-1em}  \\
& + \frac{ f_{\un i}\log^{N-n}(f_{\un i}) }{(N-n)!} \prod_{\gamma \in \m N} \left(-\frac{2}{d_{\gamma}}\right)\delta_{\gamma \un i}, \nonumber
\end{align}
with $\delta_\gamma := \delta_{\m E_{\gamma}}$. Recall that locally $\m E_g$ is given by $x_g^{i_g}=0$ and $\m E_\m I\subseteq \m E_\m J$ for $\m J \subseteq \m I \subseteq \m N$. Therefore, the support of $\ti a_{-n}$ is given by the ($k=j=0$)-summand in~\eqref{ugly formula}, carrying the product of $n$ $\delta$-distributions in the marked coordinates of an $n$-element subset of $\m N$. Varying over all $\m B$-nested sets $\m N$ (and the markings) the same holds for all $n$-element subsets of any nested set. Thus, from the expansion formula \eqref{ugly formula} we conclude that the densities $\tilde a_{-n}$ are supported on
 \begin{equation*}
  \bigcup_{ |\m N|=n} \left( \bigcap_{\gamma \in \m N} \m E_{\gamma} \right)=\bigcup_{|\m N|=n} \m E_\m N.
 \end{equation*}

3. This follows essentially from two assertions: 
First, if we view the pairing of a product of delta distributions ($\delta_g\overset{loc.}{=}\delta ( x^{i_g}_g )$) with a function $\varphi$ as an operator $\delta_{\m N}$, locally given by
\begin{equation*}
 \delta_{\m N \un i }: \varphi \in \m D(\kappa_{\un i} ( U_{\un i} ) ) \longmapsto \left(\prod_{\gamma \in \m N}\delta_{\gamma \un i}\right) [\varphi] \in \m D( \kappa_{\un i} ( U_{\un i} \cap \m E_{\m N} ) ),
\end{equation*}
then for $f=f_{\un i}$ the regular part of the pullback $\beta^* \ti v^s$ we have

\begin{equation}
\delta_{\m N}[f]=\prod_{\gamma \in \m N}f_{ \gamma \ds \m N}.
\end{equation}
 Here $f_{g\ds \m N}$ is obtained from $f$ by setting all marked elements corresponding to graphs in $\m N_{< g}$ to zero. It equals the regular part of the pullback of $\ti v^s_{g\ds \m N}$ onto the wonderful model for the graph $g\ds \m N$ in charts corresponding to the nested set $\{g\ds \m N\}$ ($g \ds \m N$ is primitive!). For a precise definition and the proof of this assertion we refer to Section \ref{renormalization group}, Theorem \ref{rg thm}, where this is elaborated in a much more general case. The important point here is that $\delta_{\m N}[f]$ is a product of maps $f_{g\ds \m N}$, each one depending only on the set of variables $\{x_e\}$ with $e$ in $E(t_g)\setminus E(\m N_{<g})$ without all marked elements.

The second assertion is that in every maximal $I(\m D)$-nested set all contracted graphs $g \ds \m N$ are primitive. Note that if $G$ is divergent and irreducible, it must be contained in every maximal nested set. To prove the assertion let $g\in \m N$ and assume $g\ds \m N$ is not primitive. This means there is an $h\in \m D$ with either $h \subseteq g\ds \m N$ or $h\ds \m N \subseteq g\ds \m N$. In both cases we can assume that $h$ is irreducible (if not, then $h$ is the union of irreducible elements and we do the following for every irreducible component of $h$). Then the set $\m N'=\m N \cup \{h\}$ is also nested if $h$ satisfies the following property: For all $g'$ in $\m N$ that are incomparable to $g$ the join $h\vee g'=h \cup g'$ must not lie in $I(\m D)$. But if there is $g'\in \m N$, incomparable to $g$, with $h\leq g'$ then $g$ and $g'$ have both $h$ as common subgraph. By Lemma \ref{irr} this implies that $g\vee g'$ is irreducible, showing that $g$ and $g'$ cannot both lie in $\m N$ because $\m 
N$ 
is $I(\m D)$-nested. If $g\ds \m N'$ is still not primitive, repeat the process until all contracted graphs are primitive. The resulting nested set $\m N'$ is then really maximal; if adding another graph does not violate the property of being nested, then it must necessarily be disjoint from all $g \in \m N'$ (otherwise some $g\ds \m N$ is not primitive) and this is impossible due to $G$ being in $\m N'$.

 For $G\in I(\m D)$ all edges of $G$ lie in some divergent subgraph (if not, say for one edge $e$, then contract all divergent subgraphs. The resulting graph is primitively divergent and contains $e$). Thus, in every maximal nested set $\m N$ all edges of an adapted spanning tree $t$ correspond to some element of $\m N$ and by Definition \ref{wonderful definition 2} we have $U_{\m N,B}=X$ for all maximal nested sets $\m N$. Let $\m N$ be such a maximal nested set. Using 2.\ and the first assertion we have locally in $U_{\un i}=U_{\m N,B}$
\begin{align*}
 \langle \ti a_{-N} | \chi \rangle & \overset{loc.}{=} \int_{\kappa_{\un i} ( U_{\un i} ) } dx  \prod_{\gamma \in \m N} \left(-\frac{2}{d_{\gamma}}\right)\delta_{\gamma \un i } [f_{\un i}] =  \prod_{\gamma \in \m N} \left(-\frac{2}{d_{\gamma}}\right)\int_{\kappa_{\un i} ( V_{\un i} )} d\hat x\, \delta_{ \m N {\un i} }[f_{\un i}] \\
 & = \int_{\kappa_{\un i} (V_{\un i} )} d\hat x \prod_{\gamma \in \m N}  \left(-\frac{2}{d_{\gamma}}\right) (f_{\gamma \ds \m N})_{\un i}.
 \end{align*}
 Here $\hat x$ denotes $\{x_e\}_{e\in E(t)}$ without all marked elements and $V_{\un i}$ is the chart domain for local coordinates on $\m E_{\m N}$, obtained by restriction of the chart $\kappa_{\un i}$ (cf.\ Section \ref{blow-ups}).
 Since it covers $\m E_{\m N}$ up to a set of measure zero (cf.\ Definition \ref{period}), integration in a single chart suffices. Moreover, two components of the exceptional divisor $\m E_{\m N}$ and $\m E_{\m M}$ have non-empty intersection if and only if $\m N \cup \m M$ is nested. But this is impossible due to maximality of $\m N$. Therefore we can sum the contributions from charts given by different maximal nested sets to obtain the global result 
 \begin{equation*}
  \langle \ti a_{-N} | \chi \rangle  = \sum_{ \un i } \int_{\kappa_{\un i}(V_{\un i})} d\hat x \prod_{\gamma \in \m N} \left(-\frac{2}{d_{\gamma}}\right) (f_{\gamma \ds \m N})_{\un i}
 \end{equation*}
where the sum is over all maximal nested sets with some marking. 
Since all $(f_{\gamma \ds \m N})_{\un i}$ depend on mutually disjoint sets of variables, the integral factorizes and since restricting $\kappa_{\un i}|_{V_{\un i}}$ further to $\{ \hat x_e \}_{e\in E(t_{\gamma})}$ is a local chart for $\m E_{\gamma \ds \m N}$, we conclude that
 \begin{align*}
 \sum_{\un i} \int_{\kappa_{\un i}(V_{\un i})} d\hat x \prod_{\gamma \in \m N} \left(-\frac{2}{d_{\gamma}}\right) (f_{\gamma \ds \m N})_{\un i} & = \sum_{\m N} \prod_{\gamma \in \m N} \int_{ \m E_{\gamma \ds \m N} } \left(-\frac{2}{d_{\gamma}}\right) f_{\gamma \ds \m N} \\
  &= \sum_{\m N} \prod_{\gamma \in \m N} \mathscr P (\gamma \ds \m N ).\\[-3.5em]
 \end{align*}
\end{proof}
This theorem is a first hint at the Hopf algebraic formulation of the renormalization group (see \cite{dk}, \cite{ck}). It shows that the poles of $\ti w^s$ are not arbitrary densities but reflect the combinatorics of $\m D$ in a special way. The highest order pole is completely determined by the structure of $\m D$. For the poles of lower order the same holds in a weaker version; they are supported on components of $\m E$ whose stratification is given by the combinatorial structure of $\m D$ as well.

\subsection{Renormalization}\label{renormalization}

With the main result of the previous section we are now able to tackle the renormalization problem. Since all poles of $\tilde w^s$ live on the components of the exceptional divisor, we can get rid of them using local subtractions depending on the direction such a pole is approached. These directions are encoded by nested sets, so that we will employ local versions of the previously defined renormalization maps $r_1$ and $r_{\nu}$, depending on the chosen coordinate system given by $\mathcal B$-nested sets and markings of an adapted basis $B$. 

\begin{defn}[(Local) minimal subtraction]\label{loc ms}
Let $R_1$ denote the collection of renormalization maps $\{ R^{\un i}_1\}$ where $\un i$ runs through all $\m B$-nested sets $\m N$ and marked, adapted bases $B$ (more precisely, the markings since the basis is fixed). $R^{\un i}_1$ removes the poles in the coordinates associated to the marked elements, i.e.
\begin{equation*}
 R_1[\ti w^s]\overset{loc.}{=}R^{\un i}_1[f_{\un i}^s \ti u^s_{\m N}]:= f_{\un i}^s \prod_{g \in \m N}r_1[u_g^s] |dx|.
\end{equation*}
\end{defn}
Recall from Section \ref{distributions} that $r_1[u_g^s]=(u_g^s)_{\heartsuit}$, so there are no poles anymore and we can take the limit $s\to 1$ to obtain a well defined density on $Y$. 
The next definition introduces a renormalization operator that produces a density for $s$ in a complex neighborhood of $1$. It should be thought of as a smooth version of minimal subtraction. 

\begin{defn} [(Local) subtraction at fixed conditions]\label{loc fs}
Let $R_{\nu}$ denote the collection of renormalization maps $\{ R^{\un i}_{\nu}\}$ where $\un i$ runs through all $\m B$-nested sets $\m N$ and markings of $B$. The symbol $\nu=\{\nu^{\un i}_g\}_{g\in \m N}$ stands for a collection of smooth functions on $\kappa(U_{\un i})$. Each $\nu^{\un i}_g$ depends only on the coordinates $x_e$ with $e\in E(t)\cap E(g \setminus \m N_{<g})$ and satisfies 
$\nu^{\un i}_g|_{x^{i_g}_g=0}=1$. Furthermore, it is compactly supported in all other directions. Similarly to $R_1$ the operator $R_{\nu}$ is defined by
\begin{equation*}
 R_{\nu}[\ti w^s]\overset{loc.}{=}R^{\un i}_{\nu}[f_{\un i}^s \ti u_{\m N}^s]:= f_{\un i}^s \prod_{g \in \m N}r_{\nu^{\un i}_g}[u_g^s] |dx|.
 \end{equation*} 
 \end{defn}
 
 In contrast to the definition of $r_{\nu}$ given in Section \ref{distributions} the maps $\nu^{\un i}_g$ depend not only on the marked coordinates $x_g^{i_g}$, but on all $\{x_e\}_{e\in E(t)\cap E(g \setminus \m N_{<g})}$. This is to ensure that all terms are well-defined densities in a neighborhood of $s=1$. There is some ambivalence in defining them, so it pays of to be careful at this point.
 
We introduce another useful expression for $R_{\nu}$. Locally in $U_{\un i}$,
\begin{equation}\label{fix sub}
 R^{\un i}_{\nu}[\tilde w_{\un i}^s]=\sum_{\m K\subseteq \m N} (-1)^{|\m K|}  \nu^{\un i}_{\m K} \cdot (\tilde w^s )_{\mathcal E_{\m K} \un i }.
\end{equation}
Here $\nu_{\m K}:= \prod_{\gamma \in \m K}\nu_{\gamma}$ and $\m E_{\m K}= \bigcap_{\gamma \in \m K}\m E_{\gamma} \subseteq \m E$. 
This is to be understood in the following way: First restrict the regular part $f_{\un i}^s$ of $\ti w_{\un i}^s$ and the test function $\varphi$ to $\kappa_{\un i}(U_{\un i} \cap \mathcal E_{\m K})$, then pull this product back onto $\kappa_{\un i}( U_{\un i} )$, then multiply by $u^s_{\m N}$ and $\nu^{\un i}_{\m K}$ and finally integrate. In formulae
\begin{equation*}
 \big\langle  \nu^{\un i}_{\m K} \cdot (\tilde w^s )_{\mathcal E_{\m K} \un i} | \varphi \big\rangle = \big\langle (p_{\m K \un i})_*(\nu^{\un i}_{\m K} u_{\m N}^s|dx| ) | \delta_{\m K \un i}[f_{\un i}^s \varphi] \big\rangle.
\end{equation*}
Here $p_{\m K}$ is the (canonical) projection $p_{\m K}: Y \to \m E_{\m K}$ and $\delta_{\m K}$ is the corresponding map $\mathcal D(Y) \rightarrow \mathcal D(\m E_\m K)$. For $\m K=\{g_1, \ldots, g_k\}$ they are locally given by 
\begin{align*}
  p_{\m K \un i}:x & \mapsto \big(x_1^1,\ldots , \hat x^{i_{g_1}}_{g_1}, \ldots ,  \hat x^{i_{g_k}}_{g_k}, \ldots , x^d_n \big), \\ 
  \delta_{\m K \un i}: \varphi &\mapsto \varphi |_{x^{i_{g_1}}_{g_1},\ldots,x^{i_{g_k}}_{g_k}=0}. 
  \end{align*}
 Note that $\delta_{\m K \un i}[f_{\un i}^s \varphi]$ remains compactly supported in the coordinates associated to $G \setminus \cup_{\gamma \in \m K} \m N_{<\gamma}$. On the other hand, $\nu^{\un i}_\m K$ is compactly supported in the coordinates associated to $G \cap ( \cup_{\gamma \in \m K} (\gamma \setminus \m N_{<\gamma} ) )$. But these sets cover $G$ and therefore the \textit{counterterms} $ \nu^{\un i}_{\m K} \cdot (\tilde w^s )_{\mathcal E_{\m K} \un i} $ are well-defined densities in all coordinates except the marked elements (cf.\ the proof of Theorem \ref{rg2}).
The notation is chosen to suggest that $( \tilde w^s )_{\m E_{\m K}}$ can be thought of as the ``restriction'' of $\ti w^s$ onto $\m E_{\m K}$ and the symbol ``$\cdot$'' in $\nu_{\m K} \cdot (\tilde w^s )_{\mathcal E_{\m K}}$ is used to highlight the fact that this expression differs from the usual product of distributions and smooth functions. We call it ``product'' as it is linear and multiplicative in $\nu$. 

\begin{lem} \label{comm rules}
Let $g,h \in \mathcal{D}$. The maps $p_{g,h}$ and $\delta_{g,h}$ both fulfill the following ``commutation rules'':
 \begin{align*}
  p_{g,h}&= p^g_{g,h}\circ p_g = p^h_{g,h} \circ p_h, \\
  \delta_{g,h}& = \delta^g_{g,h}\circ \delta_g = \delta^h_{g,h} \circ \delta_h,
 \end{align*}
 where $p^g_{g,h}:\m E_g  \longrightarrow \m E_{g,h}$ is locally given by 
 \begin{equation*}
  \big(x_1^1,\ldots , \hat x^{i_g}_g, \ldots , x^d_n \big) \mapsto \big(x_1^1,\ldots , \hat x^{i_g}_g, \ldots , \hat x^{i_h}_{h}, \ldots , x^d_n \big)
 \end{equation*}
 and $\delta^g_{g,h}: \mathcal{D}(  \m E_g) \mapsto \mathcal{D}( \m E_{g,h}) $ by
 \begin{equation*}
  \varphi \mid_{x^{i_g}_g=0} \ \mapsto \varphi \mid_{x^{i_g}_g=x^{i_h}_h=0}.
 \end{equation*}
\end{lem}

\begin{proof}
Clear from the definition of both maps.
\end{proof}

Obviously this property generalizes to the case where instead of $\{g,h\}$ a finite subset of a nested set is considered, e.g.
\begin{equation*}
 p_{g_1, \ldots , g_k} = p^{g_1,\ldots, g_{k-1} }_{g_1, \ldots , g_k} \circ \cdots \circ p^{g_1}_{g_1,g_2} \circ p_{g_1}
\end{equation*}
and similarly for $\delta_{g_1, \ldots , g_k}$.

Both renormalization operations produce well-defined densities at $s=1$ as is shown in the next proposition.
\begin{prop}
 Let $(Y,\beta)$ be a wonderful model for a building set of the divergent lattice $\m D$. Then $R_1[\tilde w^s]|_{s=1}$ defines a density on $Y$, while $R_{\nu}[\tilde w^s]$ is a density-valued holomorphic function for all $s$ in a neighborhood of $1$ in $\mathbb{C}$.
\end{prop}

\begin{proof} 
 Note that from the proof of Theorem \ref{laurent} it follows in particular that $\ti w^s$ is really a density on $Y$. By the same argumentation we are able to conclude from \eqref{fix sub} that all counterterms in $R_{\nu}$ are densities for $s$ in a neighborhood of 1: Every subtraction term has the same combination of $u_{\m N}^s$ and $f^s$, transforming under a change of coordinates according to the definition of densities.
 
 In the case of minimal subtraction, by Theorem \ref{laurent} and the definition of $r_1$, all poles of $\ti w^s$ have been discarded. Therefore, $R_1[\ti w^s]|_{s=1}$ is a finite density. From the Taylor expansion of $u^s_{\heartsuit}$ in \eqref{expansion} it follows that $R_1[\ti w^s]$ fails to be a density for $s\neq 1$ because the $u_g^s$ do not transform correctly under a change of coordinates.
 
 It remains to show finiteness of $R_{\nu}[\tilde w^s]$. We argue by induction on the cardinality of nested sets. First consider the case where the nested set consists of a single graph, $\m N=\{g\}$ for some $g\subseteq G$. Let $x^{i_g}_g$ denote the marked element. By Definition \ref{loc fs}
 \begin{align*}
  \langle R^{\un i}_{\nu}[\tilde w_{\un i}^s] | \varphi \rangle =   \langle \tilde w_{\un i}^s | \varphi \rangle   - \langle  \nu^{\un i} \cdot (\tilde w^s )_{\m E_g \un i} | \varphi \rangle,
 \end{align*} 
$\nu^{\un i}=\nu^{\un i}_g$ depending on all $x_e$ with $e\in E(t_g)$.
We expand both summands into their Laurent series (focusing on the principle part only) to get
\begin{align*}
\langle \tilde w_{\un i}^s | \varphi \rangle & = \int dx \ |x^{i_g}_g|^{-1+d_g(1-s)}f_{\un i}^s(x)\varphi(x) \\
& = \int dx^{i_g}_g |x^{i_g}_g|^{-1+d_g(1-s)} \  \int d\hat x \ f_{\un i}^s(x^{i_g}_g,\hat x) \varphi(x^{i_g}_g, \hat x) \\
& =:\! \int dx^{i_g}_g |x^{i_g}_g|^{-1+d_g(1-s)} F(s,x^{i_g}_g).
\end{align*}
Using Formula \eqref{distro split} from Section \ref{distributions},
\begin{align*}
 \langle \tilde w_{\un i}^s | \varphi \rangle =\,& -\frac{2}{d_g} F(s,0) (s-1)^{-1}  \\
& + \int dx^{i_g}_g \ |x^{i_g}_g|^{-1+d_g(1-s)} \left( F(s,x^{i_g}_g) - \theta_g(x^{i_g}_g)F(s,0) \right).
\end{align*}
For the counterterm we have
\begin{align*}
\langle \nu^{\un i} \cdot (\tilde w^s )_{\m E_g \un i} | \varphi \rangle &= \int dx \ |x^{i_g}_g|^{-1+d_g(1-s)}  \nu^{\un i} (x_{t_g}) \big( f_{\un i}^s(x)\varphi(x) \big) |_{x^{i_g}_g=0} \\
& = \int dx^{i_g}_g |x^{i_g}_g|^{-1+d_g(1-s)} \  \int d\hat x \ \nu^{\un i}(x_{t_g}) f_{\un i}^s(0,\hat x) \varphi(0, \hat x) \\
& =:\! \int dx^{i_g}_g |x^{i_g}_g|^{-1+d_g(1-s)} G_{\nu}(s,x^{i_g}_g).
\end{align*}
In the same way as above we get
\begin{align*}
\langle \nu^{\un i} \cdot (\tilde w^s )_{\m E_g \un i} | \varphi \rangle =\,& -\frac{2}{d_g} G_{\nu}(s,0) (s-1)^{-1} \\
& + \int dx^{i_g}_g \ |x^{i_g}_g|^{-1+d_g(1-s)} ( G_{\nu}(s,x^{i_g}_g) - \theta_g(x^{i_g}_g) G_{\nu}(s,0) ).
\end{align*}
Since $\nu^{\un i}(x_{t_g})|_{x^{i_g}_g=0}=1$, $F(s,0)=G_{\nu}(s,0)$ and the pole cancels in the difference. Therefore, $\langle R^{\un i}_{\nu}[\tilde w_{\un i}^s] | \varphi \rangle$ is finite for all $\varphi \in \mathcal D (\kappa_{\un i}(U_{\un i}))$. 
Now let $\m N$ be nested and $h\subseteq G$ such that $\m N':=\m N \cup \{h\}$ is also nested. For $\m K\subseteq \m N$ set $\m K':= \m K \cup \{h\}$. Assume $h$ to be minimal in $\m N'$ (if not choose another minimal element). We want to show finiteness of $R^{\un i'}_\nu [\ti w_{\un i'}^s]$ in $\kappa_{\un i'}(U_{\un i'})$ for $\un i'=(\m N',B)$ with $B$ marked for $\m N$ plus an additional marking for the element $h$ (since $h$ is minimal, all markings of $\m N'$ are of this form). By induction hypothesis $R^{\un i}_{\nu}[\tilde w_{\un i}^s]$ is a well-defined density on $\kappa_{\un i}( U_{\un i} )$ 
for all $s$ in a neighborhood of $1$ in $\mb C$. In \cite{dp} it is shown that $U_{\un i'}$ is the blow-up of the proper transform of $A^{\bot}_h$ in $U_{\un i}$. By minimality of $h$ this blow-up $\beta_h$ is locally given by $\rho_h$, i.e.\ by scaling all $\{x_e\}_{e \in E(t_h)}$ with $x_h^{i_h}$. Moreover, the chart $\kappa_{\un i'}$ is just the inverse of the composition of $\rho_h$ with 
$\Gamma(\pi_{\m B}) \circ \rho_{\m N}$. The pullback of $R^{\un i}_{\nu}[\tilde w_{\un i}^s]$ along this blow-up has an additional divergence in the coordinate $x_h^{i_h}$, one more subtraction is needed to obtain a finite density on $U_{\un i'}$:
\begin{align*}
 r_{\nu_h} \left[ \rho_h^*R^{\un i}_{\nu}[\tilde w^s_{\un i}] \right] =  r_{\nu_h} \left[ \rho_h^*\sum_{\m K\subseteq \m N} (-1)^{|\m K|}  \nu^{\un i}_{\m K} \cdot (\tilde w^s)_{\m E_{\m K} {\un i}} \right] =: *
 \end{align*}
Here minimality of $h$ is crucial. It means that $\nu^{\un i}_\gamma = \nu^{\un i'}_\gamma=: \nu_\gamma $ for all $\gamma \in \m N$, so that $R^{\un i}_{\nu}$ commutes with $\rho_h^*$. We compute the pullbacks locally in $U_{\un i}$, the power counting that produces $u_h^s$ works exactly like in the proof of Proposition \ref{pullback}.
 \begin{align*}
  \langle \rho_h^*( \nu_{\m K} \cdot (\ti w^s)_{\m E_{\m K} {\un i}} ) | \varphi \rangle & = \int (\rho_h^*dx) (\nu_{\m K}\circ \rho_h) (u_{\m N}^s\circ \rho_h) p_{\m K \un i}^*\delta_{\m K \un i}[(f_{\un i}^s\circ \rho_h) \varphi] \\
  &=  \int dx \, \nu_{\m K} u_{\m N}^s  u_h^s   p_{\m K \un i'}^*\delta_{\m K \un i'}[f_{\un i'}^s \varphi] \\
  &= \langle \nu_{\m K} \cdot (\ti w^s)_{\m E_{\m K}{\un i'}} | \varphi \rangle.
 \end{align*}
 Thus,
 \begin{align*}
 * &=  \sum_{\m K\subseteq \m N} (-1)^{|\m K|}  \nu_{\m K} \cdot (\tilde w^s)_{\m E_{\m K} {\un i'} } -  \nu_h \cdot  \sum_{\m K\subseteq \m N} (-1)^{|\m K|} \nu_{\m K} \cdot (\tilde w^s)_{\m E_{\m K \cup \{h\}} {\un i'}} \\
 &=  \sum_{\m K\subseteq \m N'} (-1)^{|\m K|} \nu_{\m K} \cdot (\tilde w^s)_{\m E_{\m K} {\un i'}} \\
 &=  R_{\nu}^{\un i'}[\ti w^s_{\un i'}],
\end{align*}
where we used minimality of $h$ again, 
\begin{align*}
 \langle  \nu_h \cdot ( \nu_{\m K} \cdot (\tilde w^s )_{\m E_{\m K} })_{\m E_h  {\un i'} } | \varphi \rangle 
  & =  \int dx \, u_{\m N'}^s  \nu_h \nu_{\m K}|_{x^{i_h}_h=0} \left( f_{\un i'}^s \varphi \right)|_{x^{i_{\m K}}_{\m K}=x^{i_h}_h=0} \\
& =  \int dx \,  u_{\m N'}^s  \nu_{\m K'} \left( f_{\un i'}^s \varphi \right)|_{x^{i_{\m K}}_{\m K}=x^{i_h}_h=0} \\
  & =  \langle  \nu_{\m K'} \cdot (\tilde w^s )_{\m E_{\m K'} {\un i'} } | \varphi \rangle. 
\end{align*}
We see that both densities coincide and the proposition is proven.
\end{proof}

 Both renormalization operators have another property that every sensible renormalization should have; they commute with multiplication by smooth functions.
\begin{lem}
Let $f\in \m C^{\infty}(\kappa_{\un i} (U_{\un i} ))$, then
\[
 R^{\un i}_1[\ti w_{\un i}^s f]= f R^{\un i}_1[\ti w_{\un i}^s], \qquad
 R^{\un i}_{\nu}[\ti w_{\un i}^s f]=f R^{\un i}_{\nu}[\ti w_{\un i}^s].
\]
\end{lem}
\begin{proof}
 Clear from the definition of both operators. The regular part of the density is treated as a test function, the same happens in the definition of multiplication of densities by smooth functions.     
\end{proof}

Finally, we are able to state a solution of the renormalization problem.

\begin{defn}[Renormalized Feynman rules]
Let $R$ denote one of the renormalization operators $R_1$ or $R_{\nu}$ on a wonderful model $(Y, \beta)$ for the divergent arrangement of an at most logarithmic graph $G$. 
Define the renormalized Feynman distribution by 
\begin{align*}
 \mathscr R[v_G] := \beta_* R[\ti w_G]|_{s=1}. 
\end{align*}
Then the \textit{renormalized Feynman rules} are given by the map
\begin{equation*}
 \Phi_R:G \longmapsto (X^G,\mathscr R(v_G)).
\end{equation*}
The pair $(X^G,\mathscr R(v_G))$ can now be evaluated at $\varphi \in \m D(X^G)$,
\begin{align*}
 \left( \text{eval}_{\varphi}\circ \Phi_R \right) (G) & = \langle \mathscr R[v_G] | \varphi \rangle= \langle \beta_* R[\ti w_G] |_{s=1} \mid \varphi \rangle \\
& = \langle R[\ti w_G] |_{s=1} \mid  \beta^* \varphi \rangle. 
\end{align*}
\end{defn}

\noindent To carry out the evaluation at $\varphi$ we choose a partition of unity $\{\chi_{\un i}\}_{\un i \in \{(\m N,B)\}}$ on $Y$, subordinate to the covering $\{U_{\un i}\}_{\un i \in \{(\m N.B)\}}$. Write $\pi_{\un i}$ for $\chi_{\un i} \circ \kappa_{\un i}^{-1}$. Then
\begin{equation*}
 \big\langle R[\ti w^s_G] |_{s=1} \mid \beta^* \varphi \big\rangle = \sum_{\un i} \big\langle \pi_{\un i} (R[\ti w^s_G])_{\un i}|_{s=1} \mid  \varphi \circ \rho_{\un i} \big\rangle.
\end{equation*}
To see that this definition does not depend on the chosen partition of unity let $\{\chi_{\un j}'\}$ denote another partition, also subordinate to the $\{U_{\un j} \}_{\un j \in \{(\m N,B)\}}$. Then
\begin{align*}
  &\sum_{\un i} \big\langle \pi_{\un i} (R[\ti w_G])_{\un i}|_{s=1} \mid  \varphi \circ \rho_{\un i} \big\rangle \\
= &\sum_{\un i} \int_{\text{supp}(\pi_{\un i})} dx \sum_{\m K \subseteq \m N} (-1)^{|\m K|} u_{\m N} \nu^{\un i}_{\m K} \delta_{\m K \un i}[f_{\un i} (\varphi \circ \rho_{\un i}) \pi_{\un i} ] \\
  = &\sum_{\un i}  \int_{ \text{supp}(\pi_{\un i})} dx\sum_{\m K \subseteq \m N} (-1)^{|\m K|} u_{\m N} \nu^{\un i}_{\m K} \delta_{\m K \un i}\left[ f_{\un i} (\varphi \circ \rho_{\un i}) \pi_{\un i} \sum_{\un j} \pi_{\un j}' \right] \\
  = &\sum_{\un i} \sum_{\un j} \int_{\text{supp}(\pi_{\un i}) \cap \text{supp}(\pi_{\un j}')}dx \sum_{\m K \subseteq \m N} (-1)^{|\m K|} u_{\m N} \nu^{\un j}_{\m K} \delta_{\m K \un j}\left[ f_{\un j} (\varphi \circ \rho_{\un j}) \pi_{\un i} \pi_{\un j}'\right] \\
 = &\sum_{\un j} \int_{\text{supp}(\pi_{\un j}')} dx\sum_{\m K \subseteq \m N} (-1)^{|\m K|} u_{\m N} \nu^{\un j}_{\m K} \delta_{\m K \un j}\left[ f_{\un j} (\varphi \circ \rho_{\un j}) \pi_{\un j}' \sum_{\un i} \pi_{\un i} \right] \\
 = &\sum_{\un j} \big\langle \pi_{\un j}' (R[\ti w_G])_{\un j}|_{s=1} \mid  \varphi \circ \rho_{\un j} \big\rangle.
\end{align*}

This finishes the process of \textit{wonderful renormalization}. From a mathematical point of view we are done, but for a physicist it is not clear yet that we have constructed a reasonable renormalization. In addition to producing finite distributions both schemes have to fulfill another condition that is dictated by physics. It is called the \textit{locality principle} (see \cite{eg}) and, roughly speaking, assures that the renormalized distributions still obey the laws of physics. There are various equivalent formulations of this; we will use a version for single graphs from \cite{bbk}, more about this in Section \ref{eg2}. 
Before that we turn our attention to the dependence of the operators $R$ on the \textit{renormalization points}, i.e.\ we study what happens if we change the collection of maps $\{\nu\}$ or the cutoff in the definition of $u^s_{\heartsuit}$, respectively.

\section{Renormalization group}\label{renormalization group}

In this section we take a closer look at the renormalized distribution densities. First we consider (local) subtraction at fixed conditions. The case of minimal subtraction then follows by similar arguments since it can be thought of as a ``non-smooth'' version of the former. 
What happens if we change the cutoff functions in the definition of the operator $R_{\nu}$? Clearly, for primitive graphs the difference is a density supported on the exceptional divisor $\m E$, its push-forward to $X$ will then be supported on the origin.
To get an idea what happens in the general case it is useful to start with an example. 
Throughout this section we drop the index $\un i$ in all local expressions; to keep track of the counterterms in the renormalization operators we write $R^{\m N}_{\nu}$ for $R^{\un i}_{\nu}$. 

\begin{example}
 Let $G$ be the dunce cap graph (Figure \ref{dunce}). Locally (in $d=4$ dimensions and $\m N=\{g,G\}$, $g$ denoting the divergent fish subgraph) we have for $\varphi \in \mathcal D(\kappa (U) )$
 \begin{align*}
  \big\langle R^{\m N}_{\nu}[\tilde w^s] | \varphi \big\rangle &=\sum_{\m K \subseteq \m N} (-1)^{|\m K|} \langle \nu_{\m K} \cdot (\tilde w^s)_{\m E_{\m K}} | \varphi \rangle \\
  &=\langle \tilde w^s | \varphi \rangle - \big\langle (p_G)_*( \nu_{G} u_{\m N}^s |dx| ) | \delta_G[f^s\varphi] \big\rangle \\
  & \quad - \big\langle (p_g)_*( \nu_{g} u_{\m N}^s |dx| ) | \delta_g[f^s \varphi] \big\rangle \\
  & \quad + \big\langle (p_{g,G})_*( \nu_g \nu_{G} u_{\m N}^s |dx| ) | \delta_{g,G}[f^s \varphi] \big\rangle.
 \end{align*}
 
Changing the renormalization point $\{\nu\}$, by linearity the difference of the two renormalized expressions is again a sum of this form. However, it will contain a mixture of $\nu$ and $\nu'$ as renormalization points. But we can express the terms with $\nu'$ again by $\nu$-terms only and obtain a finite sum of $\nu$-renormalized expressions. Another way to see this is by Taylor expansion using the calculus of variations,
\begin{align*}
 \frac{d}{dt}\mid_{t=0} \big\langle R^{\m N}_{\nu + t \mu }[\tilde w^s] | \varphi \big\rangle & =  -\big\langle (p_g)_*(\mu_g u_{\m N}^s |dx|) | \delta_g[f^s\varphi] \big\rangle \\
 & \quad -\big\langle (p_G)_*(\mu_G u_{\m N}^s |dx|)  | \delta_G[f^s\varphi] \big\rangle \\
 & \quad + \big\langle (p_{g,G})_* \big(( \nu_g \mu_G + \nu_G \mu_g ) u_{\m N}^s |dx| \big)  | \delta_{g,G}[f^s\varphi] \big\rangle , \\
 \frac{d^2}{dt^2}\mid_{t=0} \big\langle R^{\m N}_{\nu + t \mu }[\tilde w^s] | \varphi \big\rangle & =  2 \big\langle (p_{g,G})_* ( \nu_g \nu_G  u_{\m N}^s |dx|)  | \delta_{g,G}[f^s\varphi] \big\rangle, \\
 \frac{d^k}{dt^k}\mid_{t=0} \big\langle R^{\m N}_{\nu + t \mu }[\tilde w^s] | \varphi \big\rangle & = 0 \quad\text{for all $k>2$.}
\end{align*}
Thus, for $\mu_{\gamma}:= \nu_{\gamma}' - \nu_{\gamma}$, $\gamma\in \{g,G\}$, using Lemma \ref{comm rules},
\begin{align*}
 \big\langle R^{\m N}_{\nu'}[\tilde w^s] -  R^{\m N}_{\nu}[\tilde w^s] | \varphi \big\rangle 
 =\,& -\big\langle R^{\{G\}}_{\nu_G}[  \mu_g \cdot (\tilde w^s)_{\m E_g} ] | \varphi \big \rangle - \big\langle R^{ \{g\} }_{\nu_g}[  \mu_G \cdot (\ti w^s)_{\m E_G} ] | \varphi  \big\rangle \\
 &  + \langle  \nu_g \nu_G  \cdot (\tilde w^s)_{\m E_{g,G}} | \varphi \rangle.
\end{align*}
We see, as expected, that the difference is a sum of densities supported on the components of the exceptional divisor, given by subsets of the nested set $\m N$. Since $\mu_{\gamma}=0$ for $x^{i_{\gamma}}_{\gamma}=0$, they are finite, except if a point approaches the intersection of two components $\m E_{\m J} \cap \m E_{\m K}$ for $\m J,\m K \subseteq \m N$. But in this case the necessary subtractions are already provided by the counterterms associated to the set $\m J \cup \m K$. 
\end{example}

\begin{prop}\label{rg prop}
 Let $\{\nu\}$ and $\{\nu'\}$ be two collections of renormalization points on a wonderful model for the divergent arrangement. 
 Locally in $U=U_{\m N, B}$ the difference between the operators $R_{ \nu'}$ and $R_{\nu}$ acting on $\tilde w^s$ is given by 
 \begin{equation*}
  (R_{\nu '} - R_{\nu})[\tilde w^s]\overset{loc.}{=}\sum_{\emptyset \neq \m K \subseteq \m N}(-1)^{|\m K|} R^{\m N\setminus \m K}_{\nu}[\mu_{\m K} \cdot (\ti w^s)_{\m E_{\m K}} ]
 \end{equation*}
with $R^{\emptyset}_{\nu}:= \text{id}_{ \mathcal{\tilde D}'(\kappa (U))  }$.
\end{prop}

\begin{proof}
 Induction on $n=|\m N|$. The statement holds in the cases $n=1$ and $2$ (see example above). Let $\m N$ be a nested set of cardinality $n$ and $h \notin \m N$ an additional divergent subgraph such that $\m N'=\m N \cup \{h\}$ is also nested. For $\m K\subseteq \m N$ set $\m K':=\m K \cup \{h\}$.

 \begin{align*}
  \big(R^{\m N'}_{\nu '} - R^{\m N'}_{\nu}\big)[\tilde w^s] &=\sum_{\emptyset \neq \m K \subseteq \m N}(-1)^{|\m K|} (\nu_{\m K}' -  \nu_{\m K}) \cdot ( \tilde w^s  )_{\m E_{\m K}} \\ 
  & \quad + \sum_{\m K \subseteq \m N}(-1)^{|\m K'|}  ( \nu_{\m K'}' -  \nu_{\m K'})\cdot ( \tilde w^s )_{\m E_{\m K'}} \\
  & =: A +B.
 \end{align*}
 By induction hypothesis
\begin{equation*}
A = \sum_{\emptyset \neq \m K \subseteq \m N} (-1)^{|\m K|} R^{\m N\setminus \m K}_{\nu} [ \mu_{\m K} \cdot (\tilde w^s )_{\m E_{\m K}} ]. 
\end{equation*}
 Using 
 \begin{equation*}
 \prod_{i=1}^n (a_i + b_i) - \prod_{i=1}^n a_i=   \sum_{\emptyset \neq J \subseteq \set{n} } b_J a_{\set{n} \setminus J} 
 \end{equation*}
 we expand $B$ into two parts, depending on whether $\nu$ or $\mu$ carries an index~$h$,
\begin{align*}
 B = B_1 + B_2 = &\sum_{\m K\subseteq \m N} (-1)^{|\m K'|}    \sum_{\m J \subseteq \m K} \mu_{\m J'}  \nu_{\m K \setminus \m J} \cdot (\tilde w^s)_{\m E_{\m K'}} \\
  &+ \sum_{\emptyset \neq \m K\subseteq \m N} (-1)^{|\m K'|}  \sum_{\emptyset \neq \m J\subseteq \m K}  \mu_{\m J} \nu_{\m K'\setminus \m J}   \cdot (\ti w^s)_{\m E_{\m K'}}.
\end{align*}
Note that in $A$ all densities $\mu_{\m L} \cdot (\tilde w^s )_{\m E_{\m L}}$ have an additional, not yet renormalized divergence corresponding to the subgraph $h \in \m N'$. In order to renormalize them we have to add the counterterms associated to $h$, i.e.\ all terms in $B$ containing $\nu_h$. For non-empty $\m L\subseteq \m N$ fixed
\begin{align*}
&(-1)^{|\m L|} R^{\m N\setminus \m L}_{\nu} [ \mu_{\m L} \cdot (\tilde w^s)_{\m E_{\m L}} ] + \sum_{\m L \subseteq \m J \subseteq \m N} (-1)^{|\m J'|}  \mu_{\m L} \nu_{\m J'\setminus \m L } \cdot (\tilde w^s)_{\m E_{\m J'}}\\
=\,& (-1)^{|\m L|}  \sum_{\m I \subseteq \m N\setminus \m L} (-1)^{|\m I|} \big( \mu_{\m L} \nu_{\m I} \cdot (\ti w^s)_{\m E_{\m L\cup \m I}} - \mu_{\m L} \nu_{\m I'} \cdot (\ti w^s)_{\m E_{\m L\cup \m I'}} \big) \\
=\,& (-1)^{|\m L|} \sum_{\m I \subseteq \m N'\setminus \m L} (-1)^{|\m I|}  \mu_{\m L} \nu_{\m I} \cdot (\ti w^s)_{\m E_{\m L\cup \m I}}\\
=\,& (-1)^{|\m L|} R^{\m N'\setminus \m L}_{\nu} [ \mu_{\m L} \cdot (\ti w^s)_{\m E_{\m L}} ]
\end{align*}
is then a finite expression. Doing this for every non-empty $\m L\subseteq \m N$ covers the whole sum $B_2$ because every term $\mu_{\m L} \nu_{\m I'} \cdot (\ti w^s)_{\m E_{\m L \cup \m I'}}$ appears exactly once and the signs match since
\begin{align*}
  \sum_{ \m I \subseteq \m N \setminus \m L} (-1)^{|\m L|+|\m I|+1} \mu_{\m L} \nu_{\m I'} \cdot (\ti w^s)_{\m E_{\m L\cup \m I'}}
 =  \sum_{\emptyset \neq \m K\subseteq \m N, \m L \subseteq \m K } (-1)^{|\m K'|}  \mu_{\m L} \nu_{\m K'\setminus \m L}   \cdot (\ti w^s)_{\m E_{\m K'}}.
\end{align*}
The same argumentation works for $B_1$. Fix $\m L\subseteq \m N$ and consider all terms in $B_1$ containing $\mu_{\m L'}$:
\begin{align*}
 \sum_{\m K\subseteq \m N, \m L \subseteq \m K} (-1)^{|\m K'|}    \mu_{\m L'}  \nu_{\m K \setminus \m L} \cdot (\tilde w^s)_{\m E_{\m K'}} &= \sum_{\m I\subseteq \m N \setminus \m L } (-1)^{|\m L'|+|\m I|}    \mu_{\m L'}  \nu_{\m I} \cdot (\tilde w^s)_{\m E_{\m L' \cup \m I}} \\
 &=  (-1)^{ |\m L'| } R_{\nu}^{\m N' \setminus \m L'} [\mu_{\m L'}\cdot (\ti w^s)_{\m E_{\m L'}}].
\end{align*}
Putting everything together we have shown that locally the difference be\-tween two renormalization operators $R_{\nu '}$ and $R_{\nu}$ is expressible as a sum of densities, supported on the components $\m E_{\m K}$ for $\m K \subseteq \m N$ and renormalized in the remaining directions according to subsets of $\m N \setminus \m K$.
\end{proof}

This is a nice formula showing that a \textit{finite renormalization} (i.e.\ a change of renormalization points) amounts to adding a density supported on the exceptional divisor, as expected from the toy model case on $\mathbb{R}$ (or $\mathbb{R}^d$ for homogeneous distributions). But we can do even better and physics tells us what to expect: The Hopf algebraic formulation of the renormalization group predicts that the densities appearing in $(R_{\nu'} - R_{\nu})[\ti w^s]$ should correspond to graphs showing up in the coproduct of $G$ (for more on this see \cite{ck} and \cite{dk}). In the local formulation presented here the coproduct translates into local contractions, i.e.\ contractions with respect to nested sets $\m N$. 

\begin{example}
Turning back to the example at the beginning of this section where we calculated $\langle R^{\m N}_{ \nu'}[\tilde w^s] -  R^{\m N}_{\nu}[\tilde w^s] | \varphi \rangle $ to be 
\begin{align*}
 -\big\langle R^{ \{G\} }_{\nu_G}[ \mu_g \cdot (\tilde w^s )_{\m E_g} ] | \varphi \big\rangle - \big\langle R^{ \{g\} }_{\nu_g}[  \mu_G \cdot (\tilde w^s)_{\m E_G} ] | \varphi \big \rangle + \langle \nu_g \nu_G \cdot ( \tilde w^s )_{\m E_{g,G}}  | \varphi \rangle,
\end{align*}
we now examine the individual terms in more detail. Eventually we are interested in the pairing with test functions $\varphi$ that are pullbacks of test functions on $X$. Recall that locally $\beta$ is given by the scaling map $\rho$. In $\kappa( U )=X=M^2$, corresponding to $\m N=\{g,G\}$, an adapted spanning tree chosen as in the example in Section \ref{wm revis} and marked elements $x_G, y_g \ (x,y \in \mb R^4)$,
\begin{align*}
   &\quad\ \big\langle R^{ \{G\} }_{\nu_G} [  \mu_g \cdot (\tilde w^s)_{\m E_g} ] | \varphi  \big\rangle\\
   & = \int_{M^2}  d^4xd^4y \, \mu_g |x_G|^{7-8s} |y_g|^{3-4s} \left( \delta_g[f^s \varphi] - \nu_G \delta_{g,G}[f^s \varphi] \right) \\
 & = c_g \int_{M} d^4x \, |x_G|^{7-8s} \left( \frac{\psi(x_G \hat x,0)}{\hat x^{4s} } - \frac{\nu_G(x) \psi(0,0)}{\hat x^{4s} } \right) \\
 & = c_g \big\langle R^{ \{G/g \} }_{\nu_G} [ \tilde w^s_{G/g} ] | \delta_g[\varphi] \big\rangle.
\end{align*}
Here $\ti w^s_{G/ g}$ is the density associated to the contracted graph $G/g$ (more precise, its local expression in $U_{\m N',B'}$ with $B'$ spanned by $x=\{x^i_e\}_{e \in E(t) \setminus E(g) }$ and $\m N'=\{G/g\}$ - the exponent $7-8s$ in $|x_G|$ does not match but we neglect this little technical problem here; see below for the general argument). The coefficient $c_g$ is given by
\begin{align*}
 c_g=  \int_{M} d^4y \, \frac{|y_g|^{3-4s}}{\hat y^{4s}} \big( \nu_g'(y) - \nu_g(y) \big) = \big\langle R^{ \{g\} }_{\nu_g} [ \tilde w^s_g ] | \nu_g' \big\rangle
\end{align*}
because $\nu_g'|_{y_g=0}=1$. In the same manner we calculate 
\begin{align*}
  \big\langle R^{ \{g\} }_{\nu_g}[  \mu_G \cdot (\tilde w^s)_{\m E_G} ] | \varphi  \big\rangle = c_G \psi(0,0)=c_G \langle \delta_G | \varphi \rangle
\end{align*}
 with  
 \begin{align*}
  c_G &= \int_{M^2} d^4xd^4y \, \mu_G |x_G|^{7-8s} |y_g|^{3-4s} \left(  \frac{1}{\hat x^{2s} \hat y^{4s} (\hat x + y_g \hat y )^{2s}} - \frac{\nu_g(y)}{\hat x^{4s} \hat y^{4s} } \right) \\
  & = \big\langle R^{\m N}_{\nu_g, \nu_G } \left[  \tilde w_G^s\right] | \nu_G'  \big\rangle.
 \end{align*}
The last term $ \langle [ \nu_g \nu_G \cdot \tilde w^s ]_{\m E_{g,G}}  | \varphi \rangle$ evaluates to $c_{g,G} \psi(0,0)$ with 
\begin{align*}
 c_{g,G} = \big\langle R^{ \{g\} }_{\nu_g}[ \tilde w_g^s ] | \nu_g' \big\rangle \big\langle R^{ \{G\} }_{\nu_G}[  \tilde w_{G/g}^s ] | \nu_G' \big\rangle.  
\end{align*}
\end{example}

To formulate this in the general case we need to define a contraction operation $\ds$ not only on single graphs but also on nested sets.  
\begin{defn}
 Let $\m N $ be a nested set for some building set $\m B \subseteq \m D$ and let $\m J \subseteq \m N$. The contraction $\m N \ds \m J$ is defined as the poset with underlying set
 \begin{equation*}
  \m N \ds \m J := \{ g\ds \m J \mid g \in \m N \},
 \end{equation*}
partially ordered by inclusion. Since the inclusion operation differs from the one in $\m N$ (contracted graphs may not be subgraphs of $G$ anymore, although we can identify them with subgraphs via their edge sets), we denote this partial order by $\sqsubseteq$. 
\end{defn}

The partial order $\sqsubseteq$ is most easily understood by looking at the Hasse diagram of $\m N$. Replace every $g \in \m N$ by $g\ds \m J$, remove all lines that connect elements of $\m J$ to ``above'' and draw a new line from $o$ to every element that became disconnected in the process. Note that in particular all elements of $\m J$ have become maximal in $\m N\ds \m J$.
In addition, we denote by abuse of notation the corresponding contractions on adapted spanning trees by the same symbol, i.e.\ we define 
\begin{align*}
 t \ds \m J &:= t / t_{\m J} \quad\text{where } t_{\m J}:= \bigcup_{\gamma \in \m J} t_{\gamma}, \\
 t_g \ds \m J & := t_g \ds \m J_{<g}.
\end{align*}

\begin{example}
 Let $G$ be the graph shown in Figure \ref{contr1}. Denote by $\gamma_1$, $\gamma_2$ and $\gamma_3$ the three fish subgraphs from left to right, and let $g$ and $h$ be the full subgraphs on the vertex sets $V(g)=\{0,1,2,3\}$ and $V(h)=\{2,3,4,5\}$. In Figure \ref{contr1} we depict an $I(\m D(G))$-nested set $\m N=\{ \gamma_1, \gamma_2, \gamma_3, g , G\}$ and the poset $(\m N \ds \m J,\sqsubseteq)$ for $\m J=\{\gamma_1,\gamma_3,g\}$.
 \end{example}
 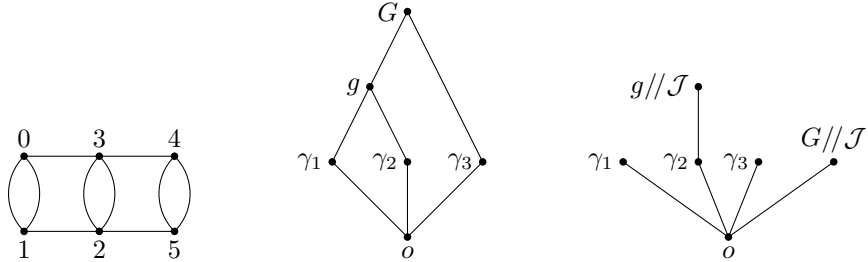
\begin{figure}[h]
 \centering
\begin{tikzpicture}[node distance=1cm and 1cm]
\coordinate[label=above:$0$] (v1);
\coordinate[below=of v1,label=below:$1$] (v2);
\coordinate[right=of v2,label=below:$2$] (v3);
\coordinate[above=of v3,label=above:$3$] (v4);
\coordinate[right=of v3,label=below:$5$] (v6);
\coordinate[right=of v4,label=above:$4$] (v5);

\draw (v2) -- (v3);
\draw (v1) -- (v4);
\draw (v3) to (v6);
\draw (v4) to (v5);
\draw (v1) to[out=-135,in=135] (v2);
\draw (v1) to[out=-45,in=45] (v2);
\draw (v4) to[out=-135,in=135] (v3);
\draw (v4) to[out=-45,in=45] (v3);
\draw (v5) to[out=-135,in=135] (v6);
\draw (v5) to[out=-45,in=45] (v6);
\fill[black] (v1) circle (.05cm);
\fill[black] (v2) circle (.05cm);
\fill[black] (v3) circle (.05cm);
\fill[black] (v4) circle (.05cm);
\fill[black] (v5) circle (.05cm);
\fill[black] (v6) circle (.05cm);
\end{tikzpicture}
\hspace{1cm}
\begin{tikzpicture}[node distance=1cm and 1cm]
\coordinate[label=below:$o$] (v0);
\coordinate[above left=of v0,label=left:$\gamma_1$] (v1);
\coordinate[above=of v0,label=left:$\gamma_2$] (v2);
\coordinate[above right=of v0,label=left:$\gamma_3$] (v3);
\coordinate[above right=of v1,xshift=-.5cm,label=left:$g$] (v4);

\coordinate[above =of v0,yshift=2cm,label=left:$G$] (v6);
\draw (v0) to (v1);
\draw (v0) to (v2);
\draw (v0) to (v3);
\draw (v1) to (v4);
\draw (v2) to (v4);
\draw (v4) to (v6);
\draw (v3) to (v6);
\fill[black] (v0) circle (.05cm);
\fill[black] (v1) circle (.05cm);
\fill[black] (v2) circle (.05cm);
\fill[black] (v3) circle (.05cm);
\fill[black] (v4) circle (.05cm);
\fill[black] (v6) circle (.05cm);
\end{tikzpicture}
 \hspace{1cm}
\begin{tikzpicture}[node distance=1cm and 1cm]
\coordinate[label=below:$o$] (v0);
\coordinate[left=of v0,xshift=0.6cm] (l);
\coordinate[right=of v0,xshift=-0.6cm] (r);
\coordinate[above left=of l,label=left:$\gamma_1$] (v1);
\coordinate[above=of l,label=left:$\gamma_2$] (v2);
\coordinate[above =of r,label=left:$\gamma_3$] (v3);
\coordinate[above right=of r,label=above:$G \ds \m J$] (v4);
\coordinate[above =of v2,label=left:$g \ds \m J$] (v5);
\draw (v0) to (v1);
\draw (v0) to (v2);
\draw (v0) to (v3);
\draw (v0) to (v4);
\draw (v2) to (v5);
\fill[black] (v0) circle (.05cm);
\fill[black] (v1) circle (.05cm);
\fill[black] (v2) circle (.05cm);
\fill[black] (v3) circle (.05cm);
\fill[black] (v4) circle (.05cm);
\fill[black] (v5) circle (.05cm);
\end{tikzpicture}
\caption{The graph $G$ and Hasse diagrams for $\m N$ and $\m N \ds \m J$}\label{contr1}
\end{figure}

\begin{thm}\label{rg thm}
Consider renormalization operators $R_{\nu}$ for two sets of subtraction points $\{\nu'\}$ and $\{\nu\}$. Let $\m N$ be nested for a building set $\m B \subseteq \m D$ and $B$ marked accordingly. Then the local expression for the difference $(R_{\nu'}-R_{\nu})[\ti w^s]$ applied on a test function $\varphi=\beta^*\psi $ for $\psi \in \m D(\beta(U))$ is given by
\begin{equation} \label{rg1}
 \big\langle (R^{\m N}_{\nu'}-R^{\m N}_{\nu})[\ti w^s] | \varphi \big\rangle = \sum_{\emptyset \neq \m K \subseteq \m N} c_{\m K} \big\langle R_{\nu}[\ti w_{G \ds \m K}^s] | \delta_{\m K}[\varphi] \big\rangle 
\end{equation}
where
\begin{equation} \label{rg2}
 c_{\m K}= \prod_{\gamma \in \m K} \big\langle R_{\nu}[\ti w^s_{\gamma \ds \m K}] | \nu_{\gamma}' \big\rangle 
\end{equation}
and $\langle R_{\nu}[\ti w_{ \emptyset }^s] | \delta_{G}[\varphi] \rangle$ is to be understood as $  \langle \delta|\psi \rangle=\psi (0)$.

Define $\m H:=\cup_{\gamma \in \m K}\m H_{\gamma}$ with $\m H_{\gamma}:=\{ h \in \m N \mid h \ds \m K \in (\m N \ds \m K)_{\sqsubset \gamma \ds \m K} \}$. Then the indices for $R_{\nu}$ in \eqref{rg1} are given by $\m N \setminus (\m K \cup \m H)$. Likewise, in \eqref{rg2} $\m H_{\gamma}\cup \{\gamma\}$ is the index in the factor associated to $\gamma \in \m K$.
\end{thm}

\begin{proof}
 Using Proposition \ref{rg prop} we examine all terms in $\langle (R^{ \m N}_{\nu'}-R^{ \m N}_{\nu})[\ti w^s] | \varphi \rangle $ separately. The proof consists of two steps. First we study how $\delta_\m K$ acts on the maps $f$ and $\varphi=\beta^*\psi$. This allows then in the second step to show that the integral arising in the evaluation of $\langle (R^{ \m N}_{\nu'}-R^{ \m N}_{\nu})[\ti w^s] | \varphi \rangle $ factorizes into a product of integrals according to \eqref{rg1} and \eqref{rg2}.
  
 Claim: For $\m J\subseteq \m N$ the map $\delta_{\m J}$ operates on $f$ and $\varphi=\beta^*\psi \overset{loc.}{=}\psi \circ \rho$ by
 \begin{align*}
 \varphi &\mapsto \delta_{\m J}[\varphi]=\varphi|_{x_{\m J}=0}, \\
f &\mapsto \delta_{\m J}[f]=\prod_{\gamma \in \m J\cup \{G\}}f_{\gamma \ds \m J}.  
 \end{align*}
Here $f_{g \ds \m J}$ is defined as follows: Contracting $t_g$ with respect to $\m J$ defines an adapted spanning tree for $g \ds \m J$ (contracting graphs in $\m N$ and $t$ accordingly does not change the properties of $t$ being spanning and adapted - cf.\ the construction in Proposition \ref{adsptr exists}). Define
\begin{equation}\label{x tilde g}
X^{g \ds \m J}:=\big\{ (x_{e_1},\ldots, x_{e_k}) \mid \{e_1,\ldots, e_k\}=E(t_g \ds \m J) \big\} 
\end{equation}
with adapted basis $B':=B|_{e \in E(t_g\ds \m J)}$. The set $\m N':={\m N\ds \m J}_{\sqsubseteq g \ds \m J}$ is nested for the building set $\m B':= \{ \gamma \ds \m J \mid \gamma \in \m B \text{ and } \gamma \leq g \}$ in the divergent arrangement of $g\ds \m J$. Mimicing the wonderful construction in this case, we obtain an open set $U_{\m N',B'}$ that is a local piece of a wonderful model for the graph $g \ds \m J$. The function $f^s_{g\ds \m J}$ is then the regular part of the pullback of $\ti v^s_{g\ds \m J}$ in this chart. The factor $f_{G \ds \m J}$ collects all the remaining parts and is defined in the same way, except for one special case: If $G$ does not lie in $\m N$, or even not in $\m D$ (locally in $U_{\m N,B}$ this is the same!), and $G \ds \m J$ is primitive, then $\m N'=\emptyset$ and we do not have a local model to pullback $\ti v^s_{G\ds \m J}$ onto. But in this case $v_{G\ds \m J}=f_{G\ds \m J}$ is already regular and no model is required. Also note that if $G \in {\m N}$, the operation $\
delta_G$ does not 
alter $f$ 
since it does not depend on the variable $x_G^{i_G}$. 

Recall that in coordinates given by an adapted spanning tree the distribution kernel $v$ is a product of factors $(y_e)^{2-d}$ with $e\in E(G)$ and 
\begin{equation*}
 y_e=\begin{cases}
      x_e & \text{ if $e$ is an edge of $t$,} \\
      \sum_{e' \in E(t_e)} \sigma_t(e') x_{e'} & \text{ if $e$ is an edge in $G\setminus t$}.
     \end{cases}
\end{equation*}
Moreover, the blow-down $\beta$ is locally given by the map $\rho=\rho_{\m N,B}$ that scales all $x_e$ with $e\in E(t_g)$ and $g\in {\m N}$ by $x_g^{i_g}$.
To prove the claimed properties of $\delta_{\m J}$ we argue in the same manner as in the proof of Theorem \ref{laurent}: 

1. Since $\varphi = \psi \circ \rho$, we have that $\delta_{\m J} [\varphi]$ is equivalent to $\varphi|_{ \{x_h^{i_h}=0\} }$ for $h$ in $\max {\m J}$, the set of maximal elements of ${\m J}$. This means that the resulting map only depends on the variables $x_e$ with $e\in E(t \cap {\m N}_{> \max \m J} )$. All other vectors are scaled by the $x^{i_h}_h$ and therefore vanish after $\delta_{\m J}$ is applied. Another way to put this is that $\delta_{\m J}[\varphi]$ depends only on the $x_e$ with $e\in E(t \ds {\m J})$. In particular, if $G\in {\m J}$ then $\delta_{\m J}[\varphi]$ is just a constant, $\delta_{\m J}[\varphi]=\langle \delta | \psi \rangle = \psi(0)$. 

2. For the second claim start with ${\m J}=\{g \}$ consisting only of a single subgraph $g\subsetneq G$. The part of $f$ that depends only on the vectors associated to edges of $g$ is unaffected by setting $x^{i_g}_g=0$ because all $x_e$ with $e\in E(t_g)$ get scaled and so the factor $x^{i_g}_g$ pulls out (it is already absorbed into the definition of $u_g^s$). On the other hand, the remaining part of $f$ depends on $x_e$ with $e\in E(t_g)$ only through special linear combinations. These linear combinations express vectors representing edges $e'$ that do not lie in $g$ but are connected to a vertex of $g$ such that $E( t_{e'}) \cap E(t_g ) \neq \emptyset $. They become independent of $x_e$ after setting $x_g^{i_g}$ to zero. Therefore, $\delta_{\m J}[f]$ splits into a product of two factors depending on the mutual disjoint sets of vectors $\{x_e\}_{e\in E(t_g)}$ or $\{x_e\}_{e\in E( t/t_g) }$, i.e.
\begin{equation*}
 \delta_g[f]= f_g f_{G/g} = f_{g\ds \m J}f_{G\ds \m J} .
\end{equation*}
Adding another graph $ h \neq G $ from $\m N$ to ${\m J}$ and using Lemma \ref{comm rules} we can express $\delta_{\m J}$ as 
\begin{equation*}
 \delta_{\m J}[f]= \delta^g_{g,h}[\delta_g[f] ]= \delta_h[ f_g f_{G/g} ].
\end{equation*}
There are three possible cases (due to Lemma \ref{irr} there cannot be two incomparable $g,h$ with non-empty overlap in any nested set): 
\begin{enumerate}
 \item 
$g$ and $h$ are incomparable. Then $f_g$ does not depend on any $x_e$ with $e \in E(t_h)$ and 
\begin{equation*}
 \delta_h[ f_g f_{G/g} ]=f_g \delta_h[f_{G/g} ] = f_{g\ds \m J} \delta_h[f_{G/g} ] .
\end{equation*}

\item
$h$ is contained in $g$; $h\subsetneq g$. Then all $\{x_e\}_{e \in E(t_h)}$ are scaled by $x_g^{i_g}$ and $f_{G/g}$ is independent of these. Thus, only $f_g$ is affected by contracting~$h$,
\begin{align*}
 \delta_h[f_g f_{G/g} ] = \delta_h [f_g] f_{G/g} = \delta_h [f_g] f_{G \ds \m J} .
\end{align*}

\item $h$ contains $g$; $g \subsetneq h$. Then all $\{x_e\}_{e \in E(t_g)}$ are scaled by $x_h^{i_h}$ and $f_g$ is not affected by setting $x_h^{i_h}=0$. Therefore,
\begin{equation*}
 \delta_h[f_g f_{G/g} ] =  f_g \delta_h [ f_{G/g} ] =  f_{g \ds \m J}  \delta_h [ f_{G / g} ].
 \end{equation*}
\end{enumerate}
In all three cases we argue like in the first step to carry out the operation of $\delta_h$ and conclude
\begin{equation*}
 \delta_{\m J}[f] = f_{g \ds \m J} f_{h \ds \m J} f_{G \ds \m J}.
\end{equation*}
For general $\m J \subseteq \m N$ we repeat this procedure for a finite number of steps to show 
\begin{equation*}
 \delta_{\m J}[f]= \prod_{\gamma \in \m J \cup \{G\}} f_{ \gamma \ds J}.
\end{equation*}

With the help of these two assertions we are now able to examine the integrals
\begin{equation}\label{rg int}
 \big\langle R^{\m N \setminus \m K}_{\nu}[ \mu_{\m K} \cdot (w^s)_{\m E_{\m K}} ] | \varphi \big\rangle = \int_{\kappa (U)} \!  dx u_{\m N}^s \mu_{\m K} \sum_{\m J \subseteq \m N \setminus \m K}(-1)^{|\m J|}  \nu_{\m J} \delta_{\m K \cup \m J}[f^s \varphi]
\end{equation}
in detail. Note that $f_{g\ds \m J}$ depends only on the variables $x_e$ associated to edges of $ E(t_g \setminus t_{h_1 \cup \cdots \cup h_k} )$ with $\{h_1,\ldots, h_k\}=\max \m J_{< g}$ (not on the marked element $x_g^{i_g}$ though!). This is exactly the set of coordinates on which the maps $\nu_g$ depend. Therefore divergences corresponding to elements $g \ds \m J \in \m N \ds \m J$ are also renormalized by the subtraction points $\nu_{g}$ associated to $g \in \m N$. 
To simplify notation, for $\m K\subseteq \m N$ write $\ti g$ for the $\m K$-contracted graph $g\ds \m K$. Let $\m K=\{g_1,\ldots, g_n\}$ (if $G\in \m K$ assume $g_n=G$) and define the subsets $\m H_i \subseteq \m N$ by 
\begin{equation*}
\m H_i :=\{  h \mid \ti h \sqsubset \ti g_i\}  \text{ for } i=1,\ldots, n.
\end{equation*}
We want to show that the integral 
\begin{equation*}
 \int_{\kappa (U)} dx \, u^s_\m N \prod_{i=1}^n (\nu_{g_i}'-\nu_{g_i})\sum_{\m J\subseteq \m N\setminus \m K} (-1)^{|\m J|} \nu_\m J  \prod_{\gamma \in \m K \cup \m J \cup \{G\} }f^s_{\ti \gamma \ds \m J} \ \delta_{\m K\cup \m J}[\varphi]
\end{equation*}
factorizes into a product of integrals according to \eqref{rg1} and \eqref{rg2}. To see this split the sum into two parts, the first one summing over subsets $\m I\subseteq \m N \setminus \m K$ that contain an element of $\m H_1$, i.e.\ $\m I\cap \m H_1\neq \emptyset$, the second one over subsets $\m J \subseteq \m N \setminus \m K$ with $\m J \cap \m H_1 = \emptyset$. The first sum can then be written as
\begin{equation*}
 \sum_{ \substack{ \m J\subseteq \m N \setminus \m K \\ \m J \cap \m H_1=\emptyset}} (-1)^{|\m J|} \sum_{\emptyset \neq \m L \subseteq \m H_1} (-1)^{|\m L|} \nu_\m J \nu_\m L \delta_{\m K\cup \m J \cup \m L}[f^s] \delta_{\m K \cup \m J \cup \m L}[\varphi].
\end{equation*}
We have $\delta_{\m K\cup \m J \cup \m L}[\varphi]=\delta_{\m K \cup \m J}[\varphi]$, because all $g \in \m L$ satisfy $\ti g \sqsubset \ti g_1$ and from this follows $g \leq g_1$ for the partial order on $\m N$. Therefore, all $g\in \m L$ are scaled by $x^{i_{g_1}}_{g_1}$ which is set to zero by $\delta_\m K$.
Again, since all elements of $\m L$ are smaller than $g_1$,
\begin{align*}
\delta_{\m K\cup \m J \cup \m L}[f^s]& = f_{\ti G \ds \m J\cup \m L } \prod_{\gamma \in (\m K \cup \m J \cup \m L) \setminus \{G\} } f_{\ti \gamma \ds \m J\cup \m L} \\
& = f_{\ti G \ds \m J} \prod_{\gamma \in \m K \setminus \{G\} }f_{\ti \gamma\ds \m J\cup \m L} \prod_{\xi \in \m J \setminus \{G\} }f_{\ti \xi \ds \m J\cup \m L} \prod_{ \eta \in \m L \setminus \{G\} }f_{\ti \eta \ds \m J\cup \m L}\\
& = f_{\ti G \ds \m J} \ f_{\ti g_1 \ds \m L} \prod_{\gamma \in \m K \setminus \{g_1,G\} }f_{\ti \gamma\ds \m J} \prod_{\xi \in \m J \setminus \{G\} }f_{\ti \xi\ds \m J} \prod_{\eta \in \m L \setminus \{G\} }f_{\ti \eta \ds \m L}  .
\end{align*}
In the last line we have used that $\ti g_1$ is immune to contraction by elements lying outside of $\m H_1$. 
Thus, the factor
\begin{align*}
 \sum_{\emptyset \neq \m L \subseteq \m H_1} (-1)^{|\m L|} \nu_\m L \ f_{\ti g_1 \ds \m L }\prod_{\eta \in \m L \setminus \{G\} }f^s_{\ti \eta\ds \m L}
\end{align*}
can be pulled out of the first sum. 
In the second sum over the subsets $\m J \subseteq \m N\setminus \m K$ with $\m J \cap \m H_1= \emptyset$ the factor $f^s_{\ti g_1}$ appears in every summand because $f^s_{\ti g_1}$ is not affected by $\delta_\m J$. Recall that $\delta_{\m K}[\varphi]$ depends only on the coordinates $\{x_e\}_{e\in E(t\ds \m K)}$ to conclude that 
 \begin{align*}
 & \int_{\kappa (U)} dx\, u^s_\m N \prod_{i=1}^n (\nu_{g_i}'-\nu_{g_i})\sum_{\m J\subseteq \m N \setminus \m K} (-1)^{|\m J|} \nu_\m J \  \delta_{\m K\cup \m J}[f^s] \ \delta_{\m K\cup \m J}[\varphi] \\
 = & \int_{V_1} dx \, u^s_{g_1} u^s_{\m H_1}  (\nu_{g_1}' - \nu_{g_1}) \left( \sum_{\emptyset \neq \m L \subseteq \m H_1} (-1)^{|\m L|} \nu_\m L \prod_{\eta \in \m L \setminus \{G\} }f^s_{\ti \eta \ds \m L} \ f^s_{\ti g_1 \ds  \m L} + f^s_{\ti g_1} \right) \\
 & \times \! \! \int_{V_2} dx \, u^s_{\m N \setminus (\m H_1\cup \{g_1\})} \prod_{i=2}^{n} (\nu_{g_i}'-\nu_{g_i}) \! \! \! \! \! \sum_{ \substack{ \m J\subseteq \m N \setminus \m K \\ \m J \cap \m H_1 = \emptyset } } (-1)^{|\m J|} \nu_\m J \  \delta_{\m K \setminus \{g_1\}\cup \m J}[f^s] \ \delta_{\m K \cup \m J}[\varphi] \\
= &\, \big\langle R^{ \m H_1 \cup \{g_1\}  }_{\nu}[\ti w_{g_1\ds \m K}] | \nu_{g_1}' \big\rangle \int \cdots.
 \end{align*}
Here we have changed the domain of integration from $\kappa(U)$ to $V_1 \times V_2$ with $V_i$ constructed as follows: Pick a linear extension of the partial order on $\m K=\{g_1,\ldots, g_n\}$ and let $g_1=\ti g_1$ be the minimal element (the proof works also without this assumption, but this simplifies it considerably). Define $X^{ \ti g_1 }$ as in \eqref{x tilde g} and $X^{G'}$ similarly for $G' := G / \ti g_1$. 
Recall the wonderful construction from Definition \ref{wonderful definition 2} and set for every $g$ in $\m B$
\begin{equation*}
Z^1_g := Z_g  \cap (X^{\ti g_1} \times \{0\} ) \quad\text{and}\quad  Z^2_g:= Z_g \cap ( \{0\} \times X^{G'} ).
\end{equation*}
Define $V_1 := X^{\ti g_1} \setminus \cup_{\gamma \in \m B} Z^1_{\gamma}$ and $V_2:=X^{G'} \setminus \cup_{\gamma \in \m B} Z^2_{\gamma}$. Then $V_1$ is a local chart domain for the wonderful model for $\ti g_1$ with respect to the nested set $\m H_1 \cup \{g_1\}$ and adapted basis $B|_{ E(t_{\ti g_1}) }$. The same holds for $G'$ with respect to $\m N \setminus (\m H_1 \cup \{g_1\})$ and $B|_{ E(t_{G / \ti g_1}) }$. 
For the original chart on $Y$ we have 
\begin{equation*}
 \kappa(U)=X^G \setminus (\cup_{\gamma \in \m B} Z_{\gamma} ) \subseteq V_1 \times V_2,
\end{equation*}
and the difference is an union of linear subspaces, i.e\ a set of measure zero. Moreover, the integrand is finite because all divergences associated to the elements of $\m K$ get ``damped'' by $\mu_\m K$, the remaining divergences coming from elements of $\m N \setminus \m K$ are renormalized and $\psi \in \m D(\beta(U_{\m N,B}))$ vanishes in a neighborhood of all $Z_{\gamma}$, which covers the divergences of $\m B \setminus \m N$. Thus, changing the domain is justified and by Fubini's theorem the integral factorizes into the desired product. The last equality holds because of
 \begin{align*}
&\quad\ (\nu_{g_1}' - \nu_{g_1}) \left( \sum_{\emptyset \neq \m L \subseteq \m H_1} (-1)^{|\m L|} \nu_\m L \prod_{\eta \in \m L }f^s_{\ti \eta \ds \m L} \ f^s_{\ti g_1 \ds  \m L} + f^s_{\ti g_1} \right)  \\
& = f^s_{\ti g_1}\nu_{g_1}' - \nu_{g_1} ( f^s_{\ti g_1} \nu_{g_1}')|_{x_{g_1}^{i_{g_1}}=0} + \sum_{\emptyset \neq \m L \subseteq \m H_1} (-1)^{|\m L|} \nu_\m L \prod_{\eta \in \m L \cup \{g_1\} }f^s_{\ti \eta \ds \m L} \nu_{g_1}' \\
& \quad - \sum_{\emptyset \neq \m L \subseteq \m H_1} (-1)^{|\m L|} \nu_\m L \nu_{g_1} \left( \prod_{\eta \in \m L \cup \{g_1\} }f^s_{\ti \eta \ds \m L} \nu_{g_1}'\right)|_{x_{g_1}^{i_{g_1}}=0} \\
& = f^s_{\ti g_1}\nu_{g_1}' + \sum_{\emptyset \neq \m L \subseteq \m H_1 \cup \{g_1\} } (-1)^{|\m L|} \nu_\m L \delta_\m L [ f^s_{\ti g_1} \nu_{g_1}']. 
 \end{align*}
 
 A technical detail: If $g_1$ had another divergent subgraph $h \in \m B \setminus \m N$, the renormalization by $R_\nu$ would not take care of this and the integral would still diverge. But in this case all variables $\{x_e\}_{e\in E(t_h)}$ are set to zero by $\delta_{\m K}$. Then the whole summand associated to $\m K$ in \eqref{rg1} vanishes because $\delta_{\m K}[\varphi]=\delta_{\m K}[\psi \circ \rho ]=0$ since supp$(\psi)\subseteq \kappa(U)$ is disjoint from $ \{ x_e=0 \mid e \in E(t_h)\} $. 
 
The remaining integral is of the same structure as the one we started with, so we can repeat the process for $g_2 \in \m K$ (notice how this relies heavily on $\m K\subseteq \m N$ being nested and the stability of $t$ under contractions). After a finite number of steps we obtain a product of renormalized densities, each factor representing an element of $\m K$, possibly times a last remaining factor,
 \begin{equation*}
  \int_{V_2} dz \, u^s_{\m N \setminus ( \cup_{i=1}^n \m H_i \cup \{g_i\})}\sum_{\substack{\m J\subseteq \m N \setminus \m K \\ \m J \cap (\cup_i \m H_i )=\emptyset}} (-1)^{|\m J|} \nu_\m J \ f^s_{\ti G\ds \m J} \prod_{\gamma \in \m J \setminus \{G\}} f^s_{\ti \gamma \ds \m J} \ \delta_{\m K\cup \m J}[\varphi].
  \end{equation*}
If $G$ lies in $\m K$, then $\delta_\m K[\varphi]=\psi(0)$ is constant and the procedure ends before this last step since the $\m H_i$ cover $\m N$. If $G$ is not in $\m K$, then $G \ds \m K$ could have remaining divergences, given by elements of $\m N \setminus (\m K \cup \m H_1 \cup \cdots \cup \m H_n)$. In this case the integral is $\langle R_{\nu}[\ti w^s_{G \ds \m K} ]|\delta_{\m K}[\varphi]\rangle$, the renormalized expression for $\ti w^s_{G \ds \m K}$, applied to the test function $\delta_\m K[\varphi]$, because   
\begin{align*}
 &\quad \sum_{\substack{\m J\subseteq \m N \setminus \m K \\ \m J \cap (\cup_i \m H_i )=\emptyset}} (-1)^{|\m J|} \nu_\m J \ f^s_{\ti G\ds \m J} \prod_{\gamma \in \m J \setminus \{G\}} f^s_{\ti \gamma \ds \m J} \ \delta_{\m K\cup \m J}[\varphi] \\
 & = \sum_{\m J \subseteq \m N \setminus ( \m K \cup \m H_1 \cup \cdots \cup \m H_n  )} (-1)^{|\m J|} \nu_\m J \ \delta_{\m J} [ f^s_{ G \ds \m K}  \delta_\m K[\varphi] ].
\end{align*}

Last but not least, we need to take care of the exponents in $u_g^s$ not matching the ones provided by the definition of $\ti w^s_{G/ \m K}$ and $\ti w^s_{g\ds \m K}$. This would not happen if we had defined subgraphs $g\subseteq G$ as given by there edge set $E(g)\subseteq E(G)$ but with $V(g)=V(G)$ (we chose not to do so because in the formulation presented here, $X$ is spanned by variables associated to edges of an adapted spanning tree, not by the elements of $V'$ like in \cite{bbk}). However, we can also just rescale the complex regularization parameter $s= 1- \frac{d_{\ti g}}{d_g}(1-\ti s)$ without affecting the whole construction to obtain the correct exponent in $u_{\ti g}^s$. On the other hand, this discrepancy does not show up in the limit $s \to 1$ which we are allowed to take because this proof shows that every term in~\eqref{rg1} and \eqref{rg2} is well-defined at $s=1$. Putting everything together we arrive at the desired formula.
\end{proof}

The case of (local) minimal subtraction works in the same manner as above. This is already clear if we think of $R_1$ as a non-smooth version of $R_{\nu}$ by making the (forbidden) substitution $\nu_g(x_g)= \theta (1-|x_g^{i_g}|)$. 
Let $\m N$ be a $\m B$-nested set. Locally the minimal subtraction operator $R_1$ is given by
\begin{equation*}
\ti w^s(x) =  \prod_{\gamma \in \m N} u^s_{\gamma}(x^{i_{\gamma}}_{\gamma}) f^s(x)|dx| \overset{R_1}{\longmapsto}  \prod_{\gamma \in \m N} \big(u^s_{\gamma}(x^{i_{\gamma}}_{\gamma} )\big)_{\heartsuit}f^s(x) |dx|,
\end{equation*}
where $(\cdots)_{\heartsuit}$ denotes the regular part
\begin{equation*}
 \langle \big(u^s_g \big)_{\heartsuit} | \varphi \rangle := \int dx \, |x|^{-1 +d_g(1- s)} \big(\varphi(x) - \theta(1-|x|) \varphi(0) \big)
\end{equation*}
and the factor $f^s$ is treated as test function. If instead $\theta^c(x):=\theta( c - |x|)$ with $c>0$ is used as cutoff, the principal part of the Laurent expansion does not change while the regular part $(u^s_g)_{\heartsuit_c}$ gets an additional $\delta$-term,
\begin{equation*}
2 \sum_{k\geq 0}\frac{ d_g^k \log^{k+1}(c) }{k+1!} (s-1)^k \delta. 
\end{equation*}
Therefore the whole Laurent series for $u_g^s$ is given by
\begin{align*}
  \langle u^s_g | \varphi \rangle & =  \frac{2}{d_g} \langle \delta | \varphi \rangle (1-s)^{-1} + \sum_{k\geq0}  \langle \frac{1}{k!} \big( |x|^{-1} \log^k(|x|) \big)_{\heartsuit_c} | \varphi \rangle (1-s)^k \\
 & =\frac{2}{d_g} \langle \delta | \varphi \rangle (1-s)^{-1} + \sum_{k\geq0}  \langle \frac{1}{k!} \big( |x|^{-1} \log^k(|x|) \big)_{\heartsuit} | \varphi \rangle (1-s)^k\\
 & \quad  + \sum_{k\geq0} \frac{1}{k+1!} d^{k}_g\log^{k+1}(c)\langle \delta | \varphi \rangle  (1-s)^k.
\end{align*}
Write $\m N=\{g_1 , \ldots, g_n \}$ and for $k \in \set{n}$ let $ x_k^{i_k}$ denote the associated marked element. Write $\hat x$ for the collection of all other coordinates. For a test function $\varphi \in \m D(\kappa (U) )$ set 
\begin{equation*}
 \phi_s(x_1^{i_1}, \ldots, x_n^{i_n}):= \int d\hat x \, f^s (x) \varphi(x). 
\end{equation*}
By expanding the successive application of the regular parts $(u^s_g)_{\heartsuit}$ and reordering the sum we see that $R_1$ is expressed by a formula similar to the one for subtraction at fixed conditions:
\begin{align*}
 \big\langle  R^\m N_1[\tilde w^s] | \varphi \big\rangle & = \big\langle (u^s_{g_1})_{\heartsuit}| \langle (u^s_{g_2})_{\heartsuit} |\cdots \langle (u^s_{g_n})_{\heartsuit} | \phi_s \rangle \cdots \rangle \big\rangle \\
& = \big\langle (u^s_{g_1})_{\heartsuit} | \langle (u^s_{g_2})_{\heartsuit} | \cdots \int dx_n^{i_n}  |x_n^{i_n}|^{- 1 + d_n (s-1)}\\
&\quad \times  \big(  \phi_s(x_1^{i_1}, \ldots, x_n^{i_n}) -  \theta^1(x_n^{i_n}) \phi_s(x_1^{i_1}, \ldots, x_{n-1}^{i_{n-1}}, 0)  \big) \rangle \cdots\big\rangle \\
& = \sum_{\m K \subseteq \m N} (-1)^{|\m K|} \int dx \,  \prod_{j=1}^n u^s_{g_j}(x_j^{i_j})  \prod_{\gamma \in \m K} \theta_{\gamma} (x_{\gamma}^{i_{\gamma}}) \, \delta_{\m K} [f^s \varphi](x).
\end{align*}
We do the same for the regular parts obtained by cutting off at $c$,
\begin{align*}
 \big\langle  R^\m N_{c}[\tilde w^s] | \varphi \big\rangle = & \sum_{\m K \subseteq \m N} (-1)^{|\m K|} \int dx \, \prod_{j=1}^n u^s_{g_j}(x_j^{i_j})  \prod_{\gamma \in \m K} \theta^c_{\gamma} (x_{\gamma}^{i_{\gamma}}) \, \delta_{\m K} [f^s \varphi](x),
\end{align*}
and write $\theta^{c'}(x)= \theta^c(x) + \vartheta^{c,c'}(x)$ with
\begin{equation*}
 \vartheta^{c,c'}(x)=\begin{cases}
                      1 & \text{ for } c<|x|< c', \\
                      0 & \text{ else.}
                     \end{cases}
\end{equation*}
Then we can express the difference between two minimal subtraction operators $ R_{c'}^{\m N}$ and $R_{c}^{\m N} $, applied to $\tilde w^s$, by the density
\begin{equation*}
 \varphi \longmapsto \sum_{\emptyset \neq \m K \subset \m N} (-1)^{|\m K|} \int dx \,  u^s_\m N   \Theta_{\m K}^{c',c} \delta_\m K[f^s \varphi].
\end{equation*}
Here $\Theta_{\m K}^{c',c}:= \prod_{\gamma \in \m K} \theta^{c'}(x_{\gamma}^{i_{\gamma}}) - \prod_{\gamma \in \m K} \theta^{c}(x_{\gamma}^{i_{\gamma}})$ is a ``multidimensional cutoff'', supported on 
\begin{equation*}
\big\{ x \in \mathbb{R}^{|\m K|} \mid  c<|x_i|<c' \text{ for all } i \in \{1, \ldots, |\m K|\} \big\}. 
\end{equation*}
Expanding $\Theta^{c,c'}_{\m K}$ and reordering the sum, we have
\begin{align*}
\Theta_{\m K}^{c',c} &= \prod_{\gamma \in \m K} ( \theta^c_{\gamma}  + \vartheta^{c,c'}_{\gamma} ) - \prod_{\gamma \in \m K} \theta^{c}_{\gamma} \\
&= \sum_{\emptyset \neq \m J\subseteq \m K}   \left(\prod_{ \gamma \in \m J} \vartheta^{c,c'}_{\gamma} \right) \left( \prod_{\eta \in \m K \setminus \m J}\theta^c_{\eta} \right) = \sum_{ \emptyset \neq \m J\subseteq \m K } \!\!\! \vartheta^{c,c'}_{ \m J} \theta^c_{\m K \setminus \m J}. 
\end{align*}
Putting everything together, we find that a change of the renormalization point $c$ is expressed by a sum of densities supported on components of the exceptional divisor given by subsets $\m K \subseteq \m N$,
\begin{align*}
 \big\langle ( R^\m N_{c'} - R^\m N_{c} )[\tilde w^s] | \varphi \big\rangle = \sum_{\emptyset \neq \m K \subseteq \m N } (-1)^{|\m K|} \sum_{\emptyset \neq \m J\subseteq \m K}  \int dx \,  u_{\m N}^s \vartheta^{c,c'}_{\m J} \theta^c_{\m K \setminus \m J}   \delta_{\m K}[f^s \varphi] .
\end{align*}

We can repeat the argumentation from the case of subtraction at fixed conditions to arrive at the formulae of Proposition \ref{rg prop} and Theorem \ref{rg thm}:
\begin{align*}
  \big\langle ( R^{\m N}_{c'} - R^{\m N}_c )[\tilde w^s] | \varphi \big\rangle &=  \sum_{\emptyset \neq \m K \subseteq \m N } (-1)^{|\m K|} \big\langle R_{c}^{\m N \setminus \m K}[ \vartheta^{c,c'}_{\m K} \cdot (\ti w^s)_{\m E_{\m K}} ]|\varphi \big\rangle, \\
 \big \langle ( R^{\m N}_{c'} - R^{\m N}_c )[\tilde w^s] | \varphi \big\rangle& =  \sum_{\emptyset \neq \m K \subseteq \m N } (-1)^{|\m K|} c_{\m K} \big\langle R_c[\ti w^s_{G/ \m K}] | \delta_{\m K} [\varphi] \big\rangle.
 \end{align*}
Viewing the maps $\theta^{c'}_{\gamma}$ as test functions, the constants $c_{\m K}$ are exactly the same as in \eqref{rg2}. They are given by densities of $\m K$-contracted graphs, evaluated at their respective renormalization points, $c_{\gamma}= \prod_{\gamma \in \m K} \langle R_c[\ti w^s_{\gamma \ds \m K} ] | \theta^{c'}_{\gamma} \rangle$. 

Eventually one would like to apply the formulae presented here not on distributions given by single graphs but on the formal sum of all graphs expressing a given interaction (an \textit{amplitude}). The study of the behaviour of amplitudes under a change of renormalization points allows in best cases even to derive statements beyond perturbation theory. The main idea is that physical observables do not depend on the choices made in fixing a renormalization scheme. This leads to a differential equation, the \textit{renormalization group equation} (cf.\ \cite{co}), or in the language of the renormalization Hopf algebra, to (combinatorial) \textit{Dyson-Schwinger equations} (cf.\ \cite{dk}). 
So far we have not used any differential methods, but to explore these objects within the wonderful framework it seems that subtraction at fixed conditions is then the way to go. This is reserved for future work.

\section{Back to physics}\label{eg2}

This section connects the geometric method of extending distributions presented here to physics. We show that wonderful renormalization satisfies the \textit{locality principle} of Epstein and Glaser \cite{eg}. After that we finish with an outlook of how to relate our approach to the method of Epstein-Glaser, i.e.\ to the renormalization of amplitudes, and how Hopf algebras can be utilized to describe the wonderful renormalization process. 

\subsection{Connection to the Epstein-Glaser method}
The Epstein-Glaser locality principle is the position space analogue of locality of counterterms. It decides whether a given theory is renormalizable, i.e.\ if adding counterterms to renormalize the Lagrangian keep its form invariant. 
In \cite{bbk} this principle is formulated in a version for single graphs.

\begin{defn}[Locality principle]
Let $G$ be a connected graph. Let $\mathscr R$ denote a renormalization operator. $\mathscr R$ satisfies the \textit{locality principle} of Epstein and Glaser if 
\begin{equation}\label{loc princ}
 \mathscr R [v_G]=\mathscr R[v_g] \mathscr R[v_h] v_{G\setminus (g\cup h)} \text{ on } X^G \setminus X_s^{G \setminus (g\cup h)}
\end{equation}
holds for all disjoint pairs $g,h$ of connected and divergent subgraphs of $G$.
\end{defn}

This is to be understood in the sense of distributions. For all test functions $\varphi \in \m D(X^G)$ with support disjoint from $X_s^{G \setminus (g \cup h)}$ the renormalization of $v_G$ is already determined by the renormalized distributions $v_g$ and $v_h$. Note that $v_g$ and $v_h$ depend on disjoint sets of variables and, outside of $X_s^{G \setminus (g \cup h)}$, $v_{G\setminus (g\cup h)}$ is regular in the coordinates associated to $G \setminus (g \cup h)$. Therefore, the product on the right hand side of \eqref{loc princ} is well-defined.
Another way to formulate the locality principle is that for causal disconnected regions ($g,h$ disjoint) the renormalized distribution is given by ``lower order'' (i.e.\ subgraph-)distributions. In the Epstein-Glaser method this is one of the main ingredients in the construction; it allows to recursively construct the $n$-th order term $T^n$ in the formal series for the $S$-matrix up to the small diagonal in $M^n$.

\begin{thm}
 Let $\mathscr R$ be given by minimal subtraction or subtraction at fixed conditions on the minimal wonderful model for the divergent arrangement of a connected and at most logarithmic graph. 
 Then $\mathscr R $ satisfies the locality principle \eqref{loc princ}.
\end{thm}

\begin{proof}
We follow the lines of \cite{bbk} but correct the proof by adding some essential details missing there.

Let $Y_g$, $Y_h$ and $Y$ denote minimal wonderful models for $g$, $h$ and $G$. Let $t$ be $\m D(G)$-adapted and set $X^-:=X^{ G^- }=\{ \{x_e\} \mid e \in E(t \cap G^-) \}$ for $G^-:=G\setminus (g \cup h)$. In the language of wonderful models the theorem states that $Y':=Y_g \times Y_h \times X^-$ is a (minimal) wonderful model for the divergent arrangement of the graph $g\cup h$ in $S:= X \setminus X_s^{G\setminus (g \cup h)}$. 
Let $\m B$ denote the minimal building set in the divergent arrangement. The proof is based on two claims: Every $\m B(g\cup h)$-nested set is given by a disjoint union $\m N_g \cupdot \m N_h$ of $\m B(g)$- and $\m B(h)$-nested sets (one of them possibly empty). Secondly, $\beta^{-1}(S) \subseteq Y$ is covered by the open sets $U_{\m N,B}$ with $\m N=\m N_g \cupdot \m N_h$ as above and $B$ marked accordingly.

Proof of first claim: Since $g$ and $h$ are disjoint, Lemma \ref{irr} implies $\m B(g\cup h)=\m B(g)  \cup \m B(h)$. This shows that if $\m N_g$ and $\m N_h$ are nested with respect to $\m B(g)$ and $\m B(h)$, then $\m N_g \cup \m N_h$ is $\m B( g\cup h )$-nested. On the other hand, every subset of a nested set is nested itself. With $\m B(g\cup h)=\m B(g)  \cup \m B(h) $ the claim follows.

Proof of second claim: If $\gamma$ is an element of $\m B(G) \setminus \m B(g\cup h)$, then it must contain an edge $e$ in $E(t_{\gamma} \setminus t_g)$. From
\begin{equation*}
 A_{\gamma}^{\bot} = \bigcap_{e' \in E(t_\gamma)}A_{e'}^{\bot} \subseteq A_e^{\bot}
\end{equation*}
and $e \in E(G^-)$ it follows that $\m E_{\gamma} \cap \beta^{-1}(S) = \emptyset$. 

Now let $x \in S$ and set $y:= \beta^{-1}(x)$. We are looking for $\m N=\m N_g \cupdot \m N_h$ and a marking of $B$ such that $y \in U_{\m N,B}$. Consider the $\m B(G)$-nested set $\m N:= \{\gamma,\eta\}$ where $\gamma \in \m B(g)$ and $\eta \in \m B(h)$ (one of them possibly the empty graph $o$). Let $B$ be marked accordingly and let $\underline{\hat x}_\gamma$ and $\underline{\hat x}_\eta$ denote the collection of coordinates $\{x_e\}_{e\in E(t_{\gamma} )}$ and $\{x_e \}_{e\in E(t_{\eta} )}$ where the marked elements $x_{\gamma}^{i_{\gamma} }$ and $x_{\eta}^{i_{\eta} }$ are set to $1$. The map $\rho_{\m N,B}$ scales $\underline{\hat x}_\gamma$ by $x_{\gamma}^{i_{\gamma} }$, $\underline{\hat x}_\eta$ by $x_{\eta}^{i_{\eta} }$ and leaves all other coordinates unaltered - it does not ``mix'' coordinates because $g$ and $h$ are disjoint. Recall that $\rho_{\m N,B}$ is the essential part in the definition of the chart 
\begin{align*}
 \kappa_{\m N,B}^{-1}: X \setminus \bigcup_{\xi \in \m B(G)} Z_{\xi} \longrightarrow U_{\m N,B}, \quad x \longmapsto \big(  x_{\gamma}^{i_\gamma} \underline{\hat x}_{\gamma}, \ldots, x_{\eta}^{i_{\eta}} \underline{\hat x}_{\eta}; [\underline{\hat x}_{\gamma}] , \ldots, [\underline{\hat x}_{\eta}] \big).
\end{align*}
If $\xi$ is in $\m B(G) \setminus \m B(g \cup h)$, the $Z_\xi$ are given by $\{x_e=0 \mid e \in E(t_{\xi})\}$ which is a subset of $X \setminus S$. Similarly, for $\xi \in \m B(g\cup h)$ with either $\xi < \gamma$ or $\xi < \eta$ we have $Z_\xi = \{x_e=0 \mid e \in E(t_{\xi})\}$. If $\xi \in \m B(g \cup h)$ and either $\xi > \gamma$ or $\xi> \eta$, then $Z_\xi = \{x_{\gamma}^{i_\gamma}=0 , x_e = 0 \mid e \in E(t_\xi \setminus t_\gamma )\}$ or with $\gamma$ replaced by $\eta$.  Finally, $Z_\gamma = Z_\eta = \emptyset$. From this description it is clear how to find $U_{\m N,B}$ containing $y$: Pick an appropriate pair $\gamma, \eta$ and a marking of $B$ such that $x$ does not lie in one of the $Z$'s, then find the local preimage of $x$ under $\rho_{\m N,B}$ by solving a trivial system of linear equations.

The previous discussion also shows that for nested sets $\m N = \m N_g \cupdot \m N_h$ and $B$ marked accordingly the map $\rho_{\m N,B}$ restricts to the identity on $X^-$, $\rho_{\m N, B} = \rho_{\m N_g,B_g} \times \rho_{\m N_h,B_h} \times \text{id}$. Moreover, we have a diffeomorphism
\begin{align*}
U_{\m N,B} \cong U_{\m N_g, B_g} \times U_{\m N_h,B_h} \times X^- \setminus \bigcup_{\gamma \in \m B(G) \setminus \m B(g \cup h)} Z_\gamma.
\end{align*}
A similar decomposition holds for the charts $\kappa_{\m N,B} = \kappa_{\m N_g,B_g} \times \kappa_{\m N_h,B_h} \times \text{id}$. Thus, on $\beta^{-1}(S)$ both models $Y$ and $Y'$ locally look the same and we can find a partition of unity
$\chi_{\m N, B}= \chi_{\m N_g,B_g}\otimes \chi_{\m N_h,B_h} \otimes 1$,
subordinate to the sets $U_{\m N, B} \cap \beta^{-1}(S)$. With the notation introduced in the proof of Theorem \ref{rg thm} we have $f=f_g f_h f_{G^-}$, so that in every such chart
\begin{align*}
  \langle R^{\m N,B}_{\nu}[\ti w_G] | \psi \rangle & = \sum_{\m K \subseteq \m N} (-1)^{|\m K|} \langle \nu_\m K \cdot (\ti w_G)_{\m E_\m K} | \psi \rangle \\
 & = \sum_{\m K_g \cupdot \m K_h \subseteq \m N} (-1)^{|\m K_g|+|\m K_h|} \int dx \, u_{\m N} \nu_{\m K_g} \nu_{\m K_h} \delta_{\m K_g \cup \m K_h} [f \psi ] \\
 & = \int dx  \sum_{\m K_g \subseteq \m N_g} (-1)^{|\m K_g|} u_{\m N} \nu_{\m K_g} \delta_{\m K_g}[f_g]    \\
& \quad \times  \sum_{\m K_h \subseteq \m N_h} (-1)^{|\m K_h|}\nu_{\m K_h} \delta_{\m K_h}[f_h] \delta_{\m K_g \cupdot \m K_h}[f_{G^-}] \delta_{\m K_g\cupdot \m K_h} [\psi] \\
& = \big\langle R^{\m N_g,B_g}_{\nu}[\ti w_g] \otimes R^{\m N_h,B_h}_{\nu}[\ti w_h] \mid \langle \ti w_{G^-} | \psi \rangle \big\rangle.
 \end{align*}
Applying this to $\psi =  ( \chi_{\m N,B} \circ \kappa_{\m N,B}^{-1}) \rho_{\m N,B}^* \varphi$ and summing over all nested sets and corresponding markings shows \eqref{loc princ}. The case of minimal subtraction works in the same way (cf.\ the discussion in Section \ref{renormalization group}). This finishes the proof. 
\end{proof}

 To connect the graph by graph method presented in this thesis with the Epstein-Glaser construction we need to renormalize the sum of all graphs with a fixed vertex order. Thus, we need a space that serves as a universal wonderful model for all at most logarithmic graphs on $n$ vertices. There are two obvious candidates, the minimal and the maximal wonderful models of the graph lattice $\m G$ for the complete graph $K_n$. Since every divergent subgraph of a graph is saturated, the set $\m G(K_n)$ contains all possible divergences of such a graph. In other words, these two models are universal in the sense that for every graph $G$ on $n$ vertices there exist canonical proper projections 
 \begin{align}\label{universal}
 p^G_{\text{max}}:Y_{\m G(K_n)} & \longrightarrow Y_{\m D(G)}, \\
 p^G_{\text{min}} : Y_{I(\m G(K_n))} & \longrightarrow Y_{I(\m D(G))}.                                                                                                                                                                                                                                                                                                                                                                                                                                                                                                                                                                                                                                                                                                                                                                                                                                                                                                                                                                                                                                                                                                                                                                                                                                                                     \end{align}
This follows from Definition \ref{wonderful definition 1}.
The theorem above suggests to focus on minimal building sets. Let $G$ be a connected and at most logarithmic graph on $n$ vertices. The idea is to compose the projection $p^G_{\text{min} }$ with the blowdown $\beta$ of the wonderful model $Y_{I(\m D(G))}$ and consider the pullback $\ti w_G$ of $v_G$ under this map, then proceed as before to obtain a renormalized density on $Y_{I(\m G(K_n))}$. A detailed description is left for future work, but we make one further observation that highlights the connection between wonderful renormalization and the Epstein-Glaser method.
Recall that the wonderful model $Y_{I(\m G(K_n))}$ is equivalent to the Fulton-MacPherson compactification of the configuration space $F_n(M)$, for which the structure of $I(\m G(K_n))$-nested sets is encoded by rooted trees \cite{fm}. As shown in \cite{bergkrei}, Epstein-Glaser renormalization can also be formulated in terms of rooted trees. On the other hand, the Hopf algebra of rooted trees $H_{\text{rt}}$ satisfies an universal property in the category of renormalization Hopf algebras \cite{dk}, as does the Fulton-MacPherson compactification in the category of (minimal) wonderful models (cf.\ Equation \eqref{universal})!

\subsection{Connection to renormalization Hopf algebras}

As shown in \cite{bbk}, the renormalization Hopf algebra of Feynman graphs is encoded in the stratification of the exceptional divisor $\m E$ of a wonderful model associated to a graph $G$. We sketch the arguments and finish with a discussion of a Hopf algebraic formulation of wonderful renormalization.

Let $H$ be the free algebra on the vector space spanned by (isomorphism classes) of connected, divergent (at most logarithmic) graphs. The multiplication on $H$ is given by disjoint union, the empty graph being the unit element. In \cite{bk} it is shown that $H$ endowed with a coproduct $\Delta$ given by 
\begin{equation*}
 \Delta(G):=\sum_{\gamma \in \m D} \gamma \otimes G\ds \gamma
\end{equation*}
is indeed a Hopf algebra. To cope with the case of minimal building sets, i.e.\ irreducible graphs, we can mod out by the ideal $I$ generated by all irreducible decompositions as defined in Section \ref{comps2}.
On $\ti H:=H/I$ it is the antipode $S: \ti H \to \ti H$ that disassembles $G$ into parts determined by its irreducible divergent subgraphs and prepares so the renormalization process. In terms of contraction relative to nested sets $S$ is given by 
\begin{equation*}
 S(G)= \sum_{\m N} (-1)^{|\m N|} \prod_{\gamma \in \m N} \gamma \ds \m N
\end{equation*}
where the sum is over $I(\m D(G))$-nested sets containing $G$.
This is the starting point of Hopf algebraic renormalization. The goal is then to formulate the whole wonderful renormalization process in terms of the convolution product of a twisted antipode with Feynman rules, similar to renormalization in momentum or parametric space. This is not straightforward due to the local formulation of the renormalization operators, but motivated by another, more direct approach: The combinatorial character of Zimmermann's forest formula is a first hint at a Hopf algebra structure underlying renormalization. Our locally defined wonderful renormalization operators resemble the classical formula for subtracting divergences only in certain charts. 
To connect with the forest formula and translate it into Hopf algebraic terms we could use the idea that if a graph has only subdivergences nested into each other, then local subtractions resemble the forest formula correctly. Working modulo primitive elements of $H$, or $\ti H$, every graph can be written as a sum of graphs that behaves like an element with purely nested subdivergences \cite{bk}. This shows that in principle wonderful renormalization fits into the Hopf algebraic framework. 
Of course, it is worthwhile to establish the connection on a more abstract level using geometrical methods. Once this is achieved, the whole world of renormalization Hopf algebras can be explored and used in the position space setting.

\bibliographystyle{alpha}
\bibliography{ref}

\end{document}